\documentclass[letterpaper, 10 pt, conference]{IEEEtran}
\IEEEoverridecommandlockouts
\let\labelindent\relax
\usepackage{enumitem}
\usepackage{balance}
\usepackage[font={small}]{caption}
\usepackage{subcaption}
\usepackage{array}
\usepackage{textcomp}
\usepackage{mathtools, nccmath}
\usepackage{graphicx}
\usepackage{amsfonts}
\usepackage{amsmath}
\usepackage{amssymb}
\usepackage{algorithm}
\usepackage{algorithmic}
\usepackage{hyperref}
\usepackage{tikz}
\usepackage{arydshln}
\usepackage{multirow}
\usepackage{bm}
\usepackage{epstopdf}
\usepackage{cite}
\usepackage{siunitx}
\usepackage{wrapfig}

\usepackage{amsthm}

\newtheorem{theorem}{Theorem}
\newtheorem{lemma}{Lemma}

\newtheorem{remark}{Remark}
\theoremstyle{definition}
\newtheorem{definition}{Definition}
\newtheorem{problem}{Problem}

\title{\LARGE \bf
Auxiliary-Variable Adaptive Control Barrier Functions}

\author{Shuo Liu$^{1}$, Wei Xiao$^{2}$ and Calin Belta$^{3}$% <-this % stops a space
\thanks{This work was supported in part by the NSF under grant IIS-2024606 at Boston University.}
\thanks{$^{1}$S. Liu is with the Department of Mechanical Engineering, Boston
University, Brookline, MA, 02215, USA. 
        {\tt\small \{liushuo\}@bu.edu}}%
\thanks{$^{2}$W. Xiao is with the Computer Science and Artificial Intelligence Lab, Massachusetts Institute of Technology, Cambridge, MA, USA 
        {\tt\small weixy@mit.edu}}%
\thanks{$^{3}$C. Belta is with the Department of Electrical and Computer Engineering and the Department of Computer Science, University of Maryland, College Park, MD, USA  
        {\tt\small calin@umd.edu}}
}

\begin{document} 
\maketitle
%{\color{blue} calin' color}
%%%%%%%%%%%%%%%%%%%%%%%%%%%%%%%%%%%%%%%%%%%%%%%%%%%%%%%%%%%%%%%%%%%%%%%%%%%%%%%%

\begin{abstract}
This paper addresses the challenge of ensuring safety and feasibility in control systems using Control Barrier Functions (CBFs). Existing CBF-based Quadratic Programs (CBF-QPs) often encounter feasibility issues due to mixed relative degree constraints, input nullification problems, and the presence of tight or time-varying control bounds, which can lead to infeasible solutions and compromised safety. To address these challenges, we propose Auxiliary-Variable Adaptive Control Barrier Functions (AVCBFs), a novel framework that introduces auxiliary variables in auxiliary functions to dynamically adjust CBF constraints without the need of excessive additional constraints. The AVCBF method ensures that all components of the control input explicitly appear in the desired-order safety constraint, thereby improving feasibility while maintaining safety guarantees. Additionally, we introduce an automatic tuning method that iteratively adjusts AVCBF hyperparameters to ensure feasibility and safety with less conservatism. We demonstrate the effectiveness of the proposed approach in adaptive cruise control and obstacle avoidance scenarios, showing that AVCBFs outperform existing CBF methods by reducing infeasibility and enhancing adaptive safety control under tight or time-varying control bounds.
\end{abstract}


\section{Introduction}
\label{sec:Introduction}
Safety is the primary concern in
the design and operation of autonomous systems. Many existing works enforce safety as constraints in optimal control problems using Barrier Functions (BF) and Control Barrier Functions (CBF). BFs are Lyapunov-like functions \cite{tee2009barrier} whose use can be traced back to optimization problems \cite{boyd2004convex}. They have been utilized to prove set invariance \cite{aubin2011viability}, \cite{prajna2007framework} to control multi-robot systems
\cite{glotfelter2017nonsmooth},\cite{lindemann2020barrier}, manipulators \cite{yu2024efficient}, unmanned aerial vehicles \cite{tayal2024control}, legged robots \cite{peng2024real} and soft robots \cite{patterson2024safe}. CBFs are extensions of BFs for control systems.  It has been shown that stabilizing an affine control system to admissible states, while minimizing a quadratic cost subject to state and control constraints, can be mapped to a sequence of Quadratic Programs (QPs) \cite{ames2016control} by unifying CBFs and Control Lyapunov Functions (CLFs) \cite{ames2012control}. In its original form, this approach, which in this paper we will refer to as CBF-CLF, works only for safety constraints with relative degree one. Exponential CBFs \cite{nguyen2016exponential} were introduced to accommodate higher relative degrees. A more general form of exponential CBFs, called High-Order CBFs (HOCBFs), has been proposed in \cite{xiao2021high}.  However, the aforementioned CBF-CLF-QP cannot always ensure the safety of an affine control system due to several unresolved challenges.

%CBFs are extensions of BFs used to enforce safety, i.e., rendering a set forward invariant, for an affine control system. It was proved in \cite{ames2016control} that if a CBF for a safe set satisfies Lyapunov-like conditions, then this set is forward invariant and safety is guaranteed. It has also been shown that stabilizing an affine control system to admissible states, while minimizing a quadratic cost subject to state and control constraints, can be mapped to a sequence of Quadratic Programs (QPs) \cite{ames2016control} by unifying CBFs and Control Lyapunov Functions (CLFs) \cite{ames2012control}. In its original form, this approach, which in this paper we will refer to as CBF-CLF, works only for safety constraints with relative degree one. Exponential CBFs \cite{nguyen2016exponential} were introduced to accommodate higher relative degrees. A more general form of exponential CBFs, called High-Order CBFs (HOCBFs), has been proposed in \cite{xiao2021high}. The CBF-CLF method has been widely used to enforce safety in many applications, including rehabilitative system control \cite{isaly2020zeroing}, adaptive cruise control \cite{ames2016control}, humanoid robot walking \cite{khazoom2022humanoid} and robot swarming \cite{cavorsi2022multi}. However, the aforementioned CBF-CLF-QP cannot always ensure the safety of an affine control system due to several unresolved challenges.

The first challenge is that some control inputs may not appear at certain states (input nullification, as discussed in \cite{lindemann2018control}) or across all states (when the CBF constraint has a mixed relative degree in systems with multiple control inputs, as discussed in \cite{xiao2022control}) in a CBF constraint. Consequently, these missing control inputs fail to effectively influence the system state, leading to the CBF's performance in ensuring system safety being substantially compromised. The second challenge arises in CBF-CLF-QP because the CBF constraint is enforced only at discrete sampling times. Between these sampling intervals, the system may evolve dynamically and potentially violate the safety constraint, as the control input cannot continuously adjust to maintain safety. This phenomenon, referred to as the inter-sampling effect as discussed in \cite{singletary2020control,breeden2021control,yang2019self}, compromises the system's overall safety. The third challenge is that the aforementioned CBF-based QP can become infeasible, particularly when tight or time-varying control bounds are present, as conflicts may arise between the CBF constraints and the control bounds.

There are several approaches that aim to enhance the feasibility of the CBF method while guaranteeing safety. One can formulate CBFs as constraints under a Nonlinear Model Predictive Control (NMPC) framework, which allows the controller to predict future state information up to a horizon larger than one. This leads to a less aggressive control strategy \cite{zeng2021enhancing}. However, the corresponding
optimization is overall nonlinear and non-convex, which could be computationally expensive for nonlinear systems. A convex MPC with linearized, discrete-time CBFs under an iterative approach was proposed in \cite{liu2023iterative,liu2024iterative,liu2024safety} to address the above challenges, but this comes at the price of losing safety guarantees. The works in  \cite{gurriet2018online,singletary2019online,gurriet2020scalable,chen2021backup} recently developed approaches in which a known backup set or backup policy is defined that can be used to extend the safe set to a larger viable set
to enhance the feasible space of the system in a finite time horizon
under input constraints. This backup approach has further been generalized to infinite time horizons \cite{squires2018constructive} \cite{breeden2021high}. One limitation of these approaches is that they require prior knowledge of finding appropriate backup sets, policy or nominal control law, which may be difficult to be predefined. Another limitation is that they only focus on feasibility, which may introduce over-aggressive or over-conservative control strategies. Sufficient conditions have been proposed in \cite{xiao2022sufficient} to guarantee the feasibility of the CBF-based QP, but they may be hard to find for general constrained control problems. All these approaches only consider time-invariant control limitations.

In order to account for time-varying control bounds, adaptive CBFs (aCBFs) \cite{xiao2021adaptive} have been proposed by introducing penalty functions in HOCBFs constraints, which provide flexible and adaptive control strategies over time. However, this approach requires extensive hyperparameter tuning, especially when the inter-sampling effect cannot be ignored. Moreover, this approach, along with the previously mentioned methods, cannot handle safety constraints with mixed relative degrees. To address feasibility issues while ensuring safety, this article proposes a novel approach to safety-critical control problems, introducing a new type of aCBFs called Auxiliary-Variable Adaptive Control Barrier Functions (AVCBFs). Compared to HOCBFs and original aCBFs, AVCBFs are designed to more effectively handle safety constraints that involve high-order or mixed relative degrees. Specifically, the contributions of this paper are as follows:

\begin{itemize}
\item We propose Auxiliary-Variable Adaptive CBFs (AVCBFs), which can be applied to the design of safety constraints regardless of whether the constraints have any relative degree or mixed relative degrees.

\item We show that AVCBFs can improve the feasibility of the CBF method under tight and time-varying control bounds. The proposed AVCBFs maintain a structure analogous to existing CBF methods, avoiding the need for excessive additional constraints. Furthermore, AVCBFs preserve the adaptive properties of aCBFs \cite{xiao2021adaptive}, while ensuring non-overshooting control policies near the boundaries of safe sets.  

\item We propose a novel parametrization method for AVCBFs that automatically and adaptively tunes their hyperparameters to ensure both safety and feasibility. This method is activated only when the safety-feasibility criterion is detected to be unsatisfied within a future horizon. By leveraging iterative hyperparameter tuning and safety-feasibility criterion evaluations, the parametrization method guarantees safety and feasibility even in scenarios where manually tuned AVCBF fails to maintain these properties.

\item We demonstrate the effectiveness of the proposed methods on adaptive cruise control (ACC) and obstacle avoidance problems with tight and time-varying control bounds, and compare it with existing CBF methods. The results show that the proposed approaches can generate more feasible, safer, and adaptive control solutions compared to existing methods, without requiring design of excessive additional constraints and complicated hyperparameter-tuning procedures.
\end{itemize}

This work is a significant extension of our conference
paper \cite{liu2023auxiliary}, in which we introduced AVCBFs specifically for safety constraints with high relative degrees. In \cite{liu2023auxiliary}, the hyperparameters related to AVCBFs are manually tuned. In addition to developing a parametrization method for
AVCBF to automatically tune its hyperparameters to ensure both safety and feasibility, this article also includes technical details
that extend the design of AVCBFs for constraints with mixed relative degrees and more complex simulation results.

\section{Preliminaries}
\label{sec:Preliminaries}

Consider an affine control system expressed as 
\begin{equation}
\label{eq:affine-control-system}
\dot{\boldsymbol{x}}=f(\boldsymbol{x})+g(\boldsymbol{x})\boldsymbol{u},
\end{equation}
 where $\boldsymbol{x}\in \mathbb{R}^{n}, f:\mathbb{R}^{n}\to\mathbb{R}^{n}$ and $g:\mathbb{R}^{n}\to\mathbb{R}^{n\times q}$ are locally Lipschitz, and $\boldsymbol{u}\in \mathcal U\subset \mathbb{R}^{q}$ denotes the control constraint set, which is defined as 
 
\begin{equation}
\label{eq:control-constraint}
\mathcal U \coloneqq \{\boldsymbol{u}\in \mathbb{R}^{q}:\boldsymbol{u}_{min}\le \boldsymbol{u} \le \boldsymbol{u}_{max} \}, \end{equation}
 with $\boldsymbol{u}_{min},\boldsymbol{u}_{max}\in \mathbb{R}^{q}$ (the vector inequalities are interpreted componentwise).
 
\begin{definition}[Class $\kappa$ function~\cite{Khalil:1173048}]
\label{def:class-k-f}
A continuous function $\alpha:[0,a)\to[0,+\infty],a>0$ is called a class $\kappa$ function if it is strictly increasing and $\alpha(0)=0.$
\end{definition}

\begin{definition}
\label{def:forward-inv}
A set $\mathcal C\subset \mathbb{R}^{n}$ is forward invariant for system \eqref{eq:affine-control-system} if its solutions for some $\boldsymbol{u} \in \mathcal U$ starting from any $\boldsymbol{x}(0) \in \mathcal C$ satisfy $\boldsymbol{x}(t) \in \mathcal C, \forall t \ge 0.$
\end{definition}

\begin{definition}
\label{def:relative-degree}
The relative degree of a differentiable function $b:\mathbb{R}^{n}\to\mathbb{R}$ is the minimum number of times we need to differentiate it along dynamics \eqref{eq:affine-control-system} until every component of $\boldsymbol{u}$ explicitly shows. 
\end{definition}
The relative degree defined above can also be referred to as the high-order relative degree, which does not include the case of mixed relative degree (discussed in \cite{xiao2022control}).
\begin{lemma}[\!\!\!\cite{glotfelter2017nonsmooth}]
\label{lem:for-invariance}
 Let \( b : [t_0, t_1] \to \mathbb{R} \) be a continuously differentiable function. If  
\( \dot{b}(t) \geq -\alpha(b(t)) \) for all \( t \in [t_0, t_1] \), where \( \alpha \) is a class \( \mathcal{K} \) function of its argument, and \( b(t_0) \geq 0 \), then \( b(t) \geq 0 \) for all \( t \in [t_0, t_1] \).
\end{lemma}

In this paper, safety is defined as the forward invariance of set $\mathcal C$. The relative degree of function $b$ is also referred to as the relative degree of constraint $b(\boldsymbol{x}) \ge 0$. For a constraint $b(\boldsymbol{x})\ge0$ with relative degree $m$, \ $b:\mathbb{R}^{n}\to\mathbb{R}$ and $\psi_{0}(\boldsymbol{x})\coloneqq b(\boldsymbol{x}),$ we can define a sequence of functions as $\psi_{i}:\mathbb{R}^{n}\to\mathbb{R},\ i\in \{1,...,m\}:$

\begin{equation}
\label{eq:sequence-f1}
\psi_{i}(\boldsymbol{x})\coloneqq\dot{\psi}_{i-1}(\boldsymbol{x})+\alpha_{i}(\psi_{i-1}(\boldsymbol{x})),\ i\in \{1,...,m\}, 
\end{equation}
where $\alpha_{i}(\cdot ),\ i\in \{1,...,m\}$ denotes a $(m-i)^{th}$ order differentiable class $\kappa$ function. A sequence of sets $\mathcal C_{i}$ are defined based on \eqref{eq:sequence-f1} as
\begin{equation}
\label{eq:sequence-set1}
\mathcal C_{i}\coloneqq \{\boldsymbol{x}\in\mathbb{R}^{n}:\psi_{i}(\boldsymbol{x})\ge 0\}, \ i\in \{0,...,m-1\}. 
\end{equation}

\begin{definition}[HOCBF~\cite{xiao2021high}]
\label{def:HOCBF}
Let $\psi_{i}(\boldsymbol{x}),\ i\in \{1,...,m\}$ be defined by \eqref{eq:sequence-f1} and $\mathcal C_{i},\ i\in \{0,...,m-1\}$ be defined by \eqref{eq:sequence-set1}. A function $b:\mathbb{R}^{n}\to\mathbb{R}$ is a High Order Control Barrier Function (HOCBF) with relative degree $m$ for system \eqref{eq:affine-control-system} if there exist $(m-i)^{\text{th}}$ order differentiable class $\kappa$ functions $\alpha_{i},\ i\in \{1,...,m\}$ such that
\begin{equation}
\label{eq:highest-HOCBF}
\begin{split}
\sup_{\boldsymbol{u}\in \mathcal U}[L_{f}^{m}b(\boldsymbol{x})+L_{g}L_{f}^{m-1}b(\boldsymbol{x})\boldsymbol{u}+O(b(\boldsymbol{x}))
+\\
\alpha_{m}(\psi_{m-1}(\boldsymbol{x}))]\ge 0,
\end{split}
\end{equation}
$\forall \boldsymbol{x}\in \mathcal C_{0}\cap,...,\cap \mathcal C_{m-1},$ where $O(\cdot)=\sum_{i=1}^{m-1}L_{f}^{i}(\alpha_{m-1}\circ\psi_{m-i-1})(\boldsymbol{x})$; $L_{f}^{m}$ denotes the $m^{\text{th}}$ Lie derivative along $f$ and $L_{g}$ denotes the matrix of Lie derivatives along the columns of $g$. $\psi_{i}(\boldsymbol{x})\ge0$ is referred to as the $i^{\text{th}}$ order HOCBF constraint. We assume that $L_{g}L_{f}^{m-1}b(\boldsymbol{x})\boldsymbol{u}\ne0$ on the boundary of set $\mathcal C_{0}\cap,...,\cap \mathcal C_{m-1}.$ 
\end{definition}

\begin{theorem}[Safety Guarantee~\cite{xiao2021high}]
\label{thm:safety-guarantee}
Given a HOCBF $b(\boldsymbol{x})$ from Def. \ref{def:HOCBF} with corresponding sets $\mathcal{C}_{0}, \dots,\mathcal {C}_{m-1}$ defined by \eqref{eq:sequence-set1}, if $\boldsymbol{x}(0) \in \mathcal {C}_{0}\cap \dots \cap \mathcal {C}_{m-1},$ then any Lipschitz controller $\boldsymbol{u}$ that satisfies the constraint in \eqref{eq:highest-HOCBF}, $\forall t\ge 0$ renders $\mathcal {C}_{0}\cap \dots \cap \mathcal {C}_{m-1}$ forward invariant for system \eqref{eq:affine-control-system}, $i.e., \boldsymbol{x} \in \mathcal {C}_{0}\cap \dots \cap \mathcal {C}_{m-1}, \forall t\ge 0.$
\end{theorem}

\begin{definition}[CLF~\cite{ames2012control}]
\label{def:control-l-f}
A continuously differentiable function $V:\mathbb{R}^{n}\to\mathbb{R}$ is an exponentially stabilizing Control Lyapunov Function (CLF) for system \eqref{eq:affine-control-system} if there exist constants $c_{1}>0, c_{2}>0,c_{3}>0$ such that for $\forall \boldsymbol{x} \in \mathbb{R}^{n}, c_{1}\left \|  \boldsymbol{x} \right \| ^{2} \le V(\boldsymbol{x}) \le c_{2}\left \|  \boldsymbol{x} \right \| ^{2},$
\begin{equation}
\label{eq:clf}
\inf_{\boldsymbol{u}\in \mathcal U}[L_{f}V(\boldsymbol{x})+L_{g}V(\boldsymbol{x})\boldsymbol{u}+c_{3}V(\boldsymbol{x})]\le 0.
\end{equation}
\end{definition}

The existing works \cite{nguyen2016exponential, xiao2021high} combine HOCBFs \eqref{eq:highest-HOCBF} for systems with high relative degree with quadratic costs to formulate safety-critical optimization problems. CLFs \eqref{eq:clf} can also be incorporated into optimization problems (see \cite{xiao2022sufficient, xiao2021adaptive}) if exponential convergence of certain states is desired. In these works, time is discretized into time intervals $[t_{k},t_{k+1})$ where $t_{0}=0, t_{N}=T,k\in \{0,...,N-1\}$,
and an optimization problem with constraints
given by HOCBFs and CLFs is solved in each
time interval. Since the state value $\boldsymbol{x}(t_{k})$ is fixed
at the beginning of the interval, these
constraints are linear in control, therefore each optimization problem is a QP. The optimal control $\boldsymbol{u}^{\ast}(t_{k})$ obtained by solving each QP is applied at the beginning of the interval and held constant for the
whole interval. During each interval, the state is updated using dynamics \eqref{eq:affine-control-system}. Since the controller updates only occur at specific sampling instances $t_{k},$ the system evolves uncontrolled between two updates within the sampling interval $[t_{k},t_{k+1}).$ This uncontrolled behavior can lead to the inter-sampling effect in \cite{singletary2020control}. Another form of uncontrolled behavior arises when the coefficient of the control input in the CBF constraint is partially zero for all states (as seen when the CBF candidate has a mixed relative degree in systems with multiple control inputs \cite{xiao2022control}) or at certain states (input nullification discussed in \cite{lindemann2018control}). These uncontrolled behaviors compromise safety.

This method, referred to as CBF-CLF-QP, ensures safety based on the assumption that the quadratic program (QP) remains feasible at every time step and that the uncontrolled behaviors are negligible. However, the feasibility is not guaranteed, especially under tight or time-varying control bounds. Moreover, mitigating the uncontrolled behaviors requires appropriate selection of both CBF candidates and the hyperparameters associated with them. However, the choice of CBF candidates and hyperparameters remains myopic. The authors of \cite{xiao2021adaptive} proposed a new type of HOCBF called PACBF, which introduced a time-varying penalty variable \( p_{i}(t) \) in front of the class \( \kappa \) function in the \( i^{\text{th}} \) order HOCBF constraint (\( i \in \{1, \dots, m\} \)), aiming to maximize the feasibility of solving CBF-CLF-QPs. However, the formulation of PACBFs requires the design of numerous additional constraints. Defining these constraints may not be straightforward and can lead to complex hyperparameter-tuning processes when both feasibility and uncontrolled behaviors are considered simultaneously, making manual hyperparameter tuning challenging for ensuring safety. To address these issues, we introduce Auxiliary-Variable Adaptive Control Barrier Functions (AVCBFs) and a parametrization method that ensures safety and feasibility by automatically tuning hyperparameters, as detailed in Sec. \ref{sec:Auxiliary-Variable Adaptive Control Barrier Functions}.

\section{Problem Formulation and Approach}
\label{sec:Problem Formulation and Approach}
Our goal is to generate a control strategy for system \eqref{eq:affine-control-system} such that it converges to a desired state, some measure of spent energy is minimized, safety is satisfied, and control limitations are observed. 

\textbf{Objective:} We consider the cost  
\begin{equation}
\label{eq:cost-function-1}
\begin{split}
 J(\boldsymbol{u}(t))=\int_{0}^{T} 
 \| \boldsymbol{u}(t) \| ^{2}dt+Q\left \| \boldsymbol{x}(T)-\boldsymbol{x}_{e} \right \| ^{2},
\end{split}
\end{equation}
where $\left \| \cdot \right \|$ denotes the 2-norm of a vector, and $T>0$ denotes the ending time; $Q>0$ denotes a weight factor and $\boldsymbol{x}_{e} \in \mathbb{R}^{n}$ is a desired state, which is assumed to be an equilibrium for system \eqref{eq:affine-control-system}. $Q\left \| \boldsymbol{x}(T)-\boldsymbol{x}_{e} \right \| ^{2}$ denotes state convergence.

\textbf{Safety Requirement:} System \eqref{eq:affine-control-system} should always satisfy one or more safety requirements of the form: 
\begin{equation}
\label{eq:Safety constraint}
b(\boldsymbol{x})\ge 0, \boldsymbol{x} \in \mathbb{R}^{n}, \forall t \in [0, T],
\end{equation}
where $b:\mathbb{R}^{n}\to\mathbb{R}$ is assumed to be a continuously differentiable equation. 

\textbf{Control Limitations:} The controller $\boldsymbol{u}$ should always satisfy \eqref{eq:control-constraint} for all $t \in [0, T].$

A control policy is \textbf{feasible} if all constraints derived from previously mentioned requirements are satisfied and mutually non-conflicting during period $[0, T]$. In this paper, we consider the following problem:

\begin{problem}
\label{prob:SACC-prob}
Find a feasible control policy for system \eqref{eq:affine-control-system} such that cost \eqref{eq:cost-function-1} is minimized.
\end{problem}

In \cite{xiao2021high}, the authors defined a HOCBF to enforce \eqref{eq:Safety constraint}. They also used a relaxed CLF to realize the state convergence in \eqref{eq:cost-function-1}. Since the cost is quadratic in $\boldsymbol{u}$, the Prob. \ref{prob:SACC-prob} using CBF-CLF-QPs was formulated as:
\begin{equation}
\label{eq:optimal control-cost}
\begin{split}
\min_{u(t),\delta(t)} \int_{0}^{T}(\left \| \boldsymbol{u}(t) \right \| ^{2}+Q\delta^{2}(t))dt.
\end{split}
\end{equation}
subject to
\begin{subequations}
\label{eq:hard constraints}
\begin{align}
L_{f}^{m}b(\boldsymbol{x})+L_{g}L_{f}^{m-1}&b(\boldsymbol{x})\boldsymbol{u}+O(b(\boldsymbol{x}))
+\alpha_{m}(\psi_{m-1}(\boldsymbol{x}))\ge 0,\label{subeq:HOCBF as 1}\\ 
L_{f}V(\boldsymbol{x})+&L_{g}V(\boldsymbol{x})\boldsymbol{u}+c_{3}V(\boldsymbol{x})\le \delta(t),\label{subeq:CLF as 2}\\
&\boldsymbol{u}_{min}\le \boldsymbol{u} \le \boldsymbol{u}_{max},\label{subeq:control bounds as 3}
\end{align}
\end{subequations}
where $V(\boldsymbol{x}(t))=(\boldsymbol{x}(t)-\boldsymbol{x}_{e})^{T}P(\boldsymbol{x}(t)-\boldsymbol{x}_{e}), P$ is positive definite, $c_{3}>0, Q>0$ and $\delta(t) \in \mathbb{R}$ is a relaxation variable (decision variable) to minimize for less violation of the strict CLF constraint. $b(\boldsymbol{x})$ has relative degree $m$ and $V(\boldsymbol{x})$ has relative degree 1. The above optimization problem is \textbf{feasible at a given state $\boldsymbol{x}$} if all the constraints define a non-empty set for the decision variables $\boldsymbol{u},\delta.$ As discussed in Sec. \ref{sec:Preliminaries}, the optimal control $\boldsymbol{u}^{\ast}(t_{k})$ obtained by solving each QP is applied at the beginning of the interval and held constant for the whole interval. However, the CBF-CLF-QPs could easily be infeasible at some $t_{k+1}$. In other words, after applying the constant vector $\boldsymbol{u}^{\ast}(t_{k})$ to system \eqref{eq:affine-control-system} for the time interval $[t_{k},t_{k+1})$, we may end up at a state $\boldsymbol{x}(t_{k+1})$ where the HOCBF constraint \eqref{subeq:HOCBF as 1} conflicts with the control bounds \eqref{subeq:control bounds as 3}, which would render the CBF-CLF-QP corresponding to getting $\boldsymbol{u}^{\ast}(t_{k+1})$ infeasible.

 \textbf{Approach:} In this paper, we introduce a time-varying auxiliary variable $a(t)$. Additionally, we will develop a new constraint $\psi_{0}(\boldsymbol{x},a(t))\ge 0$ based on the function $b(\boldsymbol{x}),$ and we require that when $\psi_{0}(\boldsymbol{x},a(t))\ge 0$ is satisfied, 
$b(\boldsymbol{x})\ge 0$ is also guaranteed to be satisfied. Therefore, 
$\psi_{0}(\boldsymbol{x},a(t))\ge 0$ will be the sufficient safety requirement. Similar to Eq. \eqref{eq:sequence-f1}, we need to differentiate $\psi_{0}$ multiple times. To ensure that the auxiliary variable $a(t)$ is smooth and can be differentiated multiple times, we design its auxiliary dynamic system, which is influenced by the auxiliary input $\nu$. Both $\boldsymbol{u}$ and $\nu$ will fully appear in the desired-order constraint,which will replace Eq. \eqref{subeq:HOCBF as 1} in QP to ensure the safety requirement.

Note that even if $b(\boldsymbol{x})$ has a mixed relative degree in systems with multiple control inputs, $\psi_{0}$ can be designed to ensure that all control inputs fully appear in a specific constraint. Moreover, $\nu$ is an unbounded decision variable incorporated into the QP formulation to dynamically adjust $a(t)$ based on the current state of the system and the Eq. \eqref{eq:optimal control-cost} becomes 
\begin{equation}
\label{eq:optimal control-cost new}
\min_{u(t),\delta(t),\nu(t)} \int_{0}^{T}(\left \| \boldsymbol{u}(t) \right \| ^{2}+Q\delta^{2}(t)+W(\nu(t)-a_{w})^{2})dt,
\end{equation}
where $W>0, a_{w} \in \mathbb{R}$ are hyperparameters related to $\nu(t).$
The variation of $a(t)$ provides additional flexibility for hyperparameter adaptation in each $\psi_{i}$ and reduces the likelihood of infeasible solutions. In terms of compromised safety and feasibility due to inter-sampling effects and inappropriate hyperparameter selection, we propose a parametrization method to tune hyperparameters (e.g., $Q_{t_{k}}, W_{t_{k}},a_{w,t_{k}},\alpha_{i,t_{k}}(\cdot)$)
such that the safety requirements and the control limitations are satisfied, i.e., $b(\boldsymbol{x}_{[0,T]})\ge 0,$ and $\boldsymbol{u}_{min}\le \boldsymbol{u}^{\ast}_{[0,T]} \le \boldsymbol{u}_{max}$, which remains a challenging problem for state-of-the-art methods.

\section{Auxiliary-Variable Adaptive Control Barrier Functions}
\label{sec:Auxiliary-Variable Adaptive Control Barrier Functions}
In this section, we introduce Auxiliary-Variable Adaptive Control Barrier Functions (AVCBFs) for safety-critical control.
We start with a simple example to motivate the need for AVCBFs.
\subsection{Motivation Example: Simplified Adaptive Cruise Control}
\label{subsec:SACC-problem}

Consider a Simplified Adaptive Cruise Control (SACC) problem with the dynamics of ego vehicle expressed as 
\begin{small}
\begin{equation}
\label{eq:SACC-dynamics}
\underbrace{\begin{bmatrix}
\dot{z}(t) \\
\dot{v}(t) 
\end{bmatrix}}_{\dot{\boldsymbol{x}}(t)}  
=\underbrace{\begin{bmatrix}
 v_{p}-v(t) \\
 0
\end{bmatrix}}_{f(\boldsymbol{x}(t))} 
+ \underbrace{\begin{bmatrix}
  0 \\
  1 
\end{bmatrix}}_{g(\boldsymbol{x}(t))}u(t),
\end{equation}
\end{small}
where $v_{p}>0, v(t)>0$ denote the velocity of the lead vehicle (constant velocity) and ego vehicle, respectively, $z(t)$ denotes the distance between the lead and ego vehicle and $u(t)$ denotes the acceleration (control) of ego vehicle, subject to the control constraints
\begin{equation}
\label{eq:simple-control-constraint}
u_{min}\le u(t) \le u_{max}, \forall t \ge0,
\end{equation}
where $u_{min}<0$ and $u_{max}>0$ are the minimum and maximum control input, respectively.

 For safety, we require that $z(t)$ always be greater than or equal to the safety distance denoted by $l_{p}>0,$ i.e., $z(t)\ge l_{p}, \forall t \ge 0.$ Based on Def. \ref{def:HOCBF}, let $\psi_{0}(\boldsymbol{x})\coloneqq b(\boldsymbol{x})=z(t)-l_{p}.$ From \eqref{eq:sequence-f1} and \eqref{eq:sequence-set1}, since the relative degree of $b(\boldsymbol{x})$ is 2, we have
\begin{equation}
\label{eq:SACC-HOCBF-sequence}
\begin{split}
&\psi_{1}(\boldsymbol{x})\coloneqq v_{p}-v(t)+k_{1}\psi_{0}(\boldsymbol{x})\ge 0
,\\
&\psi_{2}(\boldsymbol{x})\coloneqq -u(t)+k_{1}(v_{p}-v(t))+k_{2}\psi_{1}(\boldsymbol{x})\ge 0,
\end{split}
\end{equation}
where $\alpha_{1}(\psi_{0}(\boldsymbol{x}))\coloneqq k_{1}\psi_{0}(\boldsymbol{x}), \alpha_{2}(\psi_{1}(\boldsymbol{x}))\coloneqq k_{2}\psi_{1}(\boldsymbol{x}), k_{1}>0, k_{2}>0.$ The constant class $\kappa$ coefficients $k_{1}, k_{2}$ are always chosen to be small to provide the ego vehicle with a conservative control strategy for safety, i.e., smaller $k_{1}, k_{2}$ result in earlier braking (see~\cite{xiao2021high}). Suppose we aim to minimize the energy cost as $\int_{0}^{T} u^{2}(t)dt$. We can formulate the cost function in the QP with the constraint $\psi_{2}(\boldsymbol{x}) \geq 0$ and the control input constraint~\eqref{eq:simple-control-constraint} to obtain the optimal controller for the SACC problem. However, the feasible input set may become empty if $u(t) \leq k_{1} (v_{p} - v(t)) + k_{2} \psi_{1}(\boldsymbol{x}) < u_{\min}$, which leads to optimization infeasibility. In~\cite{xiao2021adaptive}, the authors introduced penalty variables in front of class $\kappa$ functions to enhance feasibility. This approach defines  
 $\psi_{0}(\boldsymbol{x}) \coloneqq b(\boldsymbol{x}) = z(t) - l_{p}$ as a PACBF, and additional constraints can be further defined as follows:
\begin{equation}
\label{eq:SACC-PACBF-sequence}
\begin{split}
\psi_{1}(\boldsymbol{x},p_{1})&\coloneqq v_{p}-v(t)+p_{1}(t)k_{1}\psi_{0}(\boldsymbol{x})\ge 0,\\
\psi_{2}(\boldsymbol{x},p_{1},&\boldsymbol{\nu})\coloneqq \nu_{1}(t)k_{1}\psi_{0}(\boldsymbol{x})+p_{1}(t)k_{1}(v_{p}
-v(t))\\
&-u(t)+\nu_{2}(t)k_{2}\psi_{1}(\boldsymbol{x},p_{1}(t))\ge 0,
\end{split}
\end{equation}
where $\nu_{1}(t)=\dot{p}_{1}(t),\nu_{2}(t)=p_{2}(t), p_{1}(t)\ge0,p_{2}(t)\ge0,\boldsymbol{\nu}=(\nu_{1}(t),\nu_{2}(t)).$ $p_{1}(t),p_{2}(t)$ are time-varying penalty variables, which alleviate the conservativeness of the control strategy and $\nu_{1}(t),\nu_{2}(t)$ are auxiliary inputs, which relax the constraints for $u(t)$ in $\psi_{2}(\boldsymbol{x},p_{1},\boldsymbol{\nu})\ge0$ and \eqref{eq:simple-control-constraint}. However, in practice, we need to define several additional constraints to make PACBF valid as shown in Eqs. (24)-(27) in \cite{xiao2021adaptive}. First, we need to define HOCBFs ($b_{1}(p_{1})=p_{1}(t),b_{2}(p_{2})=p_{2}(t))$ based on Def. \ref{def:HOCBF} to ensure $p_{1}(t)\ge0,p_{2}(t)\ge0.$ Next we need to define HOCBF ($b_{3}(p_{1})=p_{1,max}-p_{1}(t)$) to confine the value of $p_{1}(t)$ in the range $[0,p_{1,max}].$ We also need to define CLF ($V(p_{1})=(p_{1}(t)-p_{1}^{\ast})^{2}$) based on Def. \ref{def:control-l-f} to keep $p_{1}(t)$ close to a small value $p_{1}^{\ast}.$
 In fact, the strict monotonicity of a class $\kappa$ function ensures that the control feedback intensity increases as the function's input variable increases. This property allows the control strategy to vary smoothly without abrupt change, thereby reducing the risk of violating control constraints within inter-sampling intervals. Consequently, incorporating a class $\kappa$ function in the definition of a CBF enhances control safety. $b_{3}(p_{1}), V(p_{1})$ are necessary since the function $\psi_{0}(\boldsymbol{x},p_{1})\coloneqq p_{1}(t)k_{1}\psi_{0}(\boldsymbol{x})$ in first constraint in \eqref{eq:SACC-PACBF-sequence} is not a class $\kappa$ function with respect to $\psi_{0}(\boldsymbol{x}).$ This is because $p_{1}(t)k_{1}\psi_{0}(\boldsymbol{x})$ is not guaranteed to strictly increase with respect to $\psi_{0}(\boldsymbol{x})$ due to the time-varying nature of $p_{1}(t)$,
which is against Thm. \ref{thm:safety-guarantee}, therefore $\psi_{1}(\boldsymbol{x},p_{1})\ge 0$ in \eqref{eq:SACC-PACBF-sequence} may not guarantee $\psi_{0}(\boldsymbol{x})\ge 0,$ especially under inter-sampling effect. This illustrates why we have to limit the growth of $p_{1}(t)$ by defining $b_{3}(p_{1}),V(p_{1}).$ However, the way to choose appropriate values for $p_{1,max},p_{1}^{\ast}$ is not straightforward. We can imagine as the relative degree of $b(\boldsymbol{x})$ gets higher, the number of additional constraints we should define also gets larger, which results in complicated hyperparameter-tuning process. To address this issue, we introduce $a_{1}(t),a_{2}(t)$ in the form
\begin{small}
\begin{equation}
\label{eq:SACC-AVBCBF-sequence}
\begin{split}
\psi_{1}(\boldsymbol{x},\boldsymbol{a},\dot{a}_{1})\coloneqq a_{2}(t)(\dot{\psi}_{0}(\boldsymbol{x},a_{1}(t))
+k_{1}\psi_{0}(\boldsymbol{x},a_{1}(t)))\ge 0,\\
\psi_{2}(\boldsymbol{x},\boldsymbol{a},\dot{a}_{1},\boldsymbol{\nu})\coloneqq \nu_{2}(t)\frac{\psi_{1}(\boldsymbol{x},\boldsymbol{a},\dot{a}_{1}(t))}{a_{2}(t)} +a_{2}(t)(\nu_{1}(t)(z(t)\\
-l_{p})+2\dot{a}_{1}(t)(v_{p}-v(t))-a_{1}(t)u(t)+k_{1}\dot{\psi}_{0}(\boldsymbol{x},a_{1}(t)))\\
+k_{2}\psi_{1}(\boldsymbol{x},\boldsymbol{a},\dot{a}_{1}(t))\ge 0, 
\end{split}
\end{equation}
\end{small}
where $\psi_{0}(\boldsymbol{x},a_{1})\coloneqq a_{1}(t)b (\boldsymbol{x})=a_{1}(t)(z(t)-l_{p})$, $\boldsymbol{\nu}=[\nu_{1}(t),\nu_{2}(t)]^{T}=[\ddot{a}_{1}(t),\dot{a}_{2}(t)]^{T}$, $\boldsymbol{a}=[a_{1}(t),a_{2}(t)]^{T}$, $a_{1}(t),a_{2}(t)$ are time-varying auxiliary variables. Since $\psi_{0}(\boldsymbol{x},a_{1})\ge0,\psi_{1}(\boldsymbol{x},\boldsymbol{a},\dot{a}_{1})\ge 0$ will not be against $b(\boldsymbol{x})\ge 0,\dot{\psi}_{0}(\boldsymbol{x},a_{1})
+k_{1}\psi_{0}(\boldsymbol{x},a_{1})\ge 0$ iff $a_{1}(t)>0,a_{2}(t)>0,$ we need to ensure $a_{1}(t)>0,a_{2}(t)>0,$ which will be illustrated in Sec. \ref{sec:AVCBFs}.  $\nu_{1}(t),\nu_{2}(t)$ are auxiliary inputs which are used to alleviate the restriction of constraints for $u(t)$ in $\psi_{2}(\boldsymbol{x},\boldsymbol{a},\dot{a}_{1},\boldsymbol{\nu})\ge0$ and \eqref{eq:simple-control-constraint}. Different from the first constraint in \eqref{eq:SACC-PACBF-sequence}, $k_{1}\psi_{0}(\boldsymbol{x},a_{1})$ is still a class $\kappa$ function with respect to $\psi_{0}(\boldsymbol{x},a_{1}),$ therefore we do not need to define additional HOCBFs and CLFs like $b_{3}(p_{1}),V(p_{1})$ to limit the growth of $a_{1}(t).$
We can rewrite $\psi_{1} (\boldsymbol{x},\boldsymbol{a},\dot{a}_{1})$ in \eqref{eq:SACC-AVBCBF-sequence} as
\begin{equation}
\label{eq:SACC-AVBCBF-sequence-rewrite}
\begin{split}
\psi_{1}(\boldsymbol{x},\boldsymbol{a},\dot{a}_{1})\coloneqq a_{2}(t)a_{1}(t)(v_{p}-v(t)\\
+k_{1}(1+\frac{\dot{a}_{1}(t)}{k_{1}a_{1}(t)})b(\boldsymbol{x}))\ge 0.
\end{split}
\end{equation}
Compared to the first constraint in \eqref{eq:SACC-HOCBF-sequence}, $\frac{\dot{a}_{1}(t)}{a_{1}(t)}$ is a time-varying auxiliary term to alleviate the conservativeness of control that small $k_{1}$ originally has, which shows the adaptivity of auxiliary terms to constant class $\kappa$ coefficients. 
\subsection{Adaptive HOCBFs for Safety:\ AVCBFs}
\label{sec:AVCBFs}

Motivated by the SACC example in Sec. \ref{subsec:SACC-problem}, given a function $b:\mathbb{R}^{n}\to\mathbb{R}$ with relative degree $m$ for system \eqref{eq:affine-control-system} and a time-varying vector $\boldsymbol{a}(t)\coloneqq [a_{1}(t),\dots,a_{m}(t)]^{T}$ with positive components called auxiliary variables, the key idea in converting a regular HOCBF into an adaptive
one without defining excessive constraints is to place one auxiliary variable in front of each function in \eqref{eq:sequence-f1} similar to \eqref{eq:SACC-AVBCBF-sequence}. As described in Sec. \ref{subsec:SACC-problem}, we only need to ensure each $a_{i}(t)>0, i \in \{1,...,m\}.$ However, if a function $b:\mathbb{R}^{n}\to\mathbb{R}$ has a mixed relative degree for system \eqref{eq:affine-control-system} as seen in \cite{xiao2022control}, i.e., not all control inputs will first appear in the same constraint after the same number of differentiations, simply multiplying each auxiliary variable with each function will not ensure that all components of $\boldsymbol{u}$ completely appear. To address this issue, we first propose the following definition:
\begin{definition}[Minimum Relative Degree]
\label{def: minimum relative degree}
If a function $b:\mathbb{R}^{n}\to\mathbb{R}$ has a mixed or high-order relative degree for system \eqref{eq:affine-control-system}, the minimum relative degree, denoted by $\underline{m}$, is the smallest number of times $b$ needs to differentiate along dynamics \eqref{eq:affine-control-system} until at least one component of $\boldsymbol{u}$ explicitly shows. 
\end{definition}

Note that mixed relative degree is a special case of minimum relative degree. A function with high-order relative degree still satisfies the condition that at least one component of $\boldsymbol{u}$ explicitly appears after the required number of differentiations. Thus, high-order relative degree can indeed be considered another special case of minimum relative degree, specifically when the differentiation needed to reveal $\boldsymbol{u}$ (for all components) is the smallest across all possible scenarios. Given a function $b:\mathbb{R}^{n}\to\mathbb{R}$ with a minimum relative degree $\underline{m}$ for the system \eqref{eq:affine-control-system}, and a time-varying vector $\boldsymbol{a}(t)\coloneqq [a_{1}(t),\dots,a_{m_{a}}(t)]^{T}$, we introduce a set of positive auxiliary functions $\mathcal{A}_{i}(\boldsymbol{x},a_{i}(t))$ for $i \in \{1,...,m_{a}\}$, where $1 \le m_{a} \le \underline{m},$ and $ \boldsymbol{x}\in \mathbb{R}^{n}$ represents the states in \eqref{eq:affine-control-system}. These auxiliary functions will be placed in front of some functions in \eqref{eq:sequence-f1} to ensure that all components of $\boldsymbol{u}$ are revealed in $m_{a}^{\text{th}}$ function in \eqref{eq:sequence-f1}. Similarly to the auxiliary variables in \eqref{eq:SACC-AVBCBF-sequence}, these auxiliary functions are required to be positive.

 \begin{remark}[Auxiliary Functions with Reduced Relative Degree]
\label{rem:reduced degree}
We do not need to construct auxiliary functions up to the minimum relative degree \( \underline{m} \) because the introduction of auxiliary functions \( \mathcal{A}_{i}(\boldsymbol{x}, a_{i}(t)) \) for \( i \in \{1, \dots, m_{a}\} \), where \( 1 \leq m_{a} \leq \underline{m} \), is sufficient to reveal all components of \( \boldsymbol{u} \) in the \( m_a^{\text{th}} \) function in \eqref{eq:sequence-f1}. By allowing \( m_a \) to be flexible and less than or equal to \( \underline{m} \), we reduce computational complexity and avoid introducing unnecessary auxiliary functions while still ensuring that the control inputs are explicitly captured within the dynamics.
\end{remark}

To ensure each $\mathcal{A}_{i}(\boldsymbol{x},a_{i}(t))>0, i \in \{1,...,m_{a}\},$ we define auxiliary systems that contain auxiliary states $\boldsymbol{\pi}_{i}(t)$ and inputs $\nu_{i}(t)$, through which systems we can extend each $\mathcal{A}_{i}(\boldsymbol{x},a_{i}(t))$ to desired relative degree $m_{a}$, just like $b(\boldsymbol{x})$ has minimum relative degree $\underline{m}$
with respect to the dynamics \eqref{eq:affine-control-system}. Consider $m_{a}$ auxiliary systems in the form 
\begin{equation}
\label{eq:virtual-system}
\dot{\boldsymbol{\pi}}_{i}=F_{i}(\boldsymbol{\pi}_{i})+G_{i}(\boldsymbol{\pi}_{i})\nu_{i}, i \in \{1,...,m_{a}\},
\end{equation}
where $\boldsymbol{\pi}_{i}(t)\coloneqq [\pi_{i,1}(t),\dots,\pi_{i,m_{a}+1-i}(t)]^{T}\in \mathbb{R}^{m_{a}+1-i}$ denotes an auxiliary state with $\pi_{i,j}(t)\in \mathbb{R}, j \in \{1,...,m_{a}+1-i\}.$ $\nu_{i}\in \mathbb{R}$ denotes an auxiliary input for \eqref{eq:virtual-system}, $F_{i}:\mathbb{R}^{m_{a}+1-i}\to\mathbb{R}^{m_{a}+1-i}$ and $G_{i}:\mathbb{R}^{m_{a}+1-i}\to\mathbb{R}^{m_{a}+1-i}$ are locally Lipschitz. For simplicity, we just build up the connection between an auxiliary variable and the system as $a_{i}(t)=\pi_{i,1}(t), \dot{\pi}_{i,1}(t)=\pi_{i,2}(t),\dots,\dot{\pi}_{i,m_{a}-i}(t)=\pi_{i,m_{a}+1-i}(t)$ and make $\dot{\pi}_{i,m_{a}+1-i}(t)=\nu_{i}.$ The augmented systems are obtained as 
\begin{equation}
\label{eq:augmented-system}
\underbrace{\begin{bmatrix}
\dot{\boldsymbol{x}} \\
\dot{\boldsymbol{\pi}}_{i}
\end{bmatrix}}_{\dot{\boldsymbol{z}}_{i}}  
=\underbrace{\begin{bmatrix}
 f(\boldsymbol{x}) \\
 F_{i}(\boldsymbol{\pi}_{i})
\end{bmatrix}}_{\mathcal{F}_{i}(\boldsymbol{z}_{i})} 
+ \underbrace{\begin{bmatrix}
 g(\boldsymbol{x}) & \boldsymbol{0} \\
 \boldsymbol{0} & G_{i}(\boldsymbol{\pi}_{i}) 
\end{bmatrix}}_{\mathcal{G}_{i}(\boldsymbol{z}_{i})}\underbrace{\begin{bmatrix}
\boldsymbol{u}  \\
\nu_{i}
\end{bmatrix}}_{\dot{\boldsymbol{v}}_{i}}  
 i \in \{1,...,m_{a}\},
\end{equation}
where $\boldsymbol{z}_{i}\in \mathbb{R}^{n+m_{a}+1-i}, \boldsymbol{v}_{i}\in \mathbb{R}^{q+1}. ~\mathcal{F}_{i}:\mathbb{R}^{n+m_{a}+1-i}\to\mathbb{R}^{n+m_{a}+1-i}$ and $\mathcal{G}_{i}:\mathbb{R}^{n+m_{a}+1-i}\to\mathbb{R}^{(n+m_{a}+1-i)\times (q+1)}$ are locally Lipschitz. With these augmented systems, we can define many specific HOCBFs $h_{i}$ to enable $\mathcal{A}_{i}(\boldsymbol{z}_{i})$ to be positive. 
Given a function $h_{i}:\mathbb{R}^{n+m_{a}+1-i}\to\mathbb{R}$ with relative degree $m_{a}+1-i$ with respect to dynamics \eqref{eq:augmented-system}, we can define a sequence of functions $\varphi_{i,j}:\mathbb{R}^{n+m_{a}+1-i}\to\mathbb{R}, i \in\{1,...,m_{a}\}, j \in\{1,...,m_{a}+1-i\}:$
\begin{equation}
\label{eq:virtual-HOCBFs}
\varphi_{i,j}(\boldsymbol{z}_{i})\coloneqq\dot{\varphi}_{i,j-1}(\boldsymbol{z}_{i})+\alpha_{i,j}(\varphi_{i,j-1}(\boldsymbol{z}_{i})),
\end{equation}
where $\varphi_{i,0}(\boldsymbol{z}_{i})\coloneqq h_{i}(\boldsymbol{z}_{i}),$ $\alpha_{i,j}(\cdot)$ are $(m_{a}+1-i-j)^{th}$ order differentiable class $\kappa$ functions. Sets $\mathcal{B}_{i,j}$ are defined as
\begin{equation}
\label{eq:virtual-sets}
\mathcal B_{i,j}\coloneqq \{\boldsymbol{z}_{i}\in\mathbb{R}^{n+m_{a}+1-i}:\varphi_{i,j}(\boldsymbol{z}_{i})>0\}, \ j\in \{0,...,m_{a}-i\}. 
\end{equation}
Let $\varphi_{i,j}(\boldsymbol{z}_{i}),\ j\in \{1,...,m_{a}+1-i\}$ and $\mathcal B_{i,j},\ j\in \{0,...,m_{a}-i\}$ be defined by \eqref{eq:virtual-HOCBFs} and \eqref{eq:virtual-sets} respectively. By Def. \ref{def:HOCBF}, a function $h_{i}:\mathbb{R}^{n+m_{a}+1-i}\to\mathbb{R}$ is a HOCBF with relative degree $m_{a}+1-i$ for system \eqref{eq:virtual-system} if there exist class $\kappa$ functions $\alpha_{i,j},\ j\in \{1,...,m_{a}+1-i\}$ as in \eqref{eq:virtual-HOCBFs} such that
\begin{small}
\begin{equation}
\label{eq:highest-SHOCBF}
\begin{split}
\sup_{\boldsymbol{v}_{i}\in \mathbb{R}^{q+1}}[L_{\mathcal{F}_{i}}^{m_{a}+1-i}h_{i}(\boldsymbol{z}_{i})+L_{\mathcal{G}_{i}}L_{\mathcal{F}_{i}}^{m_{a}-i}h_{i}(\boldsymbol{z}_{i})\boldsymbol{v}_{i}+O_{i}(h_{i}(\boldsymbol{z}_{i}))\\
+ \alpha_{i,m_{a}+1-i}(\varphi_{i,m_{a}-i}(\boldsymbol{z}_{i}))] \ge \epsilon,
\end{split}
\end{equation}
\end{small}
$\forall \boldsymbol{z}_{i}\in \mathcal B_{i,0}\cap,...,\cap \mathcal B_{i,m_{a}-i}$. $O_{i}(\cdot)=\sum_{j=1}^{m_{a}-i}L_{\mathcal{F}_{i}}^{j}(\alpha_{i,m_{a}-i}\circ\varphi_{i,m_{a}-1-i})(\boldsymbol{z}_{i}) $ where $\circ$ denotes the composition of functions. $\epsilon$ is a positive constant which can be infinitely small. 

\begin{remark}
\label{rem:safety-guarantee-2}
If $h_{i}(\boldsymbol{z}_{i})$ is a HOCBF illustrated above and $\boldsymbol{z}_{i}(0) \in \mathcal {B}_{i,0}\cap \dots \cap \mathcal {B}_{i,m_{a}-i},$ then satisfying constraint in \eqref{eq:highest-SHOCBF} is equivalent to making $\varphi_{i,m_{a}+1-i}(\boldsymbol{z}_{i}(t))\ge \epsilon>0, \forall t\ge 0.$ Based on
\eqref{eq:virtual-HOCBFs}, since $\boldsymbol{z}_{i}(0) \in \mathcal {B}_{i,m_{a}-i}$ (i.e., $\varphi_{i,m_{a}-i}(\boldsymbol{z}_{i}(0))>0),$ then we have $\varphi_{i,m_{a}-i}(\boldsymbol{z}_{i}(t))>0$ (If there exists a $t_{1}\in (0,t_{2}]$, which makes $\varphi_{i,m_{a}-i}(\boldsymbol{z}_{i}(t_{1}))=0,$ then we have $\dot{\varphi}_{i,m_{a}-i}(\boldsymbol{z}_{i}(t_{1}))>0\Leftrightarrow \varphi_{i,m_{a}-i}(\boldsymbol{z}_{i}(t_{1}^{-}))\varphi_{i,m_{a}-i}(\boldsymbol{z}_{i}(t_{1}^{+}))<0,$ which is against the definition of $\alpha_{i,m_{a}+1-i}(\cdot),$ therefore $\forall t_{1}>0, \varphi_{i,m_{a}-i}(\boldsymbol{z}_{i}(t_{1}))>0,$ note that $t_{1}^{-},t_{1}^{+}$ denote the left and right limit). Based on \eqref{eq:virtual-HOCBFs}, since $\boldsymbol{z}_{i}(0) \in \mathcal {B}_{i,m_{a}-1-i},$ then similarly we have $\varphi_{i,m_{a}-1-i}(\boldsymbol{z}_{i}(t))>0,\forall t\ge 0.$ Repeatedly, we have $\varphi_{i,0}(\boldsymbol{z}_{i}(t))>0,\forall t\ge 0,$ therefore the sets $\mathcal {B}_{i,0},\dots,\mathcal {B}_{i,m_{a}-i}$ are forward invariant.
\end{remark}

For simplicity, we can make $h_{i}(\boldsymbol{z}_{i})=\mathcal{A}_{i}(\boldsymbol{z}_{i})= \mathcal{A}_{i}(\boldsymbol{x},a_{i}(t)).$ Based on Rem. \ref{rem:safety-guarantee-2}, each $\mathcal{A}_{i}(\boldsymbol{x},a_{i}(t))$ will be positive.

The remaining question is how to define an adaptive HOCBF to guarantee $b(\boldsymbol{x})\ge0$ with the assistance of auxiliary functions. Let $\boldsymbol{Z}(t)\coloneqq [\boldsymbol{z}_{1}(t),\dots,\boldsymbol{z}_{m_{a}}(t)]^{T}$ and $\boldsymbol{V}\coloneqq [\boldsymbol{v}_{1},\dots,\boldsymbol{v}_{m_{a}}]^{T}$ denote the augmented states and control inputs of system \eqref{eq:augmented-system}, respectively. We can define a sequence of functions 
\begin{small}
\begin{equation}
\label{eq:AVBCBF-sequence}
\begin{split}
&\psi_{0}(\boldsymbol{Z})\coloneqq \mathcal{A}_{1}(\boldsymbol{x},a_{1}(t))b(\boldsymbol{x}),\\
&\psi_{i}(\boldsymbol{Z})\coloneqq \mathcal{A}_{i+1}(\boldsymbol{x},a_{i+1}(t))(\dot{\psi}_{i-1}(\boldsymbol{Z})+\alpha_{i}(\psi_{i-1}(\boldsymbol{Z}))),
\end{split}
\end{equation}
\end{small}
where $i \in \{1,...,m_{a}-1\}, \psi_{m_{a}}(\boldsymbol{Z},\boldsymbol{V})\coloneqq \dot{\psi}_{m_{a}-1}(\boldsymbol{Z},\boldsymbol{V})+\alpha_{m_{a}}(\psi_{m_{a}-1}(\boldsymbol{Z})).$ We further define a sequence of sets $\mathcal{C}_{i}$ associated with \eqref{eq:AVBCBF-sequence} in the form 
\begin{equation}
\label{eq:AVBCBF-set}
\begin{split}
\mathcal C_{i}\coloneqq \{\boldsymbol{Z} \in \mathbb{R}^{n+\frac{m_{a}(m_{a}+1)}{2} } :\psi_{i}(\boldsymbol{Z})\ge 0\}, 
\end{split}
\end{equation}
where $i \in \{0,...,m_{a}-1\}.$
Since $\mathcal{A}_{i}(\boldsymbol{x},a_{i}(t))$ is a HOCBF with relative degree $m_{a}+1-i$ for \eqref{eq:virtual-system}, based on \eqref{eq:highest-SHOCBF}, we define a constraint set $\mathcal{U}_{\boldsymbol{V}}$ for $\boldsymbol{V}$ as 
\begin{equation}
\label{eq:constraint-up}
\begin{split}
\mathcal{U}_{\boldsymbol{V}}(\boldsymbol{Z})\coloneqq 
\{ \boldsymbol{V} \in \mathbb{R}^{q+m_{a}} :  
L_{\mathcal{F}_{i}}^{m_{a}+1-i} \mathcal{A}_{i}(\boldsymbol{x},a_{i}(t)) \\
+ L_{\mathcal{G}_{i}} L_{\mathcal{F}_{i}}^{m_{a}-i} \mathcal{A}_{i}(\boldsymbol{x},a_{i}(t)) \boldsymbol{v}_{i} 
+ O_{i}(\mathcal{A}_{i}(\boldsymbol{x},a_{i}(t))) \\
+ \alpha_{i,m_{a}+1-i} (\varphi_{i,m_{a}-i}(\mathcal{A}_{i}(\boldsymbol{x},a_{i}(t)))) 
\geq \epsilon, \\ i \in \{1,\dots,m_{a}\} \}.
\end{split}
\end{equation}
where $\varphi_{i,m_{a}-i}(\cdot)$ is defined similar to \eqref{eq:virtual-HOCBFs} and $\mathcal{A}_{i}(\boldsymbol{x},a_{i}(t))$ is ensured positive. $\epsilon$ is a positive constant which can be infinitely small. 

\begin{definition}[AVCBF]
\label{def:AVBCBF}
Let $\psi_{i}(\boldsymbol{Z}),\ i\in \{0,...,m_{a}\}$ be defined by \eqref{eq:AVBCBF-sequence} and $\mathcal C_{i},\ i\in \{0,...,m_{a}-1\}$ be defined by \eqref{eq:AVBCBF-set}. A function $b(\boldsymbol{x}):\mathbb{R}^{n}\to\mathbb{R}$ is an Auxiliary-Variable Adaptive Control Barrier Function (AVCBF) with minimum relative degree $\underline{m}$ for system \eqref{eq:affine-control-system} if every $\mathcal{A}_{i}(\boldsymbol{x},a_{i}(t)),i\in \{1,...,m_{a}\},m_{a}\le \underline{m}$ is a HOCBF with relative degree $m_{a}+1-i$ for the auxiliary system \eqref{eq:augmented-system}, and there exist $(m_{a}-j)^{th}$ order differentiable class $\kappa$ functions $\alpha_{j},j\in \{1,...,m_{a}-1\}$
and a class $\kappa$ functions $\alpha_{m_{a}}$ s.t.
\begin{small}
\begin{equation}
\label{eq:highest-AVBCBF}
\begin{split}
\sup_{\boldsymbol{V}\in \mathcal{U}_{V}}[\sum_{j=2}^{m_{a}-1}[(\prod_{k=j+1}^{m_{a}}\mathcal{A}_{k})\frac{\psi_{j-1}}{\mathcal{A}_{j}}\boldsymbol{H}_{j}] + \frac{\psi_{m_{a}-1}}{\mathcal{A}_{m_{a}}}\boldsymbol{H}_{m_{a}} \\ +(\prod_{i=2}^{m_{a}}\mathcal{A}_{i})b(\boldsymbol{x})\boldsymbol{H}_{1} +(\prod_{i=1}^{m_{a}}\mathcal{A}_{i})(L_{f}^{m_{a}}b(\boldsymbol{x})+L_{g}L_{f}^{m_{a}-1}b(\boldsymbol{x})\boldsymbol{u})\\+R(b(\boldsymbol{x}),\boldsymbol{Z})
+ \alpha_{m_{a}}(\psi_{m_{a}-1})] \ge 0,
\end{split}
\end{equation}
\end{small}
where $\boldsymbol{H}_{j}\coloneqq L_{\mathcal{F}_{j}}^{m_{a}+1-j}\mathcal{A}_{j}+L_{\mathcal{G}_{j}}L_{\mathcal{F}_{j}}^{m_{a}-j}\mathcal{A}_{j}\boldsymbol{v}_{j}, \forall \boldsymbol{Z} \in \mathcal C_{0}\cap,...,\cap \mathcal C_{m_{a}-1}$ and each $\mathcal{A}_{i}>0, i,j\in\{1,\dots,m_{a}\}.$ In \eqref{eq:highest-AVBCBF}, $R(b(\boldsymbol{x}),\boldsymbol{Z})$ denotes the remaining Lie derivative terms of $b(\boldsymbol{x})$ (or $\boldsymbol{Z}$) along $f$ (or $\mathcal{F}_{i},i\in\{1,\dots,m_{a}\}$) with degree less than $m_{a}$ (or $m_{a}+1-i$), which is similar to the form of $O(\cdot )$ in \eqref{eq:highest-HOCBF}.
\end{definition}

\begin{theorem}
\label{thm:safety-guarantee-3}
Given an AVCBF $b(\boldsymbol{x})$ from Def. \ref{def:AVBCBF} with corresponding sets $\mathcal{C}_{0}, \dots,\mathcal {C}_{m_{a}-1}$ defined by \eqref{eq:AVBCBF-set}, if $\boldsymbol{Z}(0) \in \mathcal {C}_{0}\cap \dots \cap \mathcal {C}_{m_{a}-1},$ then if there exists solution of Lipschitz controller $\boldsymbol{V}$ that satisfies the constraint in \eqref{eq:highest-AVBCBF} and also ensures $\boldsymbol{Z}\in \mathcal {C}_{m_{a}-1}$ for all $t\ge 0,$ then $\mathcal {C}_{0}\cap \dots \cap \mathcal {C}_{m_{a}-1}$ will be rendered forward invariant for system \eqref{eq:affine-control-system}, $i.e., \boldsymbol{Z} \in \mathcal {C}_{0}\cap \dots \cap \mathcal {C}_{m_{a}-1}, \forall t\ge 0.$ Moreover, $b(\boldsymbol{x})\ge 0$ is ensured for all $t\ge 0.$
\end{theorem}

\begin{proof}
If $b(\boldsymbol{x})$ is an AVCBF that is $m_{a}^{\text{th}}$ order differentiable, then satisfying constraint in \eqref{eq:highest-AVBCBF} while ensuring $\boldsymbol{Z} \in \mathcal {C}_{m_{a}-1}$ for all $t\ge 0$ is equivalent to make $\psi_{m_{a}-1}(\boldsymbol{Z})\ge 0, \forall t\ge 0.$ Since $\mathcal{A}_{m_{a}}(t)>0$, we have $\frac{\psi_{m_{a}-1}(\boldsymbol{Z})}{\mathcal{A}_{m_{a}}}\ge 0.$ Based on
\eqref{eq:AVBCBF-sequence}, since $\boldsymbol{Z}(0) \in \mathcal {C}_{m_{a}-2}$ (i.e., $\frac{\psi_{m_{a}-2}(\boldsymbol{Z}(0))}{\mathcal{A}_{m_{a}-1}(0)}\ge 0),\mathcal{A}_{m_{a}-1}(t)>0,$ then we have $\psi_{m_{a}-2}(\boldsymbol{Z})\ge 0$ (The proof of this follows from Lemma \ref{lem:for-invariance}), and also $\frac{\psi_{m_{a}-2}(\boldsymbol{Z})}{\mathcal{A}_{m_{a}-1}(t)}\ge 0.$ Based on \eqref{eq:AVBCBF-sequence}, since $\boldsymbol{Z}(0) \in \mathcal {C}_{m_{a}-3},\mathcal{A}_{m_{a}-2}(t)>0$ then similarly we have $\psi_{m_{a}-3}(\boldsymbol{Z})\ge 0$ and $\frac{\psi_{m_{a}-3}(\boldsymbol{Z})}{\mathcal{A}_{m_{a}-2}(t)}\ge 0,\forall t\ge 0.$ Repeatedly, we have $\psi_{0}(\boldsymbol{Z})\ge 0$ and $\frac{\psi_{0}(\boldsymbol{Z})}{\mathcal{A}_{1}(t)}\ge 0,\forall t\ge 0.$ Therefore the sets $\mathcal {C}_{0},\dots,\mathcal {C}_{m_{a}-1}$ are forward invariant and $b(\boldsymbol{x})=\frac{\psi_{0}(\boldsymbol{Z})}{\mathcal{A}_{1}(t)}\ge 0$ is ensured for all $t\ge 0$.
\end{proof}
Based on Thm. \ref{thm:safety-guarantee-3}, the safety regarding $b(\boldsymbol{x})=\frac{\psi_{0}(\boldsymbol{Z})}{\mathcal{A}_{1}(t)}\ge 0$ is guaranteed. The following theorem further addresses the feasibility of solving the CBF-CLF-QP discussed in Sec. \ref{sec:Problem Formulation and Approach}.
\begin{theorem}
\label{thm:feasibility-guarantee-3}
Given an AVCBF $b(\boldsymbol{x})$ from Def. \ref{def:AVBCBF} with corresponding sets $\mathcal{C}_{0}, \dots,\mathcal {C}_{m_{a}-1}$ defined by \eqref{eq:AVBCBF-set}, and the expression $\boldsymbol{H}_{j}\coloneqq L_{\mathcal{F}_{j}}^{m_{a}+1-j}\mathcal{A}_{j}+B_{j}\boldsymbol{u} + B_{j}^\nu \nu_{j}$ from \eqref{eq:highest-AVBCBF}, where
$B_{j}\coloneqq L_{\mathcal{G}_{j}}L_{\mathcal{F}_{j}}^{m_{a}-j}\mathcal{A}_{j}, B_{j}^\nu$ is the coefficient associated with $\nu_{j}.$
if $\boldsymbol{Z}(0) \in \mathcal {C}_{0}\cap \dots \cap \mathcal {C}_{m_{a}-1}$ and each $B_{j}^\nu, j\in \{1,...,m_{a}\}$ is always positive, then the optimization
problem for solving CLF-CBF-QP with constraints \eqref{subeq:CLF as 2},\eqref{subeq:control bounds as 3},\eqref{eq:constraint-up} and \eqref{eq:highest-AVBCBF} is point-wise
feasible for any $\boldsymbol{Z}$ lying inside $\mathcal {C}_{m_{a}-1},$ i.e., the optimization problem is solvable with a unique solution $\boldsymbol{V}^{\ast}$ for any $\boldsymbol{Z}$ when $\psi_{m_{a}-1}(\boldsymbol{Z})> 0$. 
\end{theorem}

\begin{proof}
If $\boldsymbol{Z}(0) \in \mathcal {C}_{0}\cap \dots \cap \mathcal {C}_{m_{a}-1}$ and for any $\boldsymbol{Z}, \psi_{m_{a}-1}(\boldsymbol{Z})>0,$ we have $\frac{\psi_{m_{a}-1}(\boldsymbol{Z})}{\mathcal{A}_{m_{a}}}\coloneqq \dot{\psi}_{m_{a}-2}(\boldsymbol{Z})+\alpha_{m_{a}-1}(\psi_{m_{a}-2}(\boldsymbol{Z})) >0.$ Based on Lemma \ref{lem:for-invariance}, Rem. \ref{rem:safety-guarantee-2} and \eqref{eq:AVBCBF-sequence}, since $\boldsymbol{Z}(0) \in \mathcal {C}_{m_{a}-2}$ (i.e., $\psi_{m_{a}-2}(\boldsymbol{Z}(0)) \ge 0),$ we have $\psi_{m_{a}-2}(\boldsymbol{Z})>0,$ and also $\frac{\psi_{m_{a}-2}(\boldsymbol{Z})}{\mathcal{A}_{m_{a}-1}}\coloneqq \dot{\psi}_{m_{a}-3}(\boldsymbol{Z})+\alpha_{m_{a}-2}(\psi_{m_{a}-3}(\boldsymbol{Z})) >0.$ Based on Lemma \ref{lem:for-invariance}, Rem. \ref{rem:safety-guarantee-2} and \eqref{eq:AVBCBF-sequence}, since $\boldsymbol{Z}(0) \in \mathcal {C}_{m_{a}-3}$ (i.e., $\psi_{m_{a}-3}(\boldsymbol{Z}(0)) \ge 0),$ we have $\psi_{m_{a}-3}(\boldsymbol{Z})>0,$ and also $\frac{\psi_{m_{a}-3}(\boldsymbol{Z})}{\mathcal{A}_{m_{a}-2}}\coloneqq \dot{\psi}_{m_{a}-4}(\boldsymbol{Z})+\alpha_{m_{a}-3}(\psi_{m_{a}-4}(\boldsymbol{Z})) >0.$ Repeatedly, we have $\psi_{0}(\boldsymbol{Z})>0,$ and also $\frac{\psi_{0}(\boldsymbol{Z})}{\mathcal{A}_{1}}\coloneqq b(\boldsymbol{x}) >0.$ Rewrite \eqref{eq:constraint-up}, we have
\begin{equation}
\label{eq:constraint-up-rewrite}
\begin{split}
L_{\mathcal{F}_{i}}^{m_{a}+1-i}\mathcal{A}_{i}+B_{i}\boldsymbol{u}
+ B_{i}^\nu \nu_{i}  \ge
- O_{i}(\mathcal{A}_{i}) \\
- \alpha_{i,m_{a}+1-i} (\varphi_{i,m_{a}-i}(\mathcal{A}_{i})) + \epsilon.
\end{split}
\end{equation}
 Rewrite \eqref{eq:highest-AVBCBF}, we have
\begin{equation}
\label{eq:highest-AVBCBF-rewrite}
\begin{split}
\sum_{j=2}^{m_{a}-1}[(\prod_{k=j+1}^{m_{a}}\mathcal{A}_{k})\frac{\psi_{j-1}}{\mathcal{A}_{j}}(L_{\mathcal{F}_{j}}^{m_{a}+1-j}\mathcal{A}_{j}+B_{j}\boldsymbol{u} + B_{j}^\nu \nu_{j})] \\+ \frac{\psi_{m_{a}-1}}{\mathcal{A}_{m_{a}}}(L_{\mathcal{F}_{m_{a}}}\mathcal{A}_{m_{a}}+B_{m_{a}}\boldsymbol{u} + B_{m_{a}}^\nu \nu_{m_{a}})\\ +(\prod_{i=2}^{m_{a}}\mathcal{A}_{i})b(\boldsymbol{x})(L_{\mathcal{F}_{1}}^{m_{a}}\mathcal{A}_{1}+B_{1}\boldsymbol{u} + B_{1}^\nu \nu_{1}) \\ \ge -(\prod_{i=1}^{m_{a}}\mathcal{A}_{i})(L_{f}^{m_{a}}b(\boldsymbol{x})+L_{g}L_{f}^{m_{a}-1}b(\boldsymbol{x})\boldsymbol{u})\\-R(b(\boldsymbol{x}),\boldsymbol{Z})
- \alpha_{m_{a}}(\psi_{m_{a}-1}).
\end{split}
\end{equation}
Since the coefficients of $\nu_{i}$ in \eqref{eq:constraint-up-rewrite} and \eqref{eq:highest-AVBCBF-rewrite} are always positive (i.e., $B_{i}^\nu >0,(\prod_{k=j+1}^{m_{a}}\mathcal{A}_{k})\frac{\psi_{j-1}}{\mathcal{A}_{j}} B_{j}^\nu >0,\frac{\psi_{m_{a}-1}}{\mathcal{A}_{m_{a}}}B_{m_{a}}^\nu >0, (\prod_{i=2}^{m_{a}}\mathcal{A}_{i})b(\boldsymbol{x})B_{1}^\nu >0, i,j\in \{1,...,m_{a}\}),$ and there is no upper bound for  $\nu_{i}$, there always exists a positive $\nu_{i}$ such that the feasible region between \eqref{eq:constraint-up-rewrite} and \eqref{eq:highest-AVBCBF-rewrite} is non-empty. Furthermore, since the constraints \eqref{subeq:CLF as 2} and \eqref{subeq:control bounds as 3} are also satisfied, therefore the optimization problem for solving CLF-CBF-QP with constraints \eqref{subeq:CLF as 2},\eqref{subeq:control bounds as 3},\eqref{eq:constraint-up} and \eqref{eq:highest-AVBCBF} is point-wise
feasible for any $\boldsymbol{Z}$ lying inside $\mathcal {C}_{m_{a}-1}.$
\end{proof}

The coefficient of the control input $B_{j}$ in \eqref{eq:highest-AVBCBF-rewrite} can be specifically designed to ensure that all components of the control input are fully represented, while also guaranteeing that no component of the control input is zero for any state. Therefore, the previously mentioned CBF constraint with mixed relative degree \cite{xiao2022control} or with input nullification \cite{lindemann2018control} can both be handled. Since both Thm. \ref{thm:safety-guarantee-3} and Thm. \ref{thm:feasibility-guarantee-3} require ensuring $\psi_{m_{a}-1}(\boldsymbol{Z})> 0$, this condition can be considered a criterion for guaranteeing feasibility and safety. The remaining question is how to design an algorithm that efficiently adjusts the hyperparameters related to AVCBF based on this criterion to ensure both safety and feasibility. This motivates the parametrization method introduced in the next subsection.

\subsection{Optimal Control with AVCBFs: A Parametrization Method for Tuning Hyperparameters}
\label{subsec: optimal-control}
Reconsider an optimal control problem from \eqref{eq:optimal control-cost} as
\begin{small}
\begin{equation}
\label{eq:cost-function-2}
\begin{split}
 \min_{\boldsymbol{u},\delta }\int_{0}^{T} 
[ D(\left \| \boldsymbol{u} \right \| )+Q\delta^{2}]dt,
\end{split}
\end{equation}
\end{small}
where $\left \| \cdot \right \|$ denotes the 2-norm of a vector, $D(\cdot)$ is a strictly increasing function of its argument and $T>0$ denotes the ending time. $\delta$ is the relaxed variable that makes state convergence a soft constraint, and its corresponding weight factor is $Q$. Since we need to introduce auxiliary inputs $v_{i}$ to enhance the feasibility of optimization, we should reformulate the cost in \eqref{eq:cost-function-2} as
\begin{small}
\begin{equation}
\label{eq:cost-function-3}
\begin{split}
 \min_{\boldsymbol{u},\boldsymbol{\nu},\delta} \int_{0}^{T} 
 [D(\left \| \boldsymbol{u} \right \| )+\sum_{i=1}^{m}W_{i}(\nu_{i}-a_{i,w})^{2}+Q\delta^{2}]dt.
\end{split}
\end{equation}
\end{small}
In \eqref{eq:cost-function-3}, $W_{i}$ is a positive scalar and $a_{i,w}\in \mathbb{R}$ is the scalar to which we hope each auxiliary input $\nu_{i}$ converges. Both are chosen to tune the performance of the controller. We can formulate the CLFs, HOCBFs and AVCBFs introduced in Def. \ref{def:control-l-f}, Sec. \ref{sec:AVCBFs} and Def. \ref{def:AVBCBF} as constraints of the QP with cost function \eqref{eq:cost-function-3} to realize safety-critical control. Next we will show AVCBFs can enhance the feasibility of solving QP compared with classical HOCBFs in Def. \ref{def:HOCBF} and PACBF \cite{xiao2021adaptive}.

In augmented system \eqref{eq:augmented-system}, if we define $\mathcal{A}_{i}(\boldsymbol{x},a_{i}(t))=a_{i}(t)=\pi_{i,1}(t)=1, \dot{\pi}_{i,1}(t)=\dot{\pi}_{i,2}(t)=\cdots=\dot{\pi}_{i,m+1-i}(t)=0,$ then the way we construct functions and sets in \eqref{eq:virtual-HOCBFs} and \eqref{eq:virtual-sets} are exactly the same as \eqref{eq:sequence-f1} and \eqref{eq:sequence-set1}, which means classical HOCBF is in fact one specific case of AVCBF. Assume that the highest order HOCBF constraint \eqref{eq:highest-HOCBF} conflicts with control input constraints \eqref{eq:control-constraint} at $t=t_{b},$ i.e., we can not find a feasible controller $\boldsymbol{u}^{\ast}(t_{b})$ to satisfy \eqref{eq:highest-HOCBF} and \eqref{eq:control-constraint}. Instead, starting from a time slot $t=t_{a}$ which is just before $t=t_{b}$ ($t_{b}-t_{a}=N_{c}\bigtriangleup t$ where $\bigtriangleup t$ is an infinitely small time interval and $N_{c}$ is a positive time window), we change the control framework of classical HOCBFs to AVCBFs instantly. Suppose we can find appropriate hyperparameters for AVCBFs to ensure two constraints in \eqref{eq:constraint-up} and \eqref{eq:highest-AVBCBF} are satisfied given $\boldsymbol{u}$ constrained by \eqref{eq:control-constraint} on $[t_{a},t_{b}],$ then there exist solutions for the optimal control problem and the feasibility of solving QP is enhanced. Relying on AVCBFs, we can discretize the whole time period $[0,T]$ into several small time intervals like $[t_{a},t_{b}]$ to maximize the feasibility of solving QP under safety constraints. Besides safety and feasibility, another benefit of using AVCBFs is that the conservativeness of the control strategy can also be ameliorated. For example, from \eqref{eq:AVBCBF-sequence}, we can rewrite $\psi_{i}(\boldsymbol{Z})\ge 0$ as
\begin{equation}
\label{eq:AVCBF-rewrite}
\begin{split}
\dot{\phi}_{i-1}(\boldsymbol{Z})+k_{i}(1+\frac{\dot{\mathcal{A}}_{i}}{k_{i}\mathcal{A}_{i}}) \phi_{i-1}(\boldsymbol{Z})\ge0,
\end{split}
\end{equation}
where $\phi_{i-1}(\boldsymbol{Z})=\frac{\psi_{i-1}(\boldsymbol{Z})}{\mathcal{A}_{i}}$, $\alpha_{i}(\psi_{i-1}(\boldsymbol{Z}))=k_{i}\mathcal{A}_{i}\phi_{i-1}(\boldsymbol{Z})$, $k_{i}>0, i\in \{1,\dots,m_{a}\}.$ The term $\frac{\dot{\mathcal{A}}_{i}}{\mathcal{A}_{i}}$ can be adjusted adaptable to ameliorate the conservativeness of control strategy that $k_{i}\phi_{i-1}(\boldsymbol{Z})$ may have, i.e., the ego vehicle can brake earlier or later given time-varying control constraint $\boldsymbol{u}_{min}(t)\le \boldsymbol{u} \le\boldsymbol{u}_{max}(t),$ which confirms the adaptivity of AVCBFs to control constraint and conservativeness of control strategy. 
Unlike PACBFs, which introduce a non-negative time-varying penalty variable $p(t)$ in front of class $\kappa$ function, we don't require $1+\frac{\dot{\mathcal{A}}_{i}}{k_{i}\mathcal{A}_{i}}\ge 0$ (the time-varying coefficient of class $\kappa$ function can be negative if needed), which allows for a broader range of feasible control inputs. This flexibility enables AVCBFs to reduce conservativeness and improve adaptivity and robustness.
The remaining question is how to develop automatic hyperparameter-tuning algorithms to find appropriate hyperparameters to ensure safety and feasibility.

We use a gradient-ascent method as the parametrization method to find the optimal hyperparameters for AVCBFs such that the safety-feasibility criterion ($\psi_{m_{a}-1}(\boldsymbol{Z})> 0$) mentioned in Sec. \ref{sec:AVCBFs} is satisfied.  
There are multiple hyperparameters to consider, such as those related to HOCBFs (e.g., the class $\kappa$ function hyperparameters in \eqref{eq:virtual-HOCBFs}), those associated with AVCBFs (e.g., the class $\kappa$ function hyperparameters in \eqref{eq:AVBCBF-sequence}), as well as $W_{i}$ and $a_{i,w}$. However, not all of these hyperparameters necessarily need to be adjusted to satisfy both safety and feasibility. The hyperparameters of the class $\kappa$ functions in \eqref{eq:virtual-HOCBFs} generally do not need to be adjusted, as each $\mathcal{A}_{i}$ is guaranteed to be positive if $\nu_{i}$ is unbounded. The hyperparameters of the class $\kappa$ functions in \eqref{eq:AVBCBF-sequence} also do not need to be adjusted, as they vary dynamically with time-varying $\frac{\dot{\mathcal{A}}_{i}}{\mathcal{A}_{i}}$. In cost function \eqref{eq:cost-function-3}, we can tune hyperparameters $W_{i}$ and $a_{i,w}$ to adjust the corresponding ratio $\frac{\dot{\mathcal{A}}_{i}}{\mathcal{A}_{i}}$ to change the performance of the optimal controller. Tuning $a_{i,w}$ defines the nominal or target value for $\nu_{i}$ while tuning $W_{i}$ controls how strongly deviations from $a_{i,w}$ are penalized. In most cases, tuning $a_{i,w}$ first is preferable, since it directly sets the nominal operating point. Then, 
$W_{i}$ can be fine-tuned to balance constraint enforcement (e.g., we can determine whether to enforce the CLF constraint more strictly to guarantee state convergence or to prioritize minimizing energy with stronger enforcement). To reduce computational complexity, we only adjust the hyperparameter $a_{i,w}$ in Alg. \ref{alg:parametrization-cbf}.
\begin{algorithm}
\caption{The Parametrization Method}
\begin{algorithmic}[1]
\label{alg:parametrization-cbf}
\begin{small}
 \renewcommand{\algorithmicrequire}{\textbf{Input:}}
 \renewcommand{\algorithmicensure}{\textbf{Output:}}
 \REQUIRE An optimal control problem as stated in Prob. \ref{prob:SACC-prob} with initial states $\boldsymbol{x}(t_{0})$, initial hyperparameters including $a_{i,w}(t_{0}),i\in \{1,...,m_{a}\}$, a sampling interval $\bigtriangleup t$, an iteration limit $J_{m}>0$, a rollback horizon $N_{c}>0,$ a threshold $\varepsilon>0$, a learning rate $\gamma >0$ and the ending time $T>0$. 
 \ENSURE The optimal hyperparameters $a_{i,w}^{\ast}(t_{l})$ and number of iterations $j(t_{l})$ at some time steps $t_{l}, l\le \frac{T}{\bigtriangleup t}$. \\
 \STATE 
 Discretize time into invervals $[t_{k},t_{k+1})$ where $t_{N}=T,k\in \{0,...,N-1\}$, set $a_{i,w}(t_{k+1})=a_{i,w}(t_{0})$.
 \STATE 
 Solve \eqref{eq:cost-function-3} with constraints \eqref{subeq:CLF as 2},\eqref{subeq:control bounds as 3},\eqref{eq:constraint-up} and \eqref{eq:highest-AVBCBF} (CBF-CLF-QP) at each time step until $\psi_{m_{a}-1}(\boldsymbol{Z}(t_{f}))> \varepsilon$ is no longer satisfied for the first time. 
 \STATE
 Revert to the time step $t_{k}$ where $t_{k}=t_{f}-N_{c}\bigtriangleup t$ (if $t_{f}-N_{c}\bigtriangleup t <0, t_{k}=0$), set iteration $j=1, j(t_{l})=0, l\in \{0,\dots,f\}$.
    \FOR{$j \le J_{m}$}
        \WHILE{$t_{k+1} \le t_{f}$} 
          \STATE Solve CBF-CLF-QP with $a_{i,w}(t_{l}), l\in \{k,\dots,f\}$, get the minimum value of $\psi_{m_{a}-1}(\boldsymbol{Z}(t_{l}))$ as $\psi_{\text{min}}$.
          \STATE
          If $\psi_{\text{min}}>\varepsilon$, jump to $13^{\text{th}}$ step of Alg. \ref{alg:parametrization-cbf}.
          \STATE
          Evaluate $\frac{\partial \psi_{\text{min}}}{\partial a_{i,w}(t_{k})}$, $a_{i,w}(t_{k})\gets a_{i,w}(t_{k})+\gamma \frac{\partial \psi_{\text{min}}}{\partial a_{i,w}(t_{k})}$. 
          \STATE
          $j(t_{k})=j$, $k = k+1$.
        \ENDWHILE
        \STATE
        $k=k-N_{c}, j= j+1$.
    \ENDFOR
 \STATE
 $a_{i,w}^{\ast}(t_{l})= a_{i,w}(t_{l}),l\in \{0,\dots,f\}$.
 \RETURN $a_{i,w}^{\ast}(t_{l}), j(t_{l}),l\in \{0,\dots,f\}$.
\end{small}
\end{algorithmic}
\end{algorithm}

In Alg. \ref{alg:parametrization-cbf}, when reaching time step $t_{f}$ and finding that solving CBF-CLF-QP is infeasible, the process reverts to $N_{c}$ time steps earlier and retunes the hyperparameters to ensure feasibility. $N_{c}$ serves as a rollback horizon that enables the algorithm to proactively test feasibility over a past time window, allowing the algorithm to systematically re-evaluate and adjust hyperparameters before reaching infeasibility at $t_{f}$. This approach acts as an implicit feasibility test, preventing local infeasibility from propagating through future time steps. From $4^{\text{th}}$ step to $12^{\text{th}}$ step, the algorithm first tunes the hyperparameters across the entire time window to ensure global consistency before proceeding to the next iteration $j+1$, rather than making iterative adjustments at each time step independently. This approach allows for a more coordinated and holistic optimization of hyperparameters over the window, reducing local inconsistencies. After optimizing the hyperparameters for the entire window, the algorithm then iterates to refine the adjustments further, ensuring the minimum value of $\psi_{m_{a}-1}$ remains above the feasibility-safety threshold $\varepsilon$. By structuring the tuning process in this way, the algorithm minimizes unnecessary corrections and prevents abrupt fluctuations in hyperparameter values. Additionally, it anticipates inter-sampling effects by iteratively checking the feasibility-safety threshold, ensuring that the system satisfies safety requirements under the optimized hyperparameters $a_{i,w}^{\ast}(t_{l})$. Note that by using the optimized hyperparameters $a_{i,w}^{\ast}(t_{l})$ in CBF-CLF-QP, we obtain the optimal control input $\boldsymbol{u}^{\ast}(t_{l})$ corresponding to these optimized hyperparameters at the same time step. Since Alg. \ref{alg:parametrization-cbf} terminates at $t_{f}$, which is likely smaller than $T$, to ensure safe control over the entire interval \( [t_0, T] \), Alg. \ref{alg:parametrization-cbf} may need to be executed multiple times, and each time, the future hyperparameters should be set as $a_{i,w}(t_{k+1})=a^{\ast}_{i,w}(t_{f})$ where $k\in \{f,...,N-1\}$. The variable \( j(t_l) \) records the number of iterations required at each time step within a single execution of the algorithm. When multiple executions of the algorithm overlap in time, the values of \( j(t_l) \) from each execution must be accumulated to obtain the total number of iterations required at each time step across all executions (i.e., \( \sum j(t_l) \)). As a result, \( \sum j(t_l) \) may exceed the iteration limit \( J_m \).

\begin{remark}
\label{rem: minimum-value}
In $6^{\text{th}}$ step of Alg. \ref{alg:parametrization-cbf}, we solve the CBF-CLF-QP with \( a_{i,w}(t_{l}) \), where \( l \in \{k, \dots, f\} \), to obtain the minimum value of \( \psi_{m_{a}-1}(\boldsymbol{Z}(t_{l})) \) at corresponding time steps. A general approach is to use the min-operator to determine the minimum value as:
\begin{equation}
\psi_{\text{min}}=\min \big[ \psi_{m_{a}-1}(\boldsymbol{Z}(t_{k})), \dots, \psi_{m_{a}-1}(\boldsymbol{Z}(t_{f})) \big].
\end{equation}
However, this operator is nonsmooth, which may lead to unstable updates in gradient-based optimization. Alternatively, we can approximate the min-operator using a smooth soft-minimum approximation in \cite{lindemann2018control,rabiee2025soft}. While this approximation ensures differentiability, the soft-minimum approximation is highly sensitive to variations in \( \psi_{m_{a}-1} \), potentially leading to overly aggressive updates in hyperparameter \( a_{i,w}(t_{l}) \). This could result in large fluctuations, causing the gradient-based optimization to overfit to local variations in \( \psi_{m_{a}-1} \). Meanwhile, the soft-minimum approximation does not yield the exact minimum value, which may result in an underestimation of the minimum and lead to overly conservative hyperparameter updates. In this paper, we choose to use the min-operator instead.

\end{remark}

\begin{remark}[Limitation of Approaches with Auxiliary Inputs]
\label{rem: PACBF-AVBCBF} 
Ensuring the satisfaction of the $i^{\text{th}}$ order AVCBF constraint as shown in \eqref{eq:AVBCBF-set} when $i\in\{1,\dots,m_{a}-1\},$ i.e., $\psi_{i}(\boldsymbol{Z})\ge 0$ will guarantee $\psi_{i-1}(\boldsymbol{Z})\ge 0$ based on the proof of Thm. \ref{thm:safety-guarantee-3}, which theoretically outperforms PACBF. However, both approaches can not ensure satisfying $\psi_{m_{a}}(\boldsymbol{Z})\ge 0$ will guarantee $\psi_{m_{a}-1}(\boldsymbol{Z})\ge 0$ since the growth of $\nu_{i}$ is unbounded. Therefore in Thm. \ref{thm:safety-guarantee-3}, $\boldsymbol{Z}\in \mathcal {C}_{m_{a}-1}$ for all $t\ge 0$ also needs to be satisfied to guarantee the forward invariance of the intersection of sets. 
\end{remark}

\begin{remark}[Auxiliary Functions with Selective Application]
\label{rem: sufficient-con}
From \eqref{eq:AVCBF-rewrite}, we observe that incorporating the \( i^{\text{th}} \) auxiliary function $\mathcal{A}_{i}$ allows the conservativeness of the \( i^{\text{th}} \) class \( \kappa \) function to be mitigated through a time-varying ratio $\frac{\dot{\mathcal{A}}_{i}}{\mathcal{A}_{i}}$. Similarly, from \eqref{eq:highest-AVBCBF}, we see that the \( m_{a}^{\text{th}} \) order constraint can be further relaxed by introducing an additional unbounded auxiliary input \( \nu_i \), thereby enhancing feasibility. However it is not necessary to apply auxiliary functions as many as from $\mathcal{A}_{1}$ to $\mathcal{A}_{m_{a}}$ if solving CBF-CLF-QP is feasible with fewer auxiliary functions, which allows us to choose an appropriate number of auxiliary functions for the AVCBF constraints to reduce the complexity. In other words, the number of auxiliary functions can be less than or equal to the relative degree $m_{a}$.
\end{remark}

\section{Implementation and Results}
\label{sec:Implementation and Results}
In this section, we consider three classes of problems where the proposed approaches are applied: (i) an Adaptive Cruise Control (ACC) problem for a ego vehicle with a high-order relative degree system, (ii) an obstacle avoidance problem for a ground vehicle with a high-order relative degree system, and (iii) an obstacle avoidance problem for a ground vehicle with a mixed relative degree system. All computations and simulations were conducted in MATLAB, where \texttt{quadprog} was used to solve the quadratic programs, and \texttt{ode45} was employed to integrate the system dynamics.
\subsection{Adaptive Cruise Control for High-Order Relative Degree Systems}
\label{subsec:Adaptive Cruise Control}
\subsubsection{Vehicle Dynamics}
\label{subsubsec:Vehicle Dynamics}
We consider a nonlinear vehicle dynamics in the form
\begin{small}
\begin{equation}
\label{eq:ACC-dynamics}
\underbrace{\begin{bmatrix}
\dot{z}(t) \\
\dot{v}(t) 
\end{bmatrix}}_{\dot{\boldsymbol{x}}(t)}  
=\underbrace{\begin{bmatrix}
 v_{p}-v(t) \\
 -\frac{1}{M}F_{r}(v(t))
\end{bmatrix}}_{f(\boldsymbol{x}(t))} 
+ \underbrace{\begin{bmatrix}
  0 \\
  \frac{1}{M} 
\end{bmatrix}}_{g(\boldsymbol{x}(t))}u(t),
\end{equation}
\end{small}
where $M$ denotes the mass of the ego vehicle and $v_{p}>0$ denotes the velocity of the lead vehicle. $z(t)$ denotes the distance between two vehicles and $F_{r}(v(t))=f_{0}sgn(v(t))+f_{1}v(t)+f_{2}v^{2}(t)$ denotes the resistance force as in \cite{Khalil:1173048}, where $f_{0},f_{1},f_{2}$ are positive scalars determined empirically and $v(t)>0$ denotes the velocity of the ego vehicle. The first term in $F_{r}(t)$ denotes the Coulomb friction force, the second term denotes the viscous friction force and the last term denotes the aerodynamic drag.

\subsubsection{Vehicle Limitations}
Vehicle limitations include vehicle constraints on safe distance, speed and acceleration.

\textbf{Safe distance constraint:} The distance is considered safe if $z(t)\ge l_{p}$ is satisfied $\forall t \in [0,T]$, where $l_{p}$ denotes the minimum distance two vehicles should maintain.

\textbf{Speed constraint:} The ego vehicle should achieve a desired speed  $v_{d}>0.$

\textbf{Acceleration constraint:} The ego vehicle should minimize the following cost
\begin{small}
\begin{equation}
\label{eq:minimal-u}
\min_{u(t)} \int_{0}^{T}(\frac{u(t)-F_{r}(v(t))}{M})^{2}dt 
\end{equation}
\end{small}
when the acceleration is constrained in the form 
\begin{equation}
\label{eq:constraint-u}
-c_{d}(t)Mg\le u(t) \le c_{a}(t)Mg, \forall t \in [0,T], 
\end{equation}
where $g$ denotes the gravity constant, $c_{d}(t)>0$ and $c_{a}(t)>0$ are deceleration and acceleration coefficients respectively.

To satisfy the constraint on speed, we define a CLF $V(\boldsymbol{x}(t)) \coloneqq(v(t)-v_{d})^{2}$ with $c_{1}=c_{2}=1$ to stabilize $v(t)$ to $v_{d}$ and formulate the relaxed constraint in \eqref{eq:clf} as
\begin{equation}
\label{eq:ACC-clf}
L_{f}V(\boldsymbol{x}(t))+L_{g}V(\boldsymbol{x}(t))u(t)+c_{3}V(\boldsymbol{x}(t))\le \delta(t), 
\end{equation}
where $\delta(t)$ is a relaxation that makes \eqref{eq:ACC-clf} a soft constraint.

To satisfy the constraints on safety distance and acceleration, we will define a continuous function $b(\boldsymbol{x}(t))=z(t)-l_{p}$ as AVCBF or PACBF to guarantee $b(\boldsymbol{x}(t))\ge 0$ and constraint \eqref{eq:constraint-u}, then formulate all constraints into QP to get the optimal controller. The parameters are $v_{p}=13.89m/s, v_{d}=24m/s, M=1650kg, g=9.81m/s^{2},z(0)=100m, l_{p}=10m, f_{0}=0.1N, f_{1}=5Ns/m, f_{2}=0.25Ns^{2}/m, c_{a}(t)=0.4, \bigtriangleup t=0.1s$.

\subsubsection{Implementation with AVCBFs}
Define $b(\boldsymbol{x}(t))=z(t)-l_{p},$ the minimum relative degree $\underline{m}$ of $b(\boldsymbol{x}(t))$ with respect to dynamics \eqref{eq:ACC-dynamics} is 2. For simplicity, based on Rem. \ref{rem: sufficient-con}, we just introduce one auxiliary function as $\mathcal{A}_{1}(\boldsymbol{x},a_{1}(t))=a_{1}(t).$ We set the desired relative degree of the auxiliary function $m_{a}$ equal to $\underline{m}$.
Motivated by Sec. \ref{sec:AVCBFs}, we define the auxiliary dynamics as
\begin{small}
\begin{equation}
\label{eq:Auxiliary-dynamics1}
\underbrace{\begin{bmatrix}
\dot{a_{1}}(t) \\
\dot{\pi}_{1,2}(t) 
\end{bmatrix}}_{\dot{\boldsymbol{\pi}}_{1}(t)}  
=\underbrace{\begin{bmatrix}
 \pi_{1,2}(t) \\
 0
\end{bmatrix}}_{F_{1}(\boldsymbol{{\pi}}_{1}(t))} 
+ \underbrace{\begin{bmatrix}
  0 \\
  1 
\end{bmatrix}}_{G_{1}(\boldsymbol{{\pi}}_{1}(t))}\nu_{1}(t).
\end{equation}
\end{small}
The HOCBFs for $\mathcal{A}_{1}$ are defined as 
\begin{equation}
\label{eq:SHOCBF-sequence-ACC}
\begin{split}
&\varphi_{1,0}(\boldsymbol{{\pi}}_{1})\coloneqq a_{1},\\
&\varphi_{1,1}(\boldsymbol{{\pi}}_{1})\coloneqq \dot{\varphi}_{1,0}(\boldsymbol{{\pi}}_{1})+l_{1}\varphi_{1,0}(\boldsymbol{{\pi}}_{1}),\\
&\varphi_{1,2}(\boldsymbol{{\pi}}_{1},
\nu_{1})\coloneqq \dot{\varphi}_{1,1}(\boldsymbol{{\pi}}_{1},\nu_{1})+l_{2}\varphi_{1,1}(\boldsymbol{{\pi}}_{1}),
\end{split}
\end{equation}
where $\alpha_{1,1}(\cdot),\alpha_{1,2}(\cdot)$ are defined as linear functions. The AVCBFs are then defined as
\begin{equation}
\label{eq:AVBCBF-sequence-ACC}
\begin{split}
&\psi_{0}(\boldsymbol{x},\boldsymbol{{\pi}}_{1})\coloneqq a_{1}b(\boldsymbol{x}),\\
&\psi_{1}(\boldsymbol{x},\boldsymbol{{\pi}}_{1})\coloneqq \dot{\psi}_{0}(\boldsymbol{x},\boldsymbol{{\pi}}_{1})+k_{1}\psi_{0}(\boldsymbol{x},\boldsymbol{{\pi}}_{1}),\\
&\psi_{2}(\boldsymbol{x},\boldsymbol{{\pi}}_{1},u,
\nu_{1})\coloneqq \dot{\psi}_{1}(\boldsymbol{x},\boldsymbol{{\pi}}_{1},u,
\nu_{1})+k_{2}\psi_{1}(\boldsymbol{x},\boldsymbol{{\pi}}_{1}),
\end{split}
\end{equation}
where $\alpha_{1}(\cdot),\alpha_{2}(\cdot)$ are set as linear functions. By formulating constraints from HOCBFs \eqref{eq:SHOCBF-sequence-ACC}, AVCBFs \eqref{eq:AVBCBF-sequence-ACC}, CLF \eqref{eq:ACC-clf} and acceleration \eqref{eq:constraint-u}, we can define cost function 
 for QP as
 \begin{small}
\begin{equation}
\label{eq:AVBCBF-cost}
\begin{split}
\min_{u(t),\nu_{1}(t),\delta(t)} \int_{0}^{T}[(\frac{u(t)-F_{r}(v(t))}{M})^{2}\\+W_{1}(\nu_{1}(t)-a_{1,w})^{2}+Q\delta(t)^{2}]dt.
\end{split}
\end{equation}
\end{small}
Other parameters are set as $v(0)=6m/s, a_{1}(0)=1, \pi_{1,2}(0)=1, c_{3}=2, W_{1}=Q=1000,\epsilon=10^{-10}.$
\begin{figure*}[t]
    \vspace{3mm}
    \centering
    \begin{subfigure}[t]{0.32\linewidth}
        \centering
        \includegraphics[width=1\linewidth]{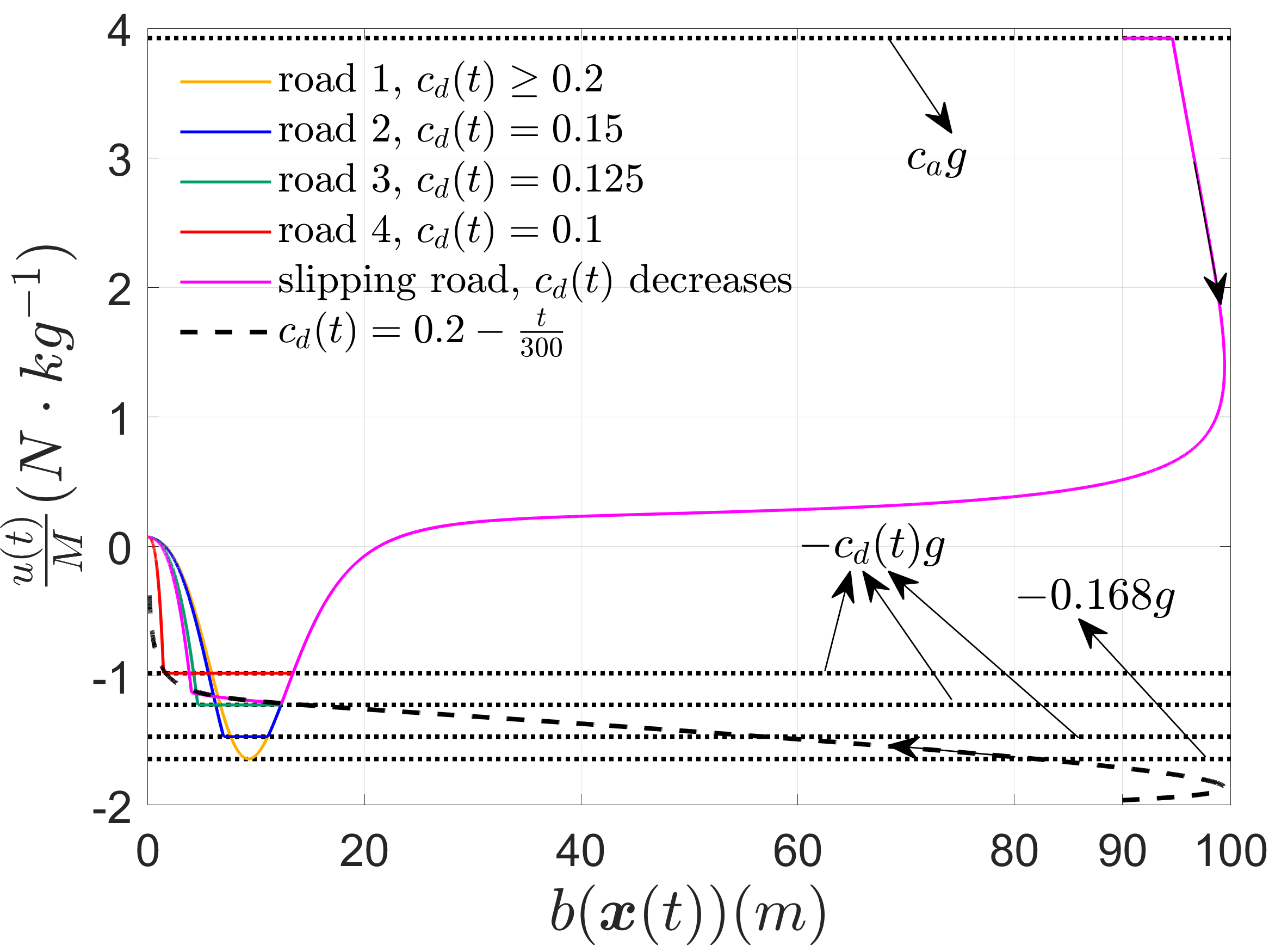}
        \caption{AVCBFs: $k_{1}=k_{2}=l_{1}=l_{2}=0.1, a_{1,w}=1$.}
        \label{fig:AVBCBFs-braking}
    \end{subfigure}
    \begin{subfigure}[t]{0.32\linewidth}
        \centering
        \includegraphics[width=1\linewidth]{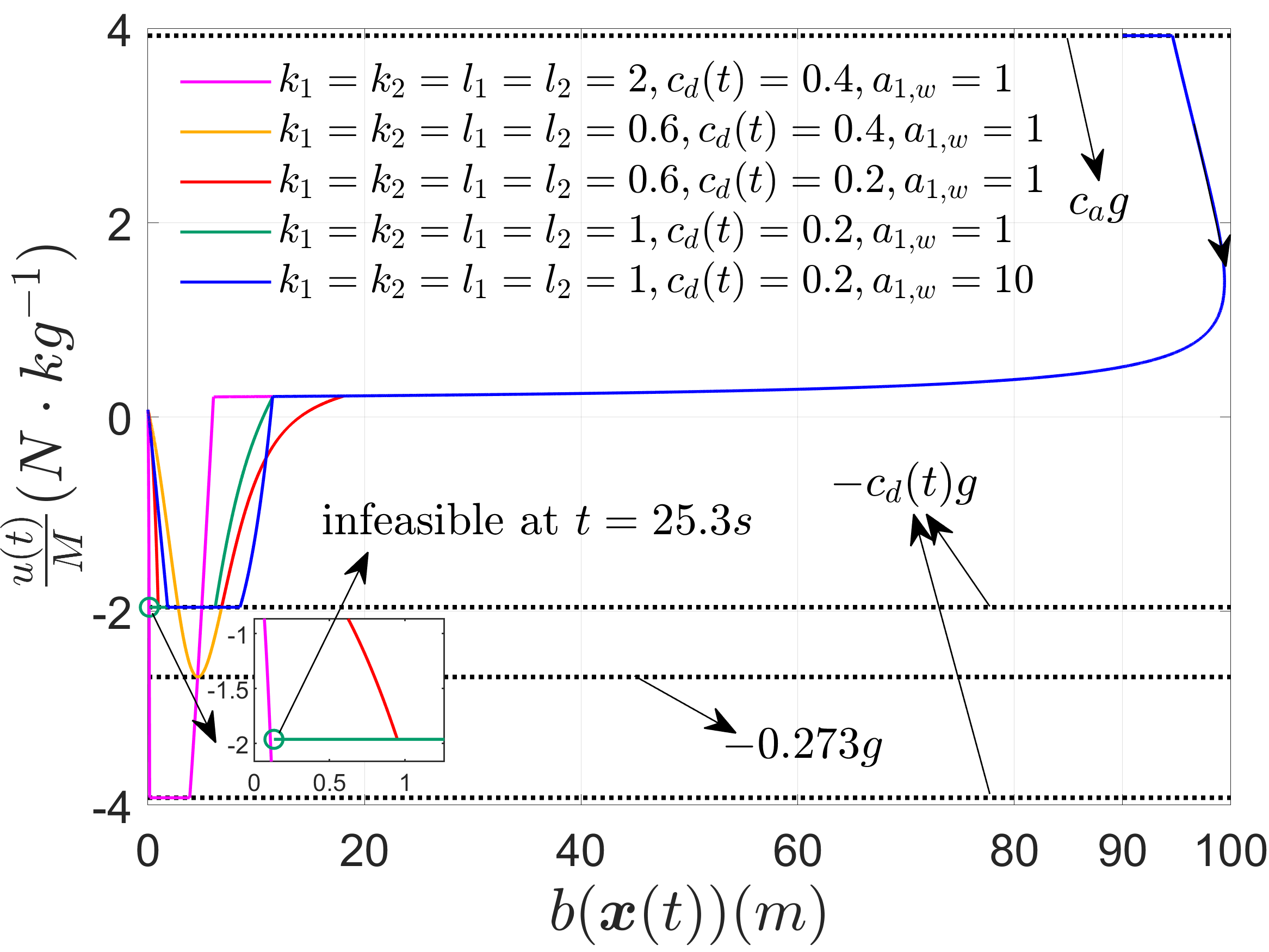}
        \caption{AVCBFs.}
        \label{fig:AVBCBFs-hyperparameters}
    \end{subfigure}  
    \begin{subfigure}[t]{0.32\linewidth}
        \centering
        \includegraphics[width=1\linewidth]{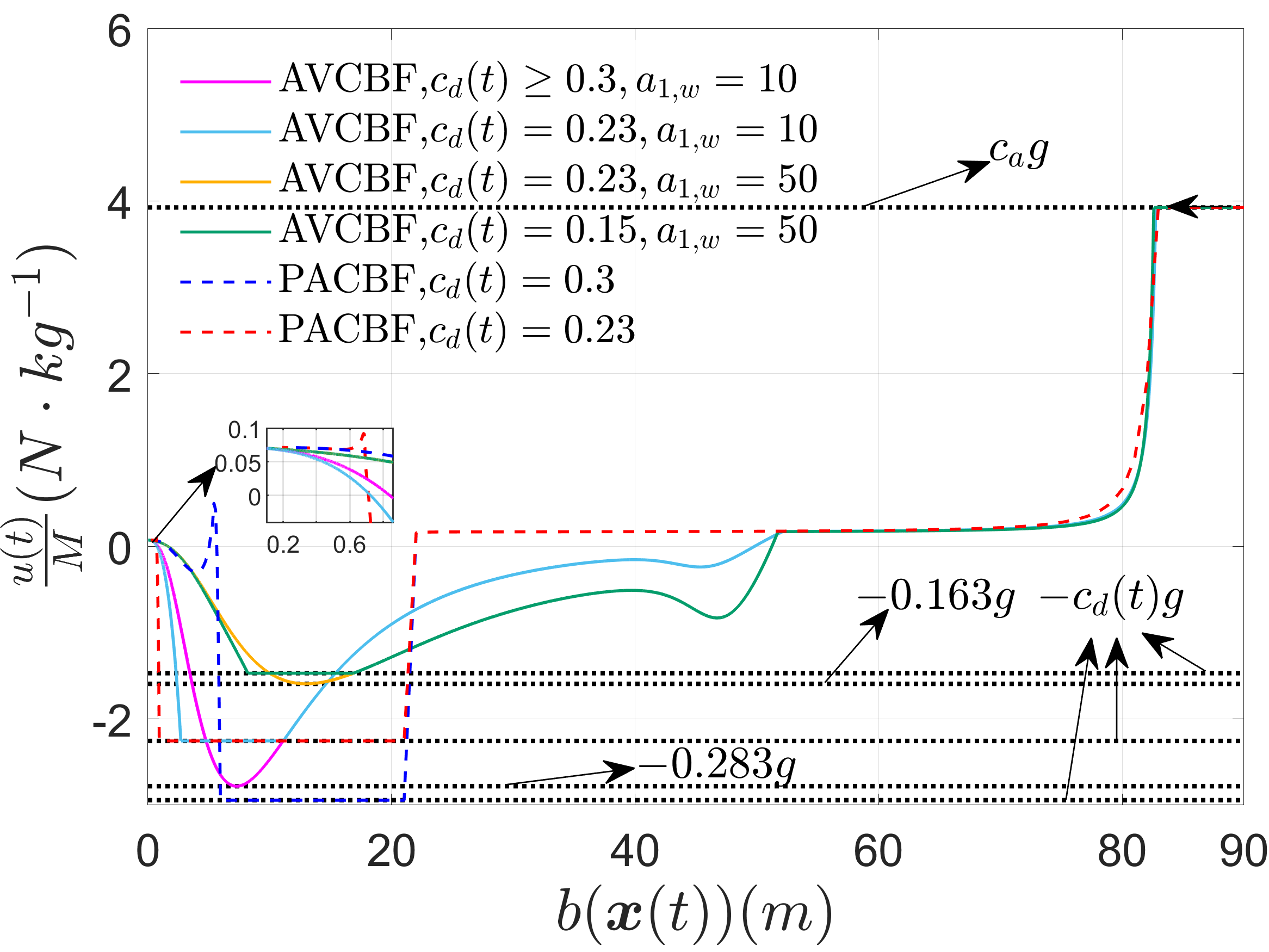}
        \caption{AVCBFs (solid curves) versus PACBFs (dashed curves)}
        \label{fig:AVBCBFs-PACBFs-1}
    \end{subfigure}
    \caption{Control input $u(t)$ varies as $b(\boldsymbol{x}(t))$ goes to 0 under different lower control bounds. The arrows denote the changing trend for $b(\boldsymbol{x}(t))$ and $c_{d}(t)$ over 50 seconds. $b(\boldsymbol{x}(0))=90$ and $b(\boldsymbol{x}(t))\ge 0$ implies safety. Different sets of hyperparameters are tested.
    } 
\end{figure*}

We evaluate adaptivity to deceleration by varying the lower control bound $-c_{d}(t)Mg$. In Fig. \ref{fig:AVBCBFs-braking}, the deceleration coefficient 
 $c_{d}(t)$ is set as a constant or a linearly decreasing function, reflecting different road conditions. In each case, the ego vehicle accelerates to the same velocity
($b(\boldsymbol{x}(t))$ reaches the same value) before decelerating simultaneously. Due to varying braking capabilities, it achieves different maximum decelerations (denoted by arrows) and eventually keeps a constant velocity the same as $v_{p}$ while maintaining the safe distance $l_{p}$ for all $t\in[0,50s]$. Notably, under poor road conditions (e.g., slippery or uneven surfaces, shown by red and magenta curves), QPs remain feasible using the AVCBFs method, demonstrating strong adaptability to control constraints.

We evaluate adaptivity to control strategy conservativeness by varying the hyperparameters 
    $k_{1},k_{2},l_{1},l_{2}$ inside the class $\kappa$ functions in \eqref{eq:SHOCBF-sequence-ACC},\eqref{eq:AVBCBF-sequence-ACC}, along with $c_{d}(t)$. For each case in Fig. \ref{fig:AVBCBFs-hyperparameters}, the ego vehicle accelerates to the same velocity ($b(\boldsymbol{x}(t))$ reaches the same value) before decelerating at different times. The orange and red curves show the earliest braking, while the green and blue curves brake later, and the magenta curve brakes latest. Different braking capabilities lead to varying maximum decelerations (denoted by arrows), after which the vehicle finally keeps a constant velocity the same as $v_{p}$, maintaining the safe distance $l_{p}$ for all $t\in[0,50s]$. Comparing the magenta and orange curves, larger hyperparameters allow later braking and greater deceleration. However, if $c_{d}(t)$ is very small (e.g., the green curve), large hyperparameters cause QP infeasibility (e.g., infeasibility at 25.3s due to insufficient braking distance). Adjusting $a_{1,w}$ in cost function \eqref{eq:AVBCBF-cost} enables faster braking (blue curve), demonstrating AVCBFs’ adaptability in managing control strategies.

\subsubsection{Implementation with PACBFs}

Similar to AVCBFs, we define $b(\boldsymbol{x}(t))=z(t)-l_{p}$ for PACBFs. We use the same penalty function and auxiliary dynamics in \cite{xiao2021adaptive} as $\boldsymbol{p}(t)=(p_{1}(t), p_{2}(t))$ and
\begin{equation}
\label{eq:AVBCBFs-PACBFs-1}
\begin{split}
\dot{p_{1}}(t)=\nu_{1}(t), p_{2}(t)=\nu_{2}(t).
\end{split}
\end{equation}
To make $p_{1}(t)$ converge to a small enough value, we define CLF constraint as
\begin{equation}
\label{eq:HOCBFs-CLFs}
\begin{split}
2(p_{1}(t)-p_{1}^{\ast })\nu_{1}(t)+\rho (p_{1}(t)-p_{1}^{\ast})^{2}\le \delta_{p}(t),
\end{split}
\end{equation}
where $\rho=10, p_{1}^{\ast}=0.103, \delta_{p}(t)$ is relaxed variable. We define HOCBF constraints for $p_{1}(t)$ as 
\begin{equation}
\label{eq:Auxiliary-PACBFs}
\begin{split}
-\nu_{1}(t)+(3-p_{1}(t))\ge0, \nu_{1}(t)+p_{1}(t)\ge0,
\end{split}
\end{equation}
which confines $p_{1}(t)$ into $[0,3].$ The PACBFs are then defined as 
\begin{equation}
\label{eq:AVBCBFs-PACBFs-2}
\begin{split}
&\psi_{0}(\boldsymbol{x})\coloneqq b(\boldsymbol{x}),\\
&\psi_{1}(\boldsymbol{x},\boldsymbol{p})\coloneqq \dot{\psi}_{0}(\boldsymbol{x})+p_{1}\psi_{0}(\boldsymbol{x})^{2},\\
&\psi_{2}(\boldsymbol{x},\boldsymbol{p},\boldsymbol{\nu})\coloneqq \dot{\psi}_{1}(\boldsymbol{x},\boldsymbol{p},\nu_{1})+\nu_{2}\psi_{1}(\boldsymbol{x},\boldsymbol{p})
\end{split}
\end{equation}
to guarantee the safety. By formulating constraints from HOCBFs \eqref{eq:Auxiliary-PACBFs}, CLFs \eqref{eq:HOCBFs-CLFs}\eqref{eq:ACC-clf}, PACBFs \eqref{eq:AVBCBFs-PACBFs-2} and acceleration \eqref{eq:constraint-u}, we can define cost function 
 for QP as
 \begin{small}
\begin{equation}
\label{eq:PACBF-cost}
\begin{split}
\min_{u(t),\nu_{1}(t),\nu_{2}(t),\delta(t),\delta_{p}(t)} \int_{0}^{T}[(\frac{u(t)-F_{r}(v(t))}{M})^{2}+W_{1}\nu_{1}(t)\\+W_{2}(\nu_{2}(t)-1)^{2}+Q\delta(t)^{2}+Q_{p}\delta_{p}(t)^{2}]dt. 
\end{split}
\end{equation}
\end{small}
Other parameters are set as $p_{1}(0)=0.103, p_{2}(0)=1, c_{3}=10, W_{1}=W_{2}=2e^{12},Q=Q_{p}=1.$
\subsubsection{Comparison between AVCBFs and PACBFs}
% \vspace{-0.5cm}

We compare AVCBFs with the state-of-the-art PACBFs in an urgent braking scenario, setting the initial velocity to $v(0) = 20 m/s$. In Fig.~\ref{fig:AVBCBFs-PACBFs-1}, we vary the lower control bound $-c_d(t)Mg$ to assess adaptivity. The magenta and cyan curves use CLF hyperparameter $c_3 = 70$, while the orange and green curves use $c_3 = 100$. Other AVCBF parameters are set as $k_1 = k_2 = l_1 = l_2 = 0.1, W_1 = 2e^5, Q = 7e^5$. We analyze cyan, orange, and red curves where $c_d(t) = 0.23$ further in Fig.~\ref{fig:AVBCBFs-PACBFs-2}. From Figs.~\ref{fig:AVBCBFs-PACBFs-1} and \ref{fig:AVBCBFs-PACBFs-2}, both methods accelerate the ego vehicle to $24 m/s$ before decelerating to $v_p$. While both work well, AVCBFs adjust $\frac{\dot{a}_{1}(t)}{a_{1}(t)}$ more effectively by simply tuning $a_{1,w}$, enabling earlier braking under stricter deceleration limits (e.g., smaller $c_d(t)$ for slippery roads). In contrast, at $t=5s$, PACBFs struggle with adaptivity, as hyperparameters must ensure $p_1(t)$ remains sufficiently small due to additional CLF constraints~\eqref{eq:HOCBFs-CLFs}. Additionally, AVCBFs generate a smoother optimal controller, as seen in Fig.~\ref{fig:AVBCBFs-PACBFs-1}, where PACBFs (blue and red curves) exhibit overshooting near the end, compromising smoothness.

\begin{figure}[ht]
    \centering
    \includegraphics[scale=0.58]{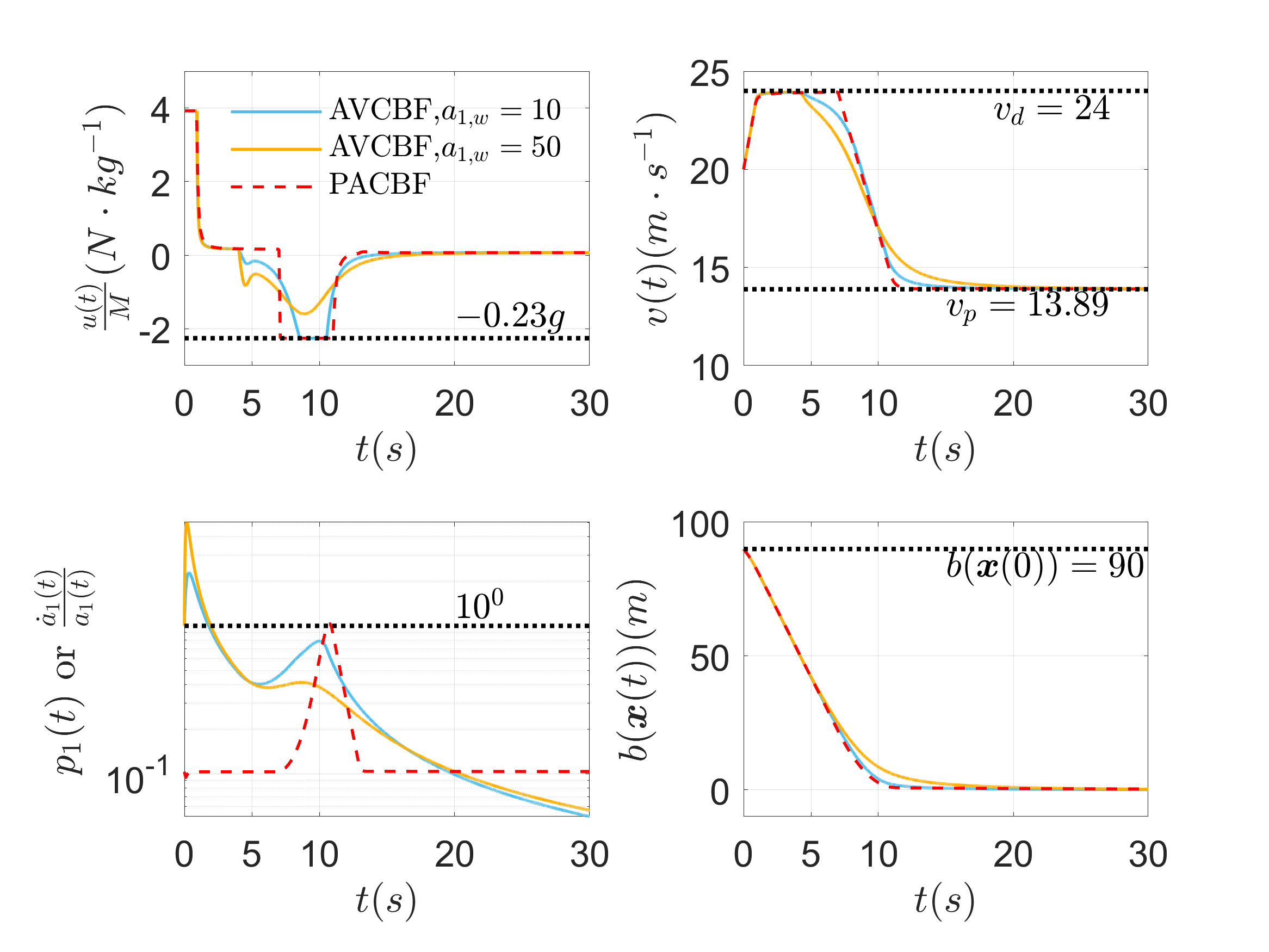}
    \caption{Control input $u(t)$, velocity $v(t)$, time-varying $p_{1}(t), \frac{\dot{a}_{1}(t)}{a_{1}(t)}$ and distance between two vehicles $b(\boldsymbol{x}(t))$ over 30 seconds for AVCBFs and PACBFs. $b(\boldsymbol{x}(t))\ge 0$ implies safety. Solid curves denote AVCBFs and dashed curve denotes PACBFs.}
    \label{fig:AVBCBFs-PACBFs-2}
\end{figure} 

\subsubsection{Implementation with AVCBFs with Reduced Relative Degree}
Define $b(\boldsymbol{x}(t))=z(t)-l_{p},$ although the minimum relative degree $\underline{m}$ of $b(\boldsymbol{x}(t))$ with respect to dynamics \eqref{eq:ACC-dynamics} is 2, based on Rem. \ref{rem:reduced degree}, we can design auxiliary function as $\mathcal{A}_{1}(\boldsymbol{x},a_{1}(t))=e^{-\frac{a_{1}(t)}{v}}$. Since the control input appears after differentiating $\mathcal{A}_{1}$ once, the desired relative degree of the auxiliary function $m_{a}$ is set to 1. We define the auxiliary dynamics as
\begin{small}
\begin{equation}
\label{eq:Auxiliary-dynamics2}
\underbrace{
\dot{a}_{1}(t) \\
}_{\dot{\boldsymbol{\pi}}_{1}(t)}  
=\underbrace{
 0 \\
}_{F_{1}(\boldsymbol{{\pi}}_{1}(t))} 
+ \underbrace{
  1 \\ 
}_{G_{1}(\boldsymbol{{\pi}}_{1}(t))}\nu_{1}(t).
\end{equation}
\end{small}
The HOCBFs for $\mathcal{A}_{1}$ are not needed in this case since $\mathcal{A}_{1}$ is always positive. The AVCBFs are then defined as
\begin{equation}
\label{eq:AVBCBF-sequence-ACC-2}
\begin{split}
&\psi_{0}(\boldsymbol{x},\boldsymbol{{\pi}}_{1})\coloneqq e^{-\frac{a_{1}}{v}}b(\boldsymbol{x}),\\
&\psi_{1}(\boldsymbol{x},\boldsymbol{{\pi}}_{1})\coloneqq \dot{\psi}_{0}(\boldsymbol{x},\boldsymbol{{\pi}}_{1})+k_{1}\psi_{0}(\boldsymbol{x},\boldsymbol{{\pi}}_{1}),
\end{split}
\end{equation}
where $\alpha_{1}(\cdot)$ is set as a linear function. By formulating constraints from AVCBFs \eqref{eq:AVBCBF-sequence-ACC-2}, CLF \eqref{eq:ACC-clf} and acceleration \eqref{eq:constraint-u}, we can define the cost function for QP in the same form as \eqref{eq:AVBCBF-cost}. Other parameters are set as $v(0)=20m/s, a_{1}(0)=-30, c_{3}=120, W_1 = 1e^5 (\text{if}~v(t) > v_p), Q = 2e^4 (\text{if}~v(t) > v_p), W_1 = \frac{1}{30} (\text{if}~v(t) \le v_p), Q = \frac{1}{150} (\text{if}~v(t) \le v_p)$. Similar to Fig. \ref{fig:AVBCBFs-hyperparameters}, we demonstrate that the adaptivity of AVCBF with a reduced relative degree improves the conservativeness of the control strategy by adjusting the hyperparameters $k_{1}$ of the class $\kappa$ function in Fig. \ref{fig:AVBCBFs-reduced degree}. Even with varying braking capability $c_{d}(t)$, the ego vehicle can always maintain a safe distance from the lead vehicle, as 
$b(\boldsymbol{x})$ remains positive at all times. Moreover, adjusting 
$a_{1,w}$ enables even earlier braking to prevent infeasibility (see the red and orange curves). Note that due to the absence of one additional order of constraint, unlike Fig. \ref{fig:AVBCBFs-PACBFs-1}, the ego vehicle's velocity $v$ will ultimately remain lower than that of the lead vehicle $v_{p}$ as shown by $\frac{u(30)}{M}= \frac{F_{r}(v(30))}{M}<\frac{F_{r}(v_{p})}{M}$. Consequently, an increase in the distance between the two vehicles $b(\boldsymbol{x})$ can be observed over time. Meanwhile, the value of $k_{1}$ exhibits a clear correlation with the ego vehicle's braking time (i.e., the smaller the value of $k_{1}$, the earlier the vehicle brakes), and small nonsmoothness can be observed in the red and cyan curves. We conclude that extending AVCBF to higher-order constraints in this case is not necessary to ensure safety and feasibility. However, higher-order constraints can be beneficial for improving the smoothness of the controller.
\begin{figure}[ht]
    \centering
    \includegraphics[scale=0.48]{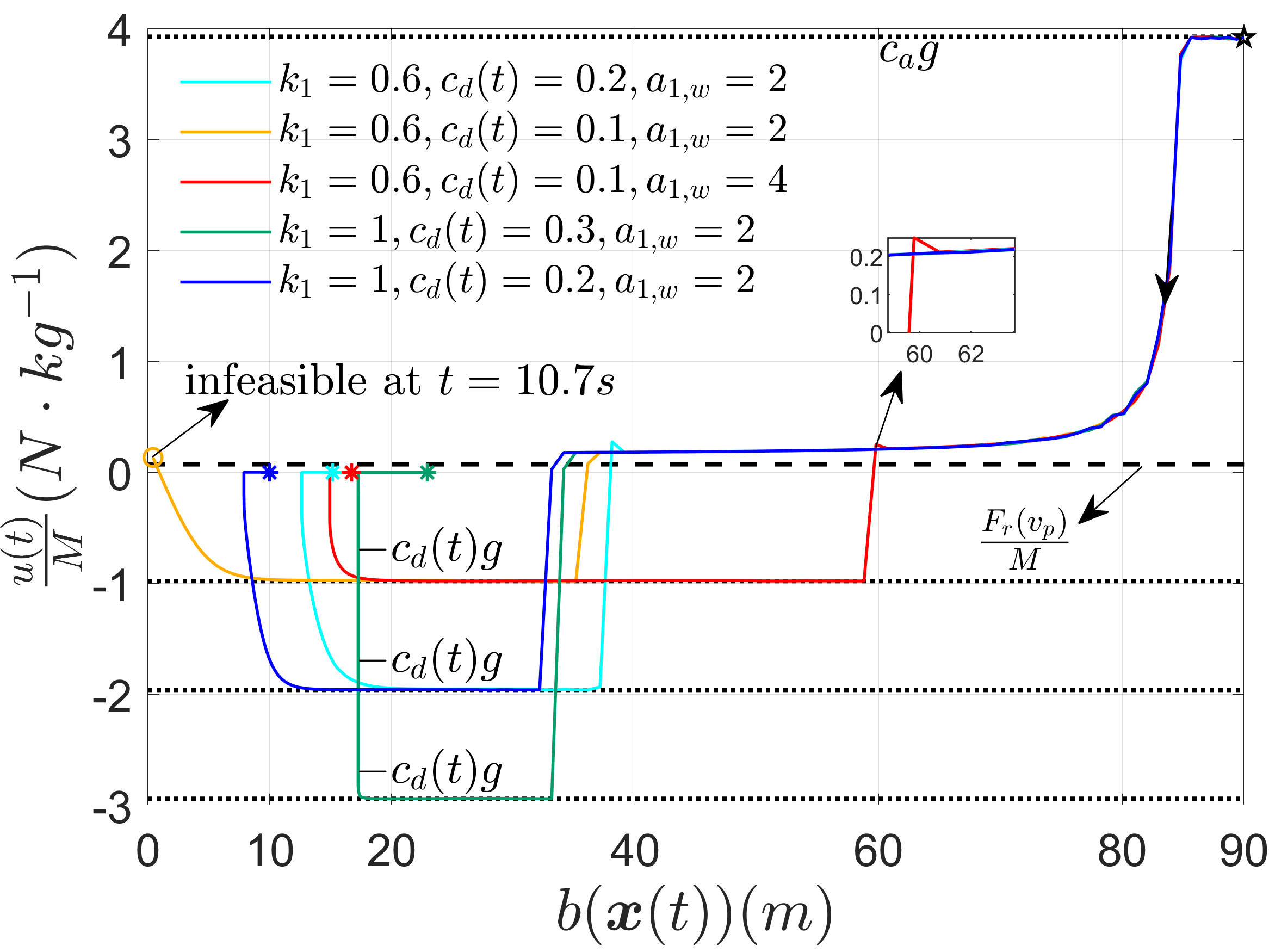}
    \caption{Control input $u(t)$ varies with $b(\boldsymbol{x}(0))$ under different lower control bounds. The arrows denote the changing trend for $b(\boldsymbol{x}(t))$ and $c_{d}(t)$ over 30 seconds. $b(\boldsymbol{x}(0))=90$ and $b(\boldsymbol{x}(t))\ge 0$ implies safety. Different sets of hyperparameters are tested. }
    \label{fig:AVBCBFs-reduced degree}
\end{figure} 

\subsection{Obstacle Avoidance for High-Order Relative Degree Systems}
\subsubsection{Vehicle Dynamics}
\label{subsubsec:Vehicle Dynamics-2}
We consider a nonlinear vehicle dynamics in the form
\begin{small}
\begin{equation}
\label{eq:UM-dynamics2}
\underbrace{\begin{bmatrix}
\dot{x}(t) \\
\dot{y}(t) \\
\dot{\theta}(t)\\
\dot{v}(t)
\end{bmatrix}}_{\dot{\boldsymbol{x}}(t)}  
=\underbrace{\begin{bmatrix}
 v(t)\cos{(\theta(t))}  \\
 v(t)\sin{(\theta(t))} \\
 0 \\
 0
\end{bmatrix}}_{f(\boldsymbol{x}(t))} 
+ \underbrace{\begin{bmatrix}
  0 & 0\\
  0 & 0\\
  1 & 0\\
  0 & \frac{1}{M} 
\end{bmatrix}}_{g(\boldsymbol{x}(t))}\underbrace{\begin{bmatrix}
   u_{1}(t)   \\
  u_{2}(t) 
\end{bmatrix}}_{\boldsymbol{u}(t)}.
\end{equation}
\end{small}
In \eqref{eq:UM-dynamics2}, $M$ denotes the mass of the vehicle and $(x, y)$ denote the coordinates of the unicycle, $v$ is its linear speed, $\theta$ denotes the heading angle, and $\boldsymbol{u}$ represent the angular velocity ($u_{1}$) and the driven force ($u_{2}$), respectively. 

\subsubsection{Vehicle Limitations}
\label{subsubsec: vehicle limitations2}
Vehicle limitations include vehicle constraints on safe distance, reachability and control input.

\textbf{Safe distance constraint:} We consider the case when the vehicle has to avoid one circular obstacle. The distance is considered safe
if $(x-x_{o})^{2}+(y-y_{o})^{2}-r_{o}^{2}\ge 0$, where $(x_{o},y_{o})$ and $r_{o}$ denote the obstacle center location and radius, respectively.

\textbf{Reachability constraint:} We consider a scenario where the vehicle aims to eventually reach a circular area centered at $(x_{d},y_{d})$ with a radius $r_{d}$.

\textbf{Control input constraint:} The vehicle should minimize the following cost
\begin{small}
\begin{equation}
\label{eq:minimal-u-2}
\min_{\boldsymbol{u}(t)} \int_{0}^{T}[\boldsymbol{u}(t)^{T}\boldsymbol{u}(t)]dt,
\end{equation}
\end{small}
when the input is constrained in the form 
\begin{equation}
\label{eq:constraint-u-2}
\boldsymbol{u}_{\text{min}}\le \boldsymbol{u}(t) \le \boldsymbol{u}_{\text{max}}, \forall t \in [0,T].
\end{equation}
To satisfy the constraint on rechability, we define a CLF $V(\boldsymbol{x}(t)) \coloneqq(\theta(t)-\theta_{d})^{2}$ with $\theta_{d}=atan2(\frac{y_{d}-y(t)}{x_{d}-x(t)}), c_{1}=c_{2}=1$ to stabilize $\theta(t)$ to $\theta_{d}$ and formulate the relaxed constraint in \eqref{eq:clf} as
\begin{equation}
\label{eq:ACC-clf-2}
L_{f}V(\boldsymbol{x}(t))+L_{g}V(\boldsymbol{x}(t))u_{1}(t) +c_{3}V(\boldsymbol{x}(t))\le \delta(t), 
\end{equation}
where $\delta(t)$ is a relaxation that makes \eqref{eq:ACC-clf-2} a soft constraint.

To satisfy the constraints on safety distance and control input, we will define a continuous function $b(\boldsymbol{x}(t))=(x-x_{o})^{2}+(y-y_{o})^{2}-r_{o}^{2}$ as HOCBF or AVCBF to guarantee $b(\boldsymbol{x}(t))\ge 0$ and constraint \eqref{eq:constraint-u-2}, then formulate all constraints into QP to get the optimal controller. The parameters are $v(0)=2m/s, \theta_{0}=0 rad, M=1650kg, (x_{o}, y_{o})= (0, 0), r_{o}=1m, r_{d}=0.1m, \bigtriangleup t=0.1s, u_{1,\text{max}}=-u_{1,\text{min}}=5rad/s, u_{2,\text{max}}=-u_{2,\text{min}}=8250N, T=10s$.

\subsubsection{Implementation with AVCBFs}
Define $b(\boldsymbol{x}(t))=(x-x_{o})^{2}+(y-y_{o})^{2}-r_{o}^{2},$ the minimum relative degree $\underline{m}$ of $b(\boldsymbol{x}(t))$ with respect to dynamics \eqref{eq:UM-dynamics2} is 2. We set the desired relative degree of the auxiliary function $m_{a}$ equal to $\underline{m}$. Based on Rem. \ref{rem: sufficient-con}, we can introduce up to two auxiliary functions as $\mathcal{A}_{1}(\boldsymbol{x},a_{1}(t))=a_{1}(t)$ and $\mathcal{A}_{2}(\boldsymbol{x},a_{2}(t))=a_{2}(t)$.
Motivated by Sec. \ref{sec:AVCBFs}, we define the auxiliary dynamics for $a_{1}(t)$ as
\begin{small}
\begin{equation}
\label{eq:Auxiliary-dynamics3-1}
\underbrace{\begin{bmatrix}
\dot{a}_{1}(t) \\
\dot{\pi}_{1,2}(t) 
\end{bmatrix}}_{\dot{\boldsymbol{\pi}}_{1}(t)}  
=\underbrace{\begin{bmatrix}
 \pi_{1,2}(t) \\
 0
\end{bmatrix}}_{F_{1}(\boldsymbol{{\pi}}_{1}(t))} 
+ \underbrace{\begin{bmatrix}
  0 \\
  1 
\end{bmatrix}}_{G_{1}(\boldsymbol{{\pi}}_{1}(t))}\nu_{1}(t).
\end{equation}
\end{small}
The HOCBFs for $\mathcal{A}_{1}$ are defined as 
\begin{equation}
\label{eq:SHOCBF-sequence-ACC-2}
\begin{split}
&\varphi_{1,0}(\boldsymbol{{\pi}}_{1})\coloneqq a_{1},\\
&\varphi_{1,1}(\boldsymbol{{\pi}}_{1})\coloneqq \dot{\varphi}_{1,0}(\boldsymbol{{\pi}}_{1})+l_{1,1}\varphi_{1,0}(\boldsymbol{{\pi}}_{1}),\\
&\varphi_{1,2}(\boldsymbol{{\pi}}_{1},\nu_{1})\coloneqq \dot{\varphi}_{1,1}(\boldsymbol{{\pi}}_{1},\nu_{1})+l_{1,2}\varphi_{1,1}(\boldsymbol{{\pi}}_{1}),
\end{split}
\end{equation}
where $\alpha_{1,1}(\cdot),\alpha_{1,2}(\cdot)$ are defined as linear functions. We define the auxiliary dynamics for $a_{2}(t)$ as
\begin{small}
\begin{equation}
\label{eq:Auxiliary-dynamics3-2}
\underbrace{
\dot{a}_{2}(t) \\
}_{\dot{\boldsymbol{\pi}}_{2}(t)}  
=\underbrace{
 0 \\
}_{F_{2}(\boldsymbol{{\pi}}_{2}(t))} 
+ \underbrace{
  1 \\ 
}_{G_{2}(\boldsymbol{{\pi}}_{2}(t))}\nu_{2}(t).
\end{equation}
\end{small}
The HOCBFs for $\mathcal{A}_{2}$ are defined as 
\begin{equation}
\label{eq:SHOCBF-sequence-ACC-3}
\begin{split}
&\varphi_{2,0}(\boldsymbol{{\pi}}_{2})\coloneqq a_{2},\\
&\varphi_{2,1}(\boldsymbol{{\pi}}_{2},\nu_{2})\coloneqq \dot{\varphi}_{2,0}(\boldsymbol{{\pi}}_{2},\nu_{2})+l_{2,1}\varphi_{2,0}(\boldsymbol{{\pi}}_{2}),\\
\end{split}
\end{equation}
where $\alpha_{2,1}(\cdot)$ is defined as a linear function. The AVCBFs are then defined as
\begin{equation}
\label{eq:AVBCBF-sequence-ACC-3}
\begin{split}
&\psi_{0}(\boldsymbol{x},\boldsymbol{{\Pi}})\coloneqq a_{1}b(\boldsymbol{x}),\\
&\psi_{1}(\boldsymbol{x},\boldsymbol{{\Pi}})\coloneqq a_{2}(\dot{\psi}_{0}(\boldsymbol{x},\boldsymbol{{\Pi}})+k_{1}\psi_{0}(\boldsymbol{x},\boldsymbol{{\Pi}})),\\
&\psi_{2}(\boldsymbol{x},\boldsymbol{{\Pi}},\boldsymbol{\nu})\coloneqq \dot{\psi}_{1}(\boldsymbol{x},\boldsymbol{{\Pi}},\boldsymbol{\nu})+k_{2}\psi_{1}(\boldsymbol{x},\boldsymbol{{\Pi}}),
\end{split}
\end{equation}
where $\alpha_{1}(\cdot),\alpha_{2}(\cdot)$ are set as linear functions, $\boldsymbol{{\Pi}}=[\boldsymbol{{\pi}}_{1},\boldsymbol{{\pi}}_{2}]^{T}$, $\boldsymbol{{\nu}}=[\nu_{1},\nu_{2}]^{T}$. By formulating constraints from HOCBFs \eqref{eq:SHOCBF-sequence-ACC-2},\eqref{eq:SHOCBF-sequence-ACC-3}, AVCBFs \eqref{eq:AVBCBF-sequence-ACC-3}, CLF \eqref{eq:ACC-clf-2} and control input \eqref{eq:constraint-u-2}, we can define cost function 
 for QP as
 \begin{small}
\begin{equation}
\label{eq:AVBCBF-cost-2}
\begin{split}
\min_{\boldsymbol{u}(t),\boldsymbol{\nu}(t),\delta(t)} \int_{0}^{T}[\boldsymbol{u}(t)^{T}\boldsymbol{u}(t)+W_{1}(\nu_{1}(t)-a_{1,w})^{2}\\
+W_{2}(\nu_{2}(t)-a_{2,w})^{2}+Q\delta(t)^{2}]dt.
\end{split}
\end{equation}
\end{small}
Since these AVCBFs include two auxiliary functions, the above method is referred to as AVCBF-2. Based on Rem. \ref{rem: sufficient-con}, for simplicity, we remove the second auxiliary function $\mathcal{A}_{2}$, which means that \eqref{eq:Auxiliary-dynamics3-2} and \eqref{eq:SHOCBF-sequence-ACC-3} are no longer needed, and the third term in the cost function \eqref{eq:AVBCBF-cost-2} is also eliminated. Since this approach includes only auxiliary function $\mathcal{A}_{1}$ (with manually tuned parameters), we refer to it as AVCBF-1. Additionally, we use the parametrization method from Alg. \ref{alg:parametrization-cbf} to tune the parameter $a_{1,w}$ in AVCBF-1, and we denote this method as AVCBF-P. Other parameters are set as $a_{1}(0)=a_{2}(0)=\pi_{1,2}(0)=0.1, c_{3}=10, W_{1}=W_{2}=1000,Q=10^{5}, \epsilon_{1}=\epsilon_{2}=10^{-10}, J_{m}=10, N_{c}=8, \varepsilon=0.1, \lambda=10, l_{1,1}=l_{1,2}=l_{2,1}=0.1$. 
\subsubsection{Implementation with HOCBFs}
Similar to AVCBFs, we define $b(\boldsymbol{x}(t))=(x-x_{o})^{2}+(y-y_{o})^{2}-r_{o}^{2}$ for HOCBF candidate. The HOCBFs are then defined as
\begin{equation}
\label{eq:HOCBF-benchmark}
\begin{split}
&\psi_{0}(\boldsymbol{x})\coloneqq b(\boldsymbol{x}),\\
&\psi_{1}(\boldsymbol{x})\coloneqq \dot{\psi}_{0}(\boldsymbol{x})+k_{1}\psi_{0}(\boldsymbol{x}),\\
&\psi_{2}(\boldsymbol{x})\coloneqq \dot{\psi}_{1}(\boldsymbol{x})+k_{2}\psi_{1}(\boldsymbol{x}),
\end{split}
\end{equation}
where $\alpha_{1}(\cdot),\alpha_{2}(\cdot)$ are set as linear functions. By formulating constraints from HOCBFs \eqref{eq:HOCBF-benchmark}, CLF \eqref{eq:ACC-clf-2} and control input \eqref{eq:constraint-u-2}, we can define cost function for QP the same as \eqref{eq:optimal control-cost}. Other parameters are set as $c_{3}=10, Q=10^{5}$. 
\subsubsection{Comparison between AVCBFs and HOCBFs}
 We test the adaptivity to conservativeness of control strategy from four different controllers by changing the hyperparameters $k_{1},k_{2}$ inside the class $\kappa$ functions in \eqref{eq:AVBCBF-sequence-ACC-3} and \eqref{eq:HOCBF-benchmark}. Since the initial heading angle is $0$, the vehicle is initially directed toward the center of the obstacle, making this a highly challenging case. The closer the initial location of the vehicle is to the obstacle, the more difficult it becomes to avoid a collision. It can be observed in Fig. \ref{fig:HOCBF-trj} that under the HOCBF-based QP, the optimization becomes infeasible (denoted by the cross symbol) before the vehicle reaches the green target area, regardless of the chosen set of hyperparameters. In contrast, AVCBF-1 enables the vehicle to reach the target area under certain initial location and hyperparameter settings shown in Fig. \ref{fig:AVCBF-1-trj}; however, infeasibility may still arise when hyperparameters are large (e.g., if $k_{1}=k_{2}=10$, the vehicle tends to brake late, which can lead to an aggressive control strategy). If we use AVCBF-P (Fig. \ref{fig:AVCBF-P-trj}) or AVCBF-2 (Fig. \ref{fig:AVCBF-2-trj}), the vehicle will safely reach the target area, regardless of the chosen initial locations or hyperparameters. By comparing AVCBF-1 and AVCBF-2, we observe that the additional auxiliary function further enhances the feasibility of solving the QP (as stated in Rem. \ref{rem: sufficient-con}). Through AVCBF-P, we demonstrate that even with only one auxiliary function, feasibility and safety can still be guaranteed by adjusting the hyperparameter $a_{1,w}(t)$ using the parametrization method. Moreover, the conservativeness of the control strategy can be reduced by AVCBF-2 or AVCBF-P, as the auxiliary function can dynamically influence $k_{1}$ and 
$k_{2}$ regardless of how conservatively they are set.

We selected the four controllers corresponding to the solid magenta curve in Fig. \ref{fig:closed-loop-trj}, and analyzed them in Fig. \ref{fig:benchmark-details} and Fig. \ref{fig:AVBCBFs-control-input} to provide a more detailed illustration of how AVCBF-2 and AVCBF-P ensure feasibility and safety. In Fig. \ref{fig:AVBCBFs-control-input} We see that due to the presence of tight control bounds, solving the QP with HOCBF is infeasible at $t=1.2s$. When we keep the hyperparameters unchanged and replace HOCBF with AVCBF-1, solving the QP at 
$t=1.2s$ remains infeasible. This is because the safety-feasibility criterion $\psi_{1}>0 $ is not satisfied (in Fig. \ref{fig:first-order-candidate}, the $\psi_{1}(\boldsymbol{x}(1.2))$ of AVCBF-1 is negative). In contrast, the controllers under AVCBF-2 and AVCBF-P remain within the control bounds without making the QP infeasible. From Fig. \ref{fig:first-order-candidate}, we can see that this is because the 
$\psi_{1}(\boldsymbol{x}(t))$ of both methods remains positive. From Fig. \ref{fig:avcbf-parameter} and Fig. \ref{fig:avcbf-iteration}, we can see that although AVCBF-P has one fewer auxiliary function compared to AVCBF-2, it starts iteratively updating $a_{1,w}(t)$ from $t=1.2-0.1N_{c}=0.4s$, within the time window $[0.4s, 1.2s]$, ensuring that solving the QP at $t=1.2s$ and future time steps remains feasible. Fig. \ref{fig:cbf-candidate} shows that both AVCBF-2 and AVCBF-P guarantee safety over time. Moreover, AVCBF-P allows the vehicle to get extremely close to obstacles between 2s and 6s. We conclude that increasing the number of auxiliary functions or using the proposed parametrization method for AVCBFs can ensure both safety and feasibility. Additionally, the proposed parametrization method significantly reduces the conservatism of the control strategy.
\begin{figure*}[t]
    \vspace{3mm}
    \centering
    \begin{subfigure}[t]{0.24\linewidth}
        \centering
        \includegraphics[width=1\linewidth]{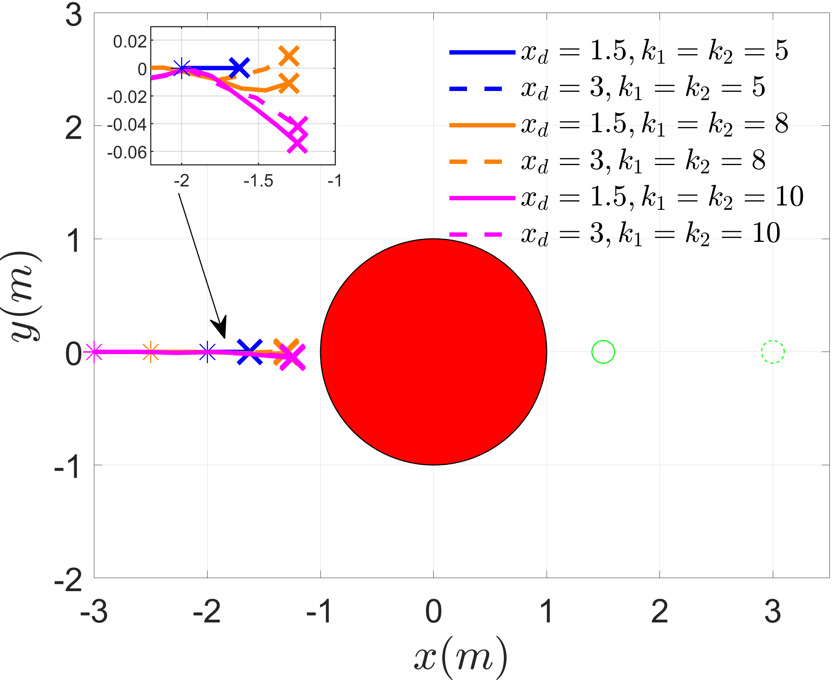}
        \caption{HOCBF.}
        \label{fig:HOCBF-trj}
    \end{subfigure}
    \begin{subfigure}[t]{0.24\linewidth}
        \centering
        \includegraphics[width=1\linewidth]{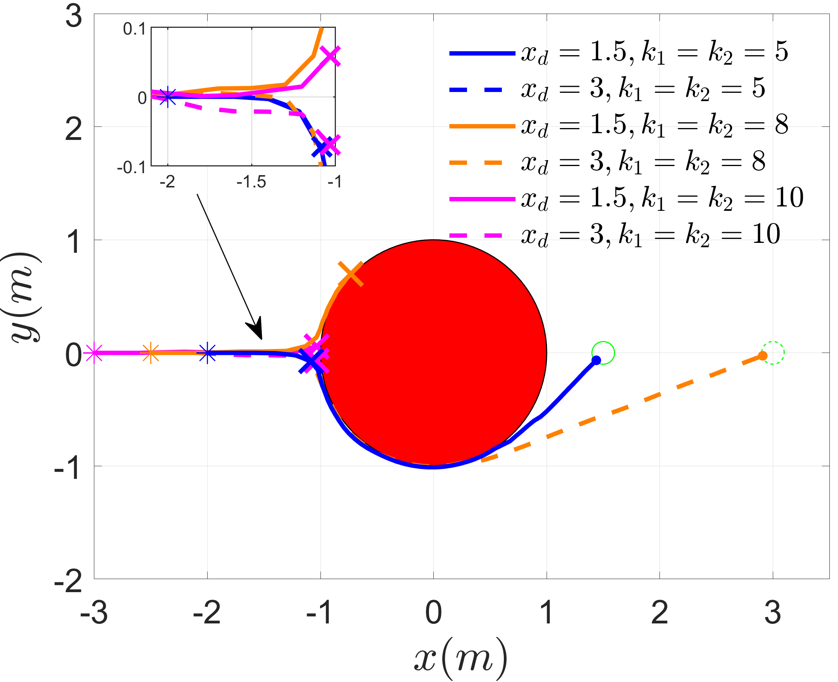}
        \caption{AVCBF-1.}
        \label{fig:AVCBF-1-trj}
    \end{subfigure}  
    \begin{subfigure}[t]{0.24\linewidth}
        \centering
        \includegraphics[width=1\linewidth]{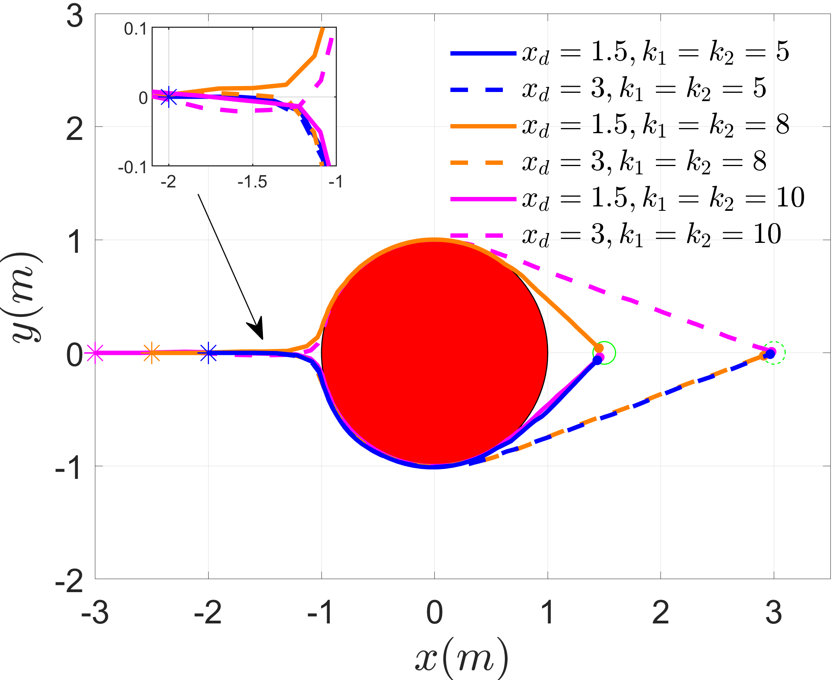}
        \caption{AVCBF-P.}
        \label{fig:AVCBF-P-trj}
    \end{subfigure}
    \begin{subfigure}[t]{0.24\linewidth}
        \centering
        \includegraphics[width=1\linewidth]{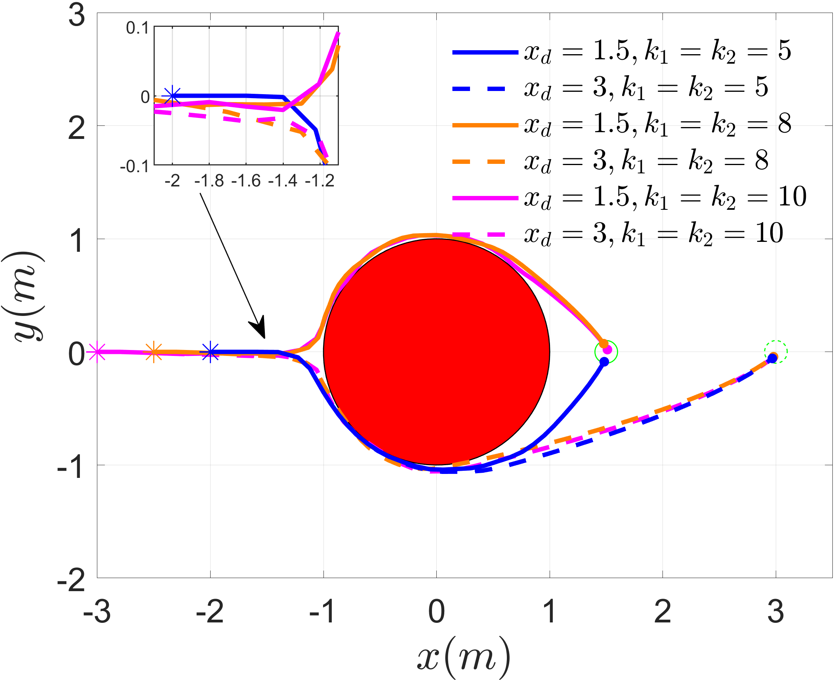}
        \caption{AVCBF-2.}
        \label{fig:AVCBF-2-trj}
    \end{subfigure}
    \caption{Closed-loop trajectories over time with different controllers: several safe closed-loop trajectories starting at the initial locations (depicted by asterisk, magenta: $(x(0), y(0))=(-3,0)$; orange: $(x(0), y(0))=(-2.5,0)$; blue: $(x(0), y(0))=(-2,0)$) terminates within the target areas (depicted by green circle, solid: $(x_{d}, y_{d})=(1.5,0)$; dashed: $(x_{d}, y_{d})=(3,0)$). The cross symbol indicates that the QP is infeasible at this time step. Different sets of hyperparameters are tested. 
    } 
    \label{fig:closed-loop-trj}
\end{figure*}
\begin{figure*}[t]
    \vspace{3mm}
    \centering
    \begin{subfigure}[t]{0.26\linewidth}
        \centering
        \includegraphics[width=1\linewidth]{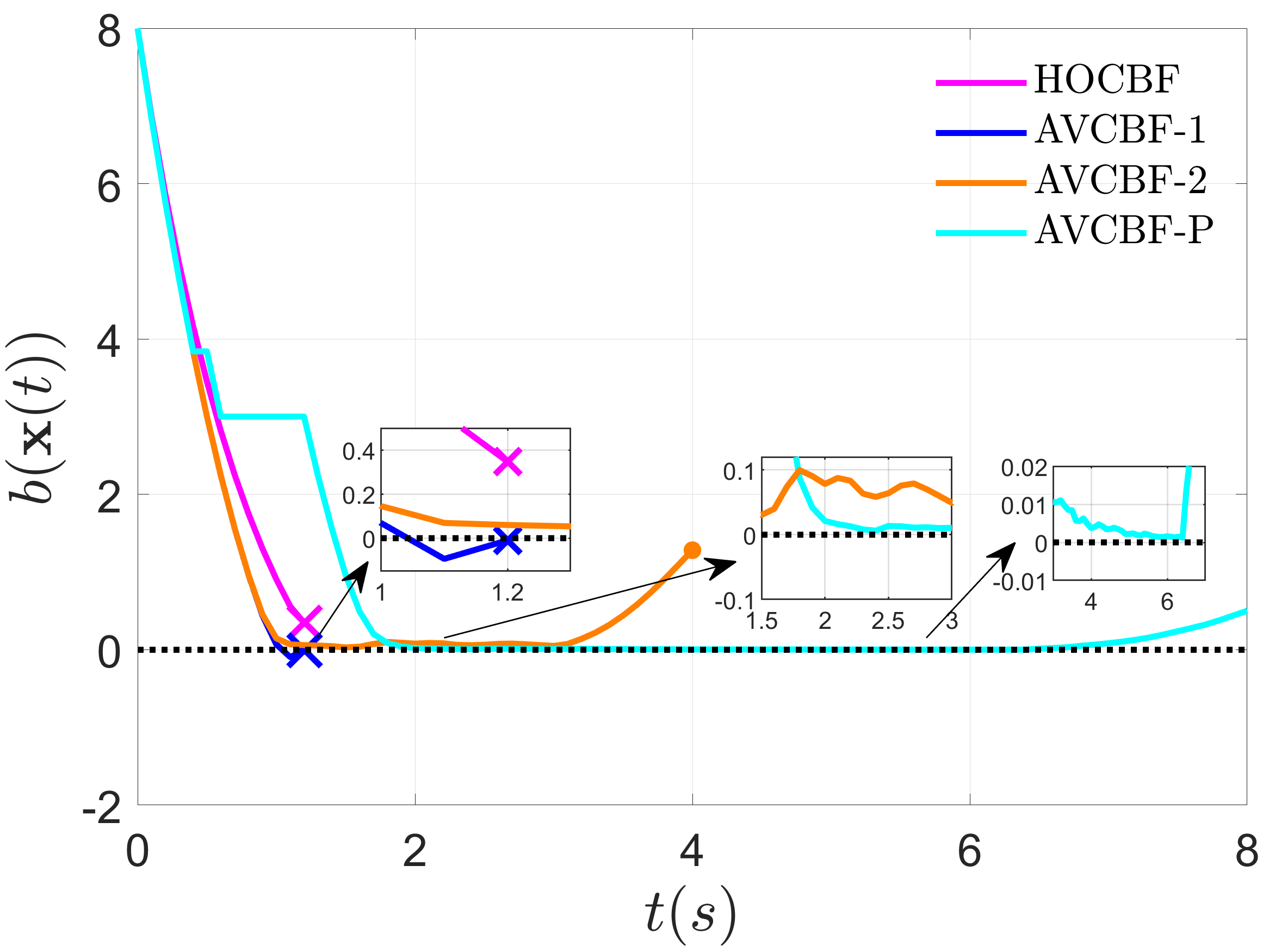}
        \caption{Safe distance over time.}
        \label{fig:cbf-candidate}
    \end{subfigure}
    \begin{subfigure}[t]{0.26\linewidth}
        \centering
        \includegraphics[width=1\linewidth]{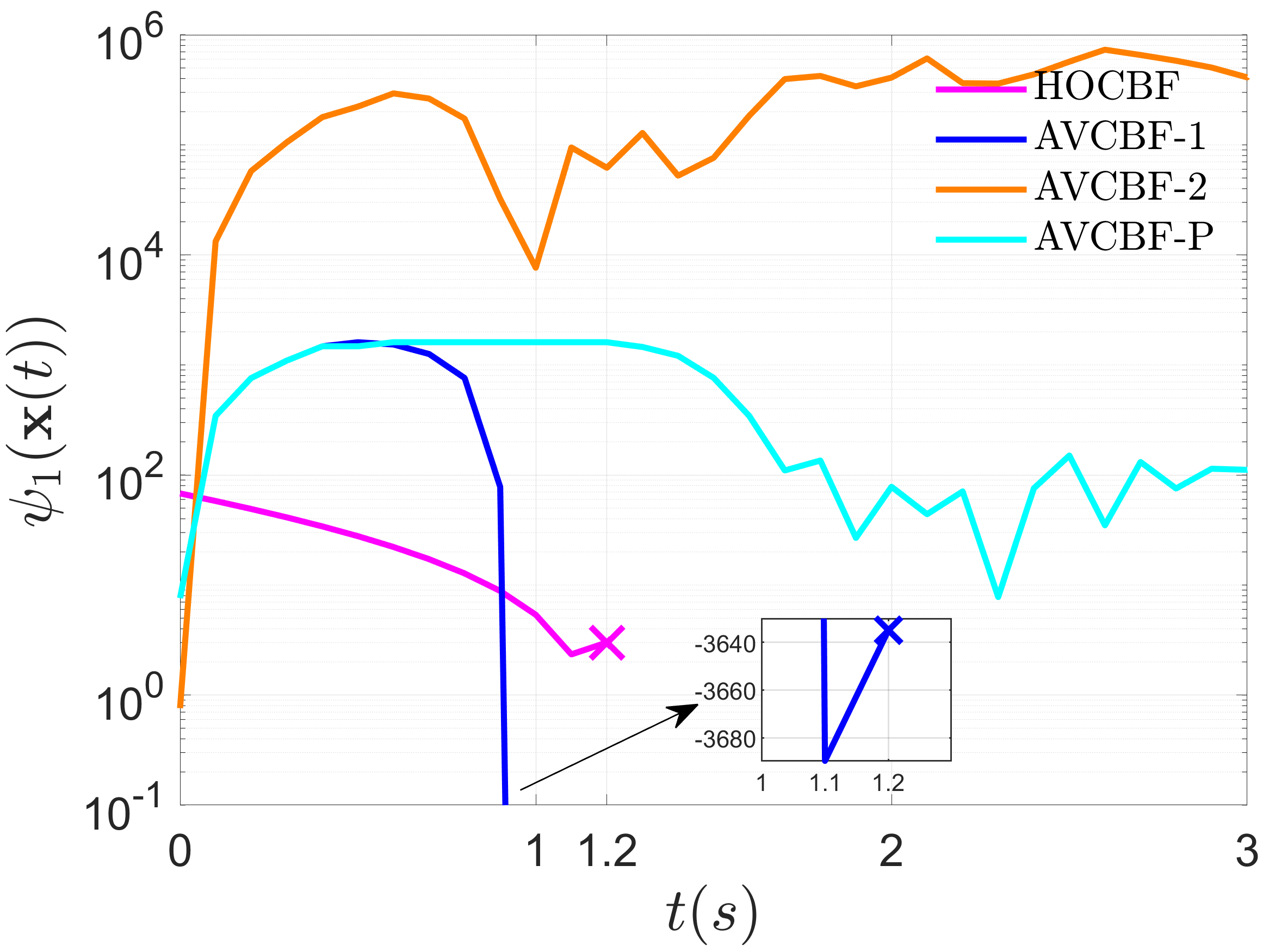}
        \caption{Safety-feasibility criterion check.}
        \label{fig:first-order-candidate}
    \end{subfigure}  
    \begin{subfigure}[t]{0.23\linewidth}
        \centering
        \includegraphics[width=1\linewidth]{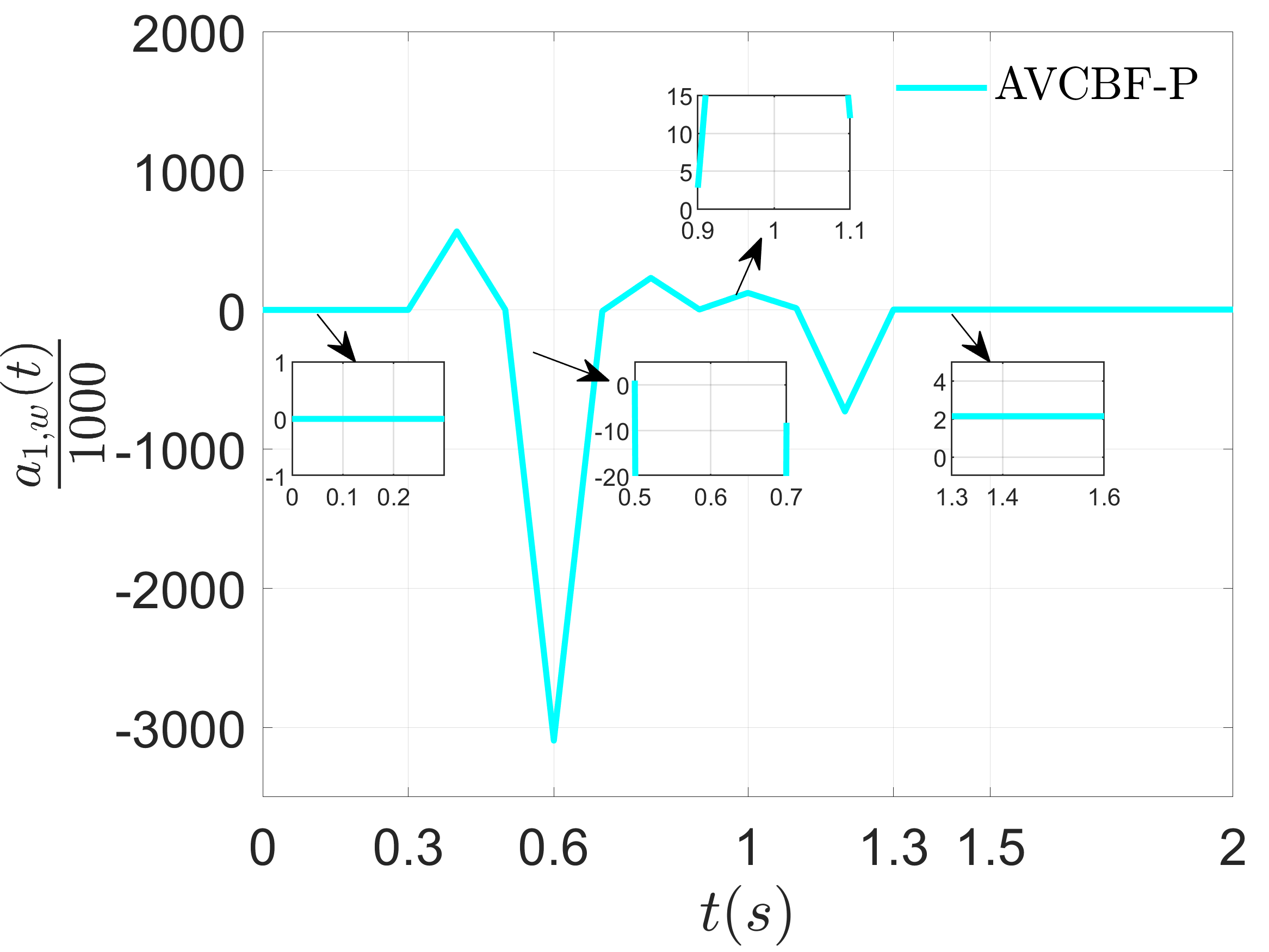}
        \caption{AVCBF-P: hyperparameter tuning over time.}
        \label{fig:avcbf-parameter}
    \end{subfigure}
    \begin{subfigure}[t]{0.23\linewidth}
        \centering
        \includegraphics[width=1\linewidth]{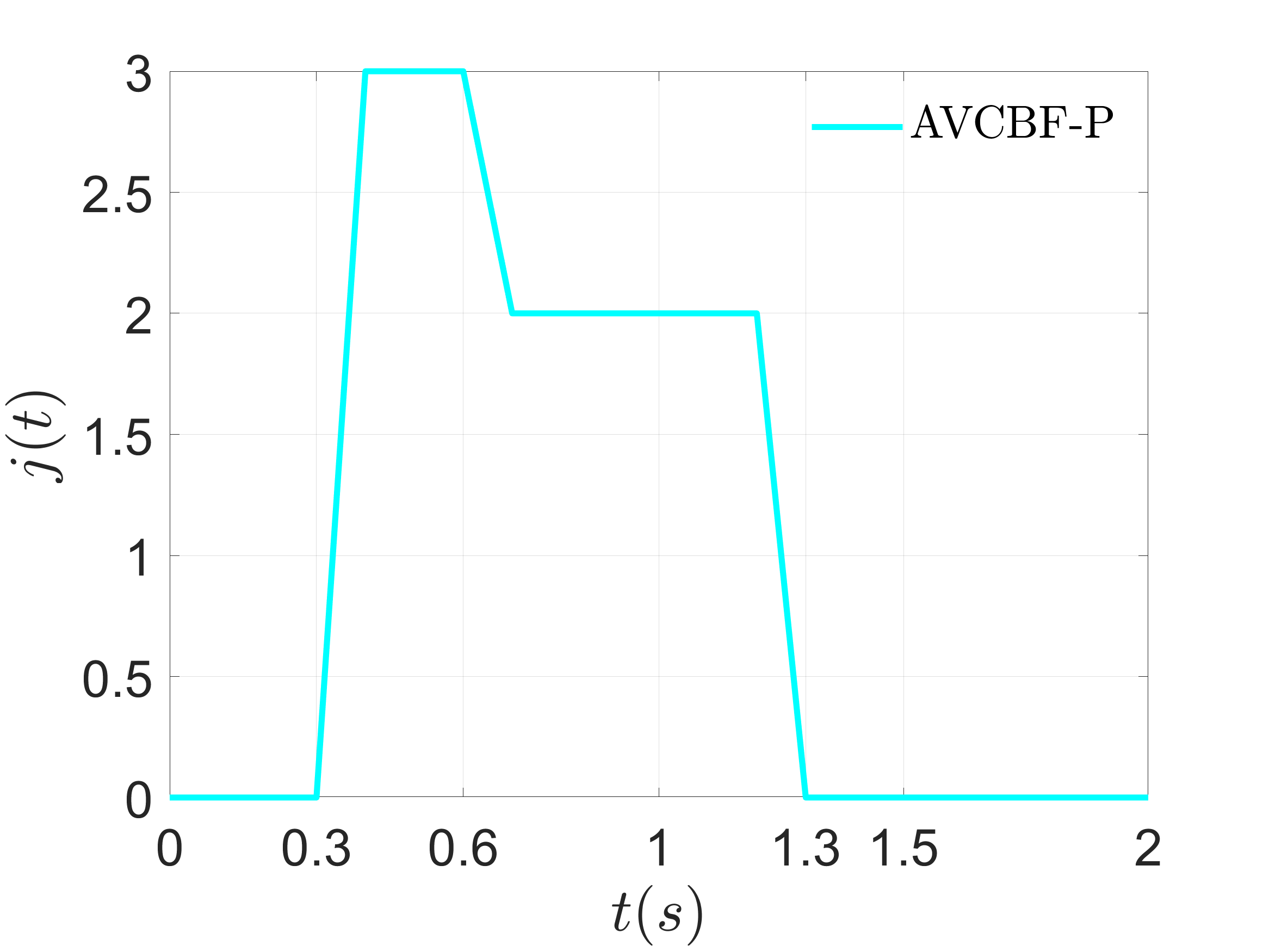}
        \caption{AVCBF-P: total number of iterations over time.}
        \label{fig:avcbf-iteration}
    \end{subfigure}
    \caption{Analysis of safety and feasibility for different controllers (Fig. \ref{fig:cbf-candidate} and Fig. \ref{fig:first-order-candidate}), and hyperparameter adaptation in AVCBF-P (Fig. \ref{fig:avcbf-parameter} and Fig. \ref{fig:avcbf-iteration}), where $x(0)=-3m,x_{d}=1.5m,k_{1}=k_{2}=10, a_{1,w}(0)=0$.
    } 
    \label{fig:benchmark-details}
\end{figure*}

\begin{figure}[ht]
    \centering
    \includegraphics[scale=0.48]{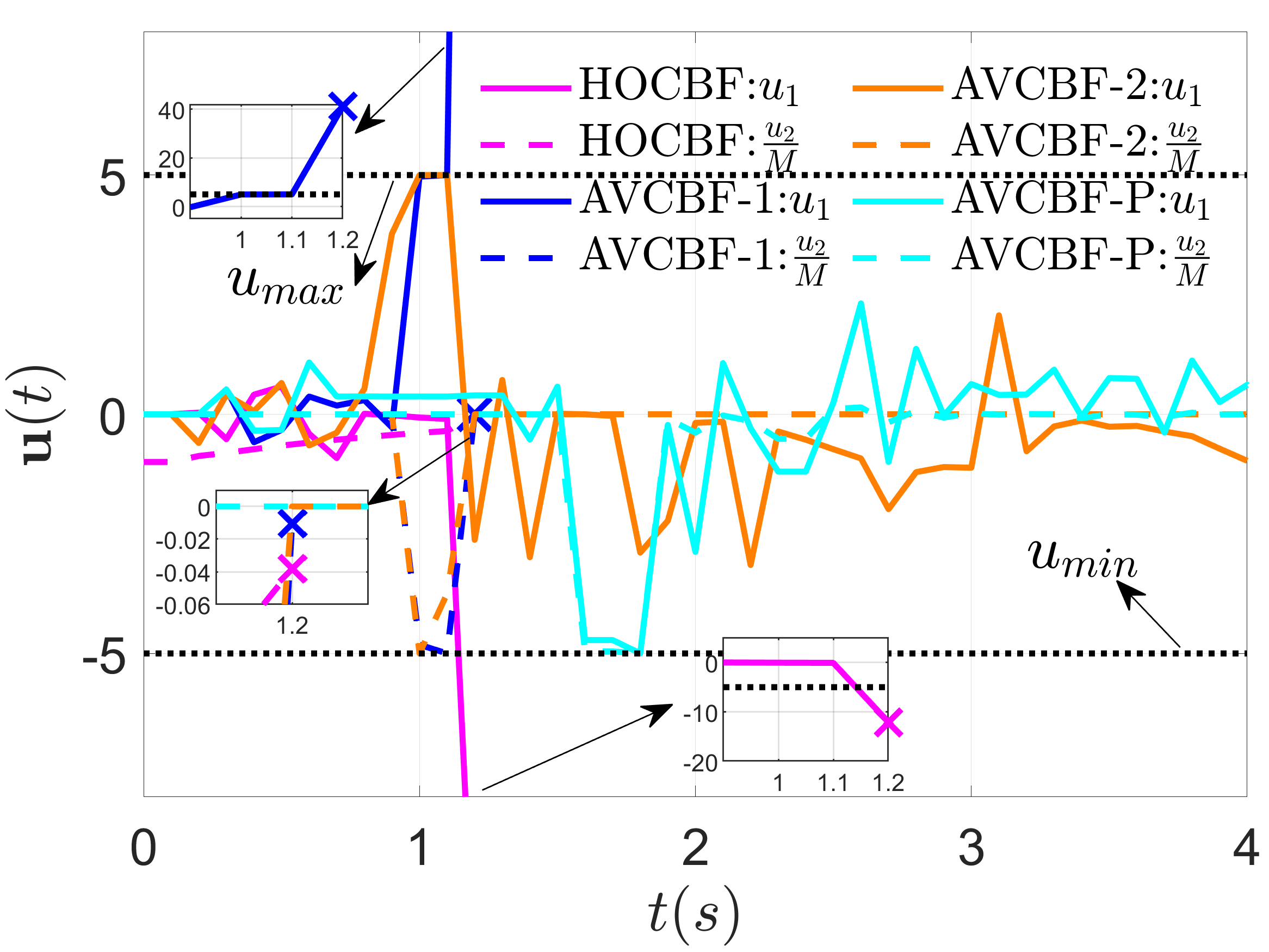}
    \caption{Control input $\boldsymbol{u}(t)$ varies over time with different controllers. The cross symbol indicates that the QP is infeasible at this time step. $x(0)=-3m,x_{d}=1.5m,k_{1}=k_{2}=10, a_{1,w}(0)=0$.}
    \label{fig:AVBCBFs-control-input}
\end{figure} 

\subsubsection{Implementation with AVCBFs with Reduced Relative Degree}
Define $b(\boldsymbol{x}(t))=(x-x_{o})^{2}+(y-y_{o})^{2}-r_{o}^{2}$, although the minimum relative degree $\underline{m}$ of $b(\boldsymbol{x}(t))$ with respect to dynamics \eqref{eq:UM-dynamics2} is 2, based on Rem. \ref{rem:reduced degree}, we can design auxiliary function as $\mathcal{A}_{1}(\boldsymbol{x},a_{1}(t))=a_{1}(t)+v(t)+\theta(t)$. Since the control input appears after differentiating $\mathcal{A}_{1}$ once, the desired relative degree of the auxiliary function $m_{a}$ is set to 1. We define the auxiliary dynamics the same as \eqref{eq:Auxiliary-dynamics2}. The HOCBFs for $\mathcal{A}_{1}$ are defined as 
\begin{small}
\begin{equation}
\label{eq:SHOCBF-sequence-ACC-4}
\begin{split}
&\varphi_{1,0}(\boldsymbol{x},\boldsymbol{{\pi}}_{1})\coloneqq a_{1}+v+\theta,\\
&\varphi_{1,1}(\boldsymbol{x}, \boldsymbol{{\pi}}_{1},\boldsymbol{u},\nu_{1})\coloneqq \dot{\varphi}_{1,0}(\boldsymbol{x}, \boldsymbol{{\pi}}_{1},\boldsymbol{u},\nu_{1})+l_{1,1}\varphi_{1,0}(\boldsymbol{x}, \boldsymbol{{\pi}}_{1}),\\
\end{split}
\end{equation}
\end{small}
where $\alpha_{1,1}(\cdot)$ is defined as a linear function. The AVCBFs are then defined as
\begin{equation}
\label{eq:AVBCBF-sequence-ACC-4}
\begin{split}
&\psi_{0}(\boldsymbol{x},\boldsymbol{{\pi}}_{1})\coloneqq (a_{1}+v+\theta)b(\boldsymbol{x}),\\
&\psi_{1}(\boldsymbol{x}, \boldsymbol{{\pi}}_{1},\boldsymbol{u},\nu_{1})\coloneqq \dot{\psi}_{0}(\boldsymbol{x}, \boldsymbol{{\pi}}_{1},\boldsymbol{u},\nu_{1})+k_{1}\psi_{0}(\boldsymbol{x},\boldsymbol{{\pi}}_{1}),
\end{split}
\end{equation}
where $\alpha_{1}(\cdot)$ is set as a linear function. By formulating constraints from HOCBFs \eqref{eq:SHOCBF-sequence-ACC-4}, AVCBFs \eqref{eq:AVBCBF-sequence-ACC-4}, CLF \eqref{eq:ACC-clf-2} and control input \eqref{eq:constraint-u-2}, we can define the cost function for QP in the same form as \eqref{eq:optimal control-cost new}. We refer to this method as AVCBF-R (reduced relative degree). Additionally, we use the parametrization method from Alg. \ref{alg:parametrization-cbf} to tune the hyperparameter $a_{1,w}$ in AVCBF-R, and we denote this method as AVCBF-R-P. Other parameters are set as $a_{1}(0)=50, c_{3}=10, W_{1}=1000,Q=10^{5}, \epsilon_{1}=10^{-10}, J_{m}=10, N_{c}=6, \varepsilon=0.1, \lambda=10, l_{1,1}=0.5, a_{1,w}(0)=0$. 

We test the adaptivity to conservativeness of control strategy from AVCBF-R and AVCBF-R-P by changing the hyperparameters $k_{1}$ inside the class $\kappa$ function in \eqref{eq:AVBCBF-sequence-ACC-4}. In Fig. \ref{fig:avcbf-reduced-trj}, we can see that starting from $(-3, 0)$, AVCBF-R enables the vehicle to reach the green target areas under certain hyperparameter settings. However, infeasibility (denoted by the cross symbol) arises when hyperparameter is large (e.g., if $k_{1}=3$, the vehicle tends to brake late, which can lead to an aggressive control strategy). In contrast, Fig. \ref{fig:avcbf-reduced-P-trj} shows that AVCBF-R-P enables the vehicle to safely reach same target areas with large hyperparameter ($k_{1}=3$) for any initial location used in AVCBF-R. Moreover, AVCBF-R-P assists the vehicle in starting from a position closer to the obstacle ($x(0)=-2m$) and successfully reaching an area closer to the obstacle ($x_{d}=1.2m$). We extract the infeasible solid blue and solid orange curves from Fig. \ref{fig:avcbf-reduced-trj} and the feasible solid magenta curve from Fig. \ref{fig:avcbf-reduced-P-trj} and compare them in Figs. \ref{fig:cbf-candidate-2} and \ref{fig:first-order-candidate-2}. We find that, similar to Fig. \ref{fig:first-order-candidate}, AVCBF-R-P can ensure $\psi_{0}>0 $ (safety-feasibility criterion) through hyperparameter tuning, whereas AVCBF-R cannot. Therefore, AVCBF-R cannot guarantee feasibility and safety for the given set of hyperparameters. We conclude that AVCBFs with a reduced relative degree can ensure both safety and feasibility if the safety-feasibility criterion is satisfied. Additionally, the proposed parametrization method significantly reduces the conservatism of the control strategy.

\begin{figure*}[t]
    \vspace{3mm}
    \centering
    \begin{subfigure}[t]{0.26\linewidth}
        \centering
        \includegraphics[width=1\linewidth]{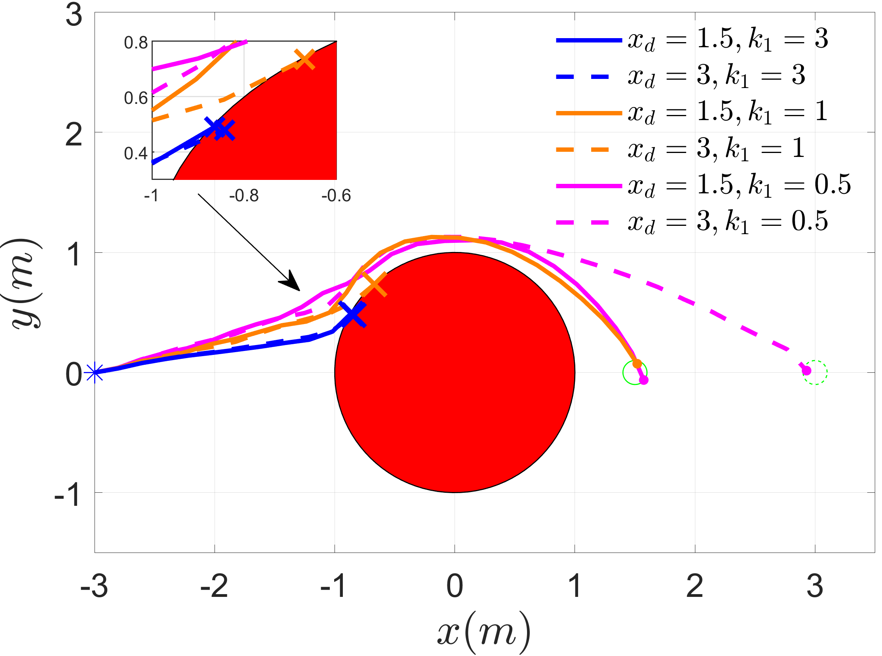}
        \caption{AVCBF-R: Closed-loop trajectories over time.}
        \label{fig:avcbf-reduced-trj}
    \end{subfigure}
    \begin{subfigure}[t]{0.26\linewidth}
        \centering
        \includegraphics[width=1\linewidth]{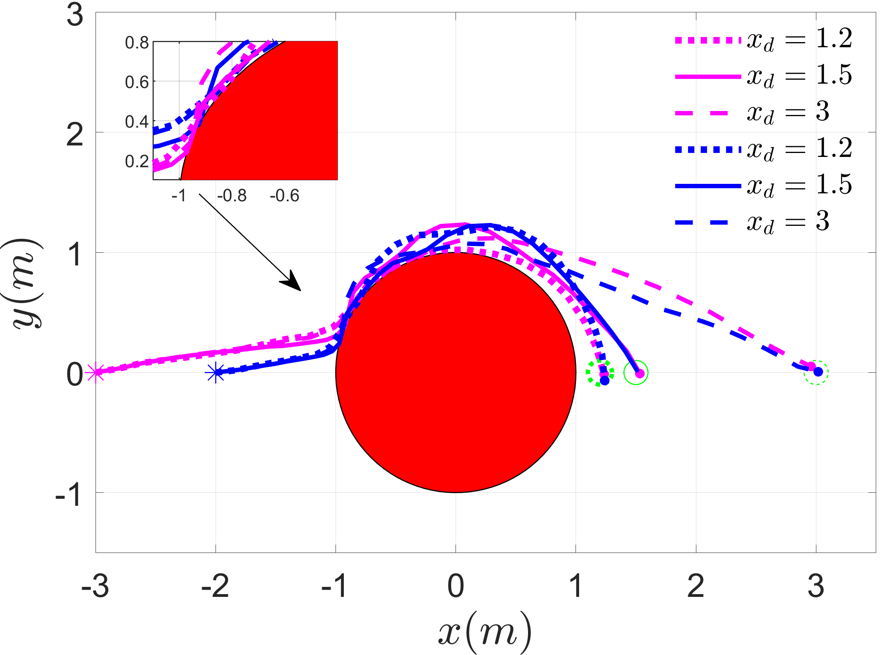}
        \caption{AVCBF-R-P: Closed-loop trajectories over time, $k_{1}=3$.}
        \label{fig:avcbf-reduced-P-trj}
    \end{subfigure}  
    \begin{subfigure}[t]{0.23\linewidth}
        \centering
        \includegraphics[width=1\linewidth]{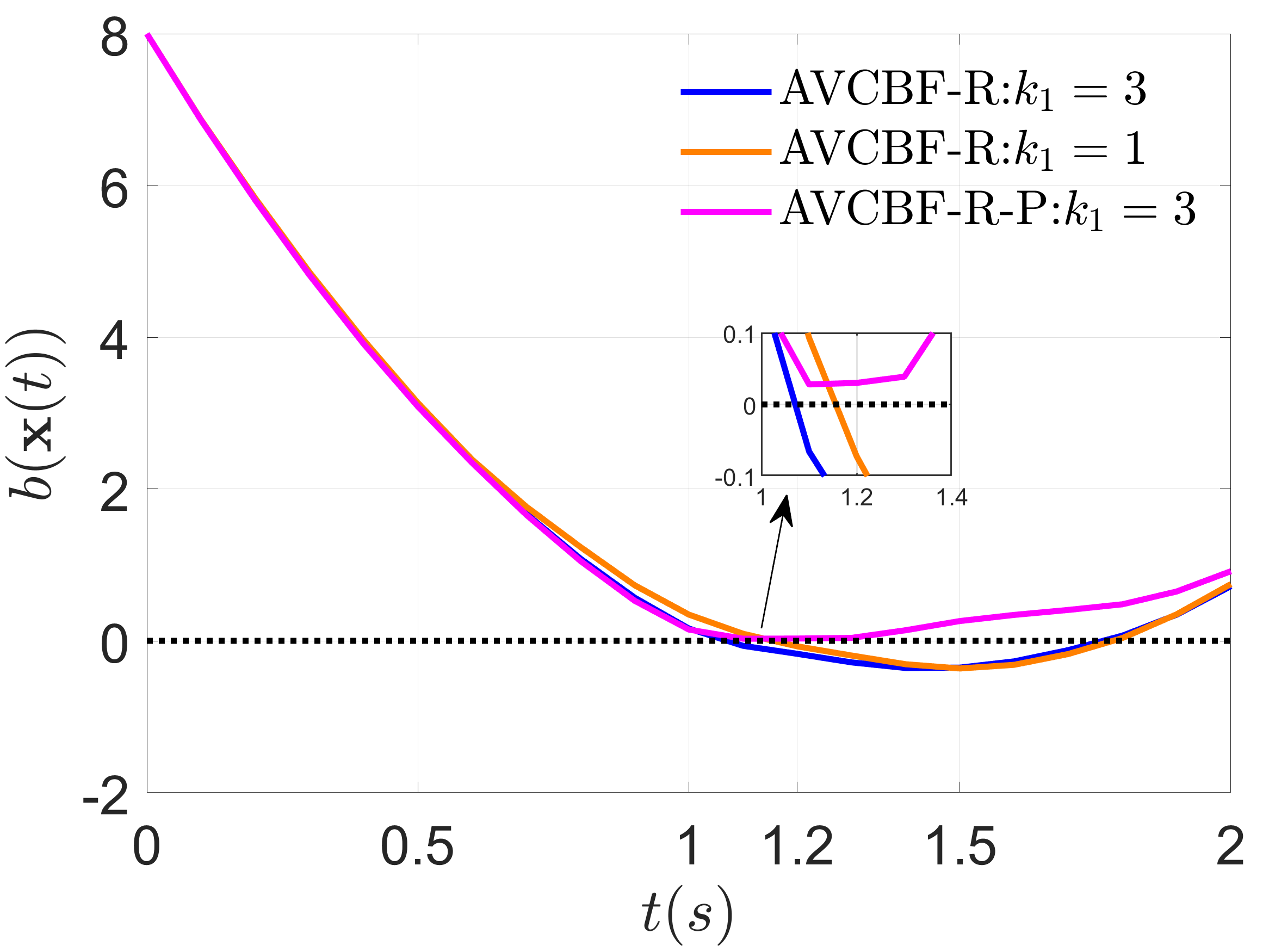}
        \caption{Safe distance over time.}
        \label{fig:cbf-candidate-2}
    \end{subfigure}
    \begin{subfigure}[t]{0.23\linewidth}
        \centering
        \includegraphics[width=1\linewidth]{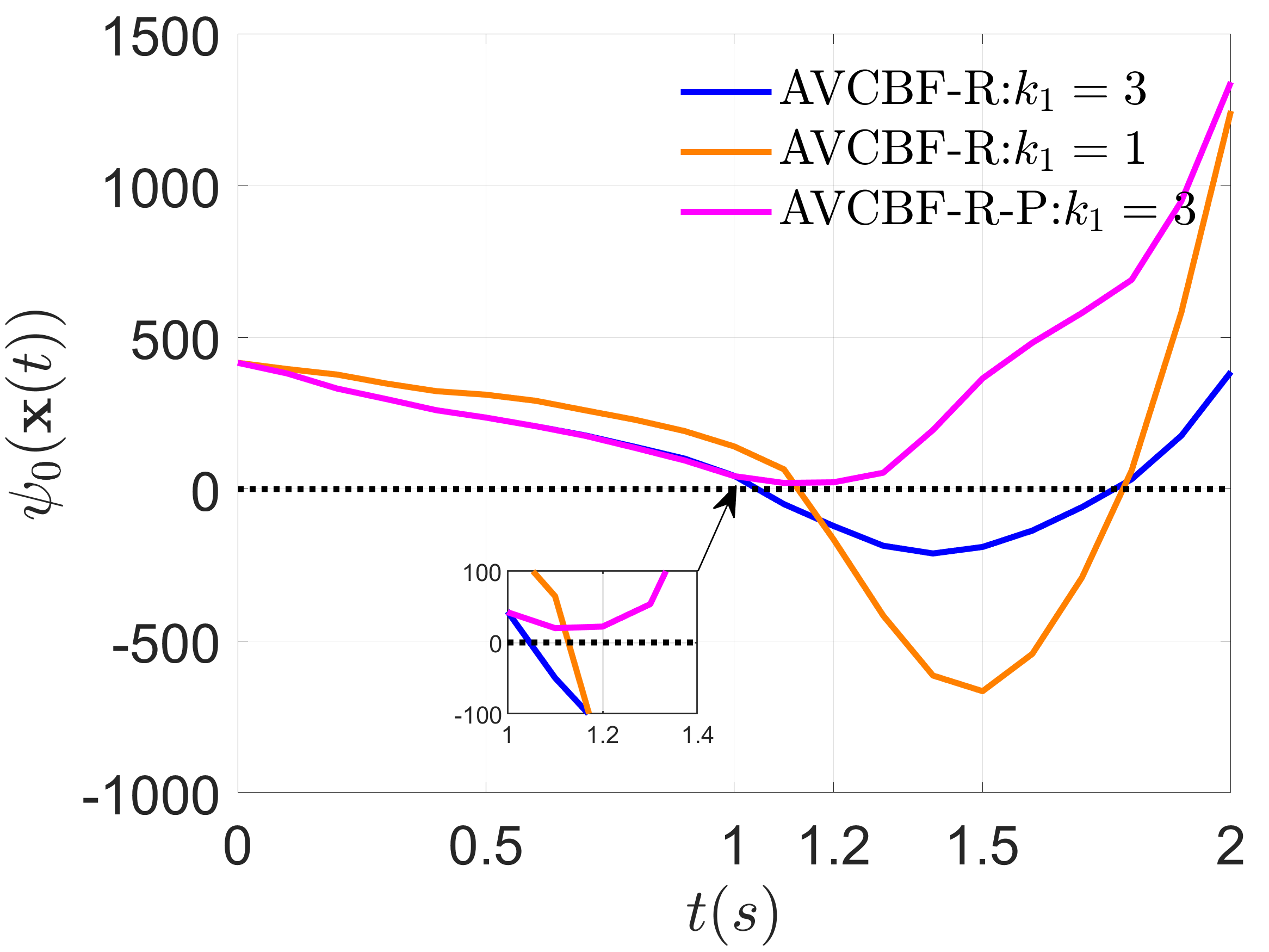}
        \caption{Safety-feasibility criterion check.}
        \label{fig:first-order-candidate-2}
    \end{subfigure}
    \caption{Closed-loop trajectories over time with AVCBF-R (Fig. \ref{fig:avcbf-reduced-trj}) and AVCBF-R-P (Fig. \ref{fig:avcbf-reduced-P-trj}) and analysis of safety and feasibility for two controllers (in Fig. \ref{fig:cbf-candidate-2}, Fig. \ref{fig:first-order-candidate-2}). The cross symbol indicates that the QP is infeasible at this time step. Different sets of hyperparameters are tested. 
    } 
\end{figure*}

\subsection{Obstacle Avoidance for Mixed Relative Degree Systems}
We consider a nonlinear vehicle dynamics in the form
\begin{small}
\begin{equation}
\label{eq:UM-dynamics3}
\underbrace{\begin{bmatrix}
\dot{x}(t) \\
\dot{y}(t) \\
\dot{\theta}(t)\\
 \dot{\phi} (t)\\
\dot{v}(t)
\end{bmatrix}}_{\dot{\boldsymbol{x}}(t)}  
=\underbrace{\begin{bmatrix}
 v(t)\cos{(\theta(t))}  \\
 v(t)\sin{(\theta(t))} \\
 \phi (t) \\
 0 \\
 0
\end{bmatrix}}_{f(\boldsymbol{x}(t))} 
+ \underbrace{\begin{bmatrix}
  0 & 0\\
  0 & 0\\
  0 & 0\\
  1 & 0\\
  0 & \frac{1}{M} 
\end{bmatrix}}_{g(\boldsymbol{x}(t))}\underbrace{\begin{bmatrix}
   u_{1}(t)   \\
  u_{2}(t) 
\end{bmatrix}}_{\boldsymbol{u}(t)}.
\end{equation}
\end{small}
In \eqref{eq:UM-dynamics3}, $M$ denotes the mass of the vehicle and $(x, y)$ denote the coordinates of the unicycle, $v$ is its linear speed, $\theta$ denotes the heading angle, $\phi$ denotes the rotation speed, and $\boldsymbol{u}$ represent the angular acceleration ($u_{1}$) and the driven force ($u_{2}$), respectively. Vehicle limitations are the same as those in Sec. \ref{subsubsec: vehicle limitations2}. 
To satisfy the constraint on rechability, we define a CLF $V(\boldsymbol{x}(t)) \coloneqq(0.1(\theta(t)-\theta_{d})+\phi(t))^{2}$ with $\theta_{d}=atan2(\frac{y_{d}-y(t)}{x_{d}-x(t)}), c_{1}=c_{2}=1$ to stabilize $\theta(t)$ to $\theta_{d}$ and $\phi(t)$ to 0. The relaxed constraint in \eqref{eq:clf} is then formulated the same as \eqref{eq:ACC-clf-2} as a soft constraint. The parameters are $v(0)=2m/s, \phi(0)=0.01 rad/s, M=1650kg, (x_{o}, y_{o})= (0, 0), r_{o}=1m, r_{d}=0.2m, \bigtriangleup t=0.01s, u_{1,\text{max}}=-u_{1,\text{min}}=5rad/s^{2}, u_{2,\text{max}}=-u_{2,\text{min}}=8250N, T=5s$.

We use a continuous function $b(\boldsymbol{x}(t))=(x-x_{o})^{2}+(y-y_{o})^{2}-r_{o}^{2}$ to measure safe distance. We can see that the relative
degree of $b(\boldsymbol{x}(t))$ with respect to $u_{1}$ is 3, and the relative degree
of $b(\boldsymbol{x}(t))$ with respect to $u_{2}$ is 2. Therefore, the minimum relative degree $\underline{m}$ of $b(\boldsymbol{x}(t))$ with respect to \eqref{eq:UM-dynamics3} is 2. If we implement $b(\boldsymbol{x}(t))\ge0 $ with a standard HOCBF, only the control input $u_{2}$ shows up in the HOCBF constraint and the vehicle cannot use angular acceleration to avoid the obstacle. In order to make all control inputs show up in one constraint, based on Rem. \ref{rem:reduced degree}, we can design auxiliary function as $\mathcal{A}_{1}(\boldsymbol{x},a_{1}(t))=a_{1}(t)+v(t)+\phi(t)$. The auxiliary dynamics is the same as \eqref{eq:Auxiliary-dynamics2}. The HOCBFs for $\mathcal{A}_{1}$ are defined as 
\begin{equation}
\label{eq:SHOCBF-sequence-ACC-5}
\begin{split}
&\varphi_{1,0}(\boldsymbol{x},\boldsymbol{{\pi}}_{1})\coloneqq a_{1}+v+\phi,\\
&\varphi_{1,1}(\boldsymbol{x},\boldsymbol{{\pi}}_{1},\boldsymbol{u},\nu_{1})\coloneqq \dot{\varphi}_{1,0}(\boldsymbol{x},\boldsymbol{{\pi}}_{1},\boldsymbol{u},\nu_{1})+l_{1,1}\varphi_{1,0}(\boldsymbol{x},\boldsymbol{{\pi}}_{1}),\\
\end{split}
\end{equation}
where $\alpha_{1,1}(\cdot)$ is defined as a linear function. The AVCBFs are then defined as
\begin{equation}
\label{eq:AVBCBF-sequence-ACC-5}
\begin{split}
&\psi_{0}(\boldsymbol{x},\boldsymbol{{\pi}}_{1})\coloneqq (a_{1}+v+\phi)b(\boldsymbol{x}),\\
&\psi_{1}(\boldsymbol{x}, \boldsymbol{{\pi}}_{1},\boldsymbol{u},\nu_{1})\coloneqq \dot{\psi}_{0}(\boldsymbol{x}, \boldsymbol{{\pi}}_{1},\boldsymbol{u},\nu_{1})+k_{1}\psi_{0}(\boldsymbol{x},\boldsymbol{{\pi}}_{1}),
\end{split}
\end{equation}
where $\alpha_{1}(\cdot)$ is set as a linear function. By formulating constraints from HOCBFs \eqref{eq:SHOCBF-sequence-ACC-5}, AVCBFs \eqref{eq:AVBCBF-sequence-ACC-5}, CLF \eqref{eq:ACC-clf-2} and control input \eqref{eq:constraint-u-2}, we can define the cost function for QP in the same form as \eqref{eq:optimal control-cost new}. We refer to this method as AVCBF-M (mixed relative degree). Other parameters are set as $a_{1}(0)=0.1, c_{3}=10, W_{1}=1,Q=10^{3}, \epsilon_{1}=10^{-10}, k_{1}=l_{1,1}=0.1, a_{1,w}(t)=0$. 

In the above subfigure of Fig. \ref{fig:AVBCBFs-mixed-order}, we can see that AVCBF-M enables the vehicle to reach the green target area from different initial locations and heading angles. If we select the solid blue curve and present its control inputs over time in the below subfigure of Fig. \ref{fig:AVBCBFs-mixed-order}, we observe that due to the design of the auxiliary function, $u_{1}$ appears in the AVCBF-M constraint and remains within the input bounds while influencing the vehicle’s obstacle avoidance behavior. We conclude that auxiliary functions can be specifically designed for AVCBFs to ensure both safety and feasibility when the system's control inputs have mixed relative degrees.
\begin{figure}[ht]
    \centering
    \includegraphics[scale=0.28]{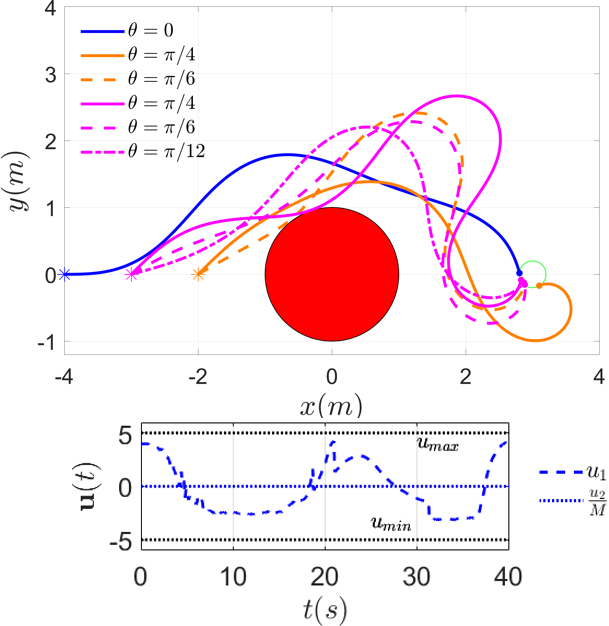}
    \caption{Closed-loop trajectories with AVCBF-M (above): multiple safe trajectories from different initial positions (asterisks) converge to the target area (green solid circle). Initial locations: blue $(-4,0)$, magenta $(-3,0)$, orange $(-2,0)$; target: $(x_d, y_d) = (3,0)$. Control inputs $\boldsymbol{u}(t)$ (below) vary over time for the trajectory starting at $(-4,0)$ with $\theta(0) = 0 rad.$}
    \label{fig:AVBCBFs-mixed-order}
\end{figure} 
% \begin{figure}[ht]
%     \centering
%     \includegraphics[scale=0.23]{figures/AVCBF2-CBFstatesinputs-mixed-new.png}
%     \caption{States $\boldsymbol{x}(t)$ and control inputs $\boldsymbol{u}(t)$ vary over time with AVCBF-M. $x(0)=-4m, \theta(0)=0 rad,x_{d}=3m$. }
%     \label{fig:AVBCBFs-mixed-input}
% \end{figure} 
\section{Conclusion}
\label{sec:Conclusion}
In this paper, we introduced Auxiliary-Variable Adaptive Control Barrier Functions (AVCBFs) to address the feasibility and safety challenges associated with standard CBF methods. We demonstrated that AVCBFs effectively mitigate infeasibility issues caused by mixed relative degrees, input nullification, and tight or time-varying control bounds. By incorporating auxiliary functions, our method ensures that all control inputs explicitly appear in the desired-order safety constraint, thereby improving feasibility while maintaining safety guarantees.

Additionally, we proposed an automatic hyperparameter tuning method, which dynamically adjusts AVCBF hyperparameters to enhance feasibility and reduce conservatism, eliminating the need for excessive manual hyperparameter tuning.

We demonstrate the effectiveness of AVCBFs and compare them with the state-of-the-art adaptive CBF method across a series of control problems. Future work will integrate AVCBFs into differentiable quadratic programs via machine learning to eliminate auxiliary system design, enhancing efficiency and performance.

\bibliographystyle{IEEEtran}
\balance
\bibliography{references.bib}

\vspace{10pt}

\setlength{\intextsep}{2pt}
\begin{wrapfigure}{l}{25mm} 
 \includegraphics[width=1in,height=1.25in,clip,keepaspectratio]{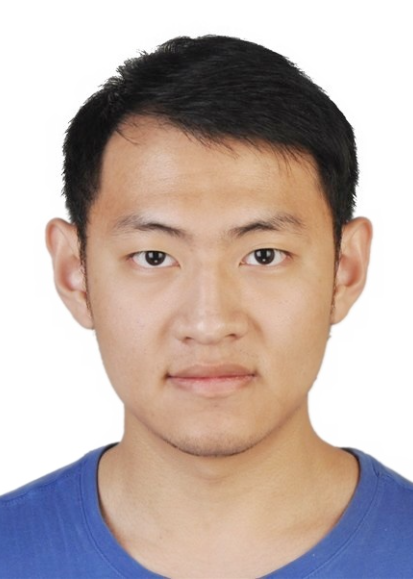}
\end{wrapfigure}\par
{\small \textbf{Shuo Liu} (Student Member, IEEE) received his M.S. degree in Mechanical Engineering from Columbia University, New York, NY, USA, in 2020 and B.Eng. degree in Mechanical Engineering from Chongqing University, Chongqing, China, in 2018. He is currently a Ph.D. candidate in Mechanical Engineering at Boston University, Boston, USA and his research interests include optimization, nonlinear control, deep learning and robotics. \par}

\vspace{10pt}

\setlength{\intextsep}{2pt}
\begin{wrapfigure}{l}{25mm} 
 \includegraphics[width=1in,height=1.25in,clip,keepaspectratio]{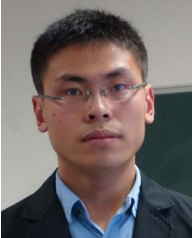}
\end{wrapfigure}\par
{\small \textbf{Wei Xiao} (Member, IEEE) received the B.Sc. degree
in mechanical engineering and automation from the University of Science and Technology Beijing, Beijing, China, the M.Sc. degree in robotics from the Chinese Academy of Sciences (Institute of Automation), Beijing, China, and the Ph.D. degree in systems
engineering from Boston University, Boston, MA, USA, in 2013, 2016, and 2021, respectively.

He is currently a Postdoctoral Associate with the Massachusetts Institute of Technology, Cambridge, MA, USA. His current research interests include control theory and machine learning, with particular emphasis on robotics and traffic control. \par}

\vspace{10pt}

\setlength{\intextsep}{2pt}
\begin{wrapfigure}{l}{25mm} 
 \includegraphics[width=1in,height=1.25in,clip,keepaspectratio]{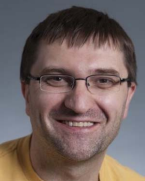}
\end{wrapfigure}\par
{\small \textbf{Calin Belta} (Fellow, IEEE) received the B.Sc.
and M.Sc. degrees from the Technical University of Iasi, Iasi, Romania, in 1995 and 1997, respectively, and the M.Sc. and Ph.D. degrees from the University of Pennsylvania, Philadelphia, PA, USA, in 2001 and 2003. 

He is currently the Brendan Iribe Endowed Professor of Electrical and Computer Engineering and Computer Science at the University of Maryland, College Park, MD, USA. His research focuses on control theory and formal methods, with particular emphasis on hybrid and cyber-physical systems, synthesis and verification, and applications in robotics and biology.

Dr. Belta was a recipient of the Air Force Office of Scientific Research
Young Investigator Award and the National Science Foundation CAREER Award.\par}
\end{document}